\pdfoutput=1

\documentclass[%
    reprint,
    showpacs,preprintnumbers,
    showkeys,
    amsmath,amssymb,
    aps,
    prb,
    floatfix,
    longbibliography,
]{revtex4-2}

\usepackage[english]{babel}
\usepackage[utf8]{inputenc}
\usepackage{graphicx}
\usepackage{bm}
\usepackage{hyperref}
\usepackage{xcolor}
\usepackage{epsfig}
\usepackage{amsmath}
\usepackage{amsfonts}
\usepackage{amssymb}
\usepackage{amsthm}
\usepackage{tikz}
\usepackage{tikz-cd}
\usepackage{dsfont}
\usepackage[all]{nowidow}
\usepackage{pifont}
\usepackage{algorithm}
\usepackage{algpseudocode}
\usepackage[capitalise]{cleveref}
\usepackage[export]{adjustbox}
\usepackage{xspace}
\usepackage{nicefrac}
\let\doi\undefined
\usepackage{doi}
\usepackage{lipsum}

\makeatletter
\AddToHook{cmd/appendix/before}{%
    \crefalias{section}{appendix}%
    \crefalias{subsection}{appendix}
    \crefalias{subsubsection}{appendix}
}
\makeatother

\usepackage[letterspace=130]{microtype}

\definecolor{refColor}{HTML}{0376E9}
\definecolor{figColor}{HTML}{E90303}
\definecolor{urlColor}{HTML}{0bcb9a}

\hypersetup{
    unicode=false,                 
    pdftoolbar=true,               
    pdfmenubar=true,               
    pdffitwindow=false,            
    pdfstartview={FitH},           
    pdftitle={},                   
    pdfauthor={Nicolai Lang},      
    pdfsubject={Quantum physics},  
    pdfcreator={Nicolai Lang},     
    pdfproducer={Nicolai Lang},    
    pdfkeywords={},                
    pdfnewwindow=true,             
    colorlinks=true,               
    linkcolor=figColor,            
    citecolor=refColor,            
    filecolor=magenta,             
    urlcolor=urlColor              
}

\newcommand{\bra}[1]{\mathinner{\langle{#1}|}}
\newcommand{\ket}[1]{\mathinner{|{#1}\rangle}}

\newcommand{\BraKet}[2]{\langle #1 | #2 \rangle}

\renewcommand{\vec}[1]{\mathbf{#1}}

\newcommand{\bigslant}[2]{{\raisebox{.2em}{$#1$}\left/\raisebox{-.2em}{$#2$}\right.}}
\renewcommand{\vec}[1]{\boldsymbol{#1}}

\newcommand{\Tr}[2]{\operatorname{Tr}_{#2}\left[#1\right]}
\newcommand{\E}{\mathbb{E}}

\newcommand{\com}[2]{\left[#1,#2\right]}

\renewcommand{\ol}{\overline}

\newcommand{\const}{\textrm{const}}
\renewcommand{\d}{\mathrm{d}}
\renewcommand{\O}{\mathcal{O}}
\newcommand{\A}{\mathcal{A}}
\renewcommand{\H}{\mathcal{H}}

\newcommand{\NP}{{\normalfont\textsf{NP}}}

\newcommand{\sub}[1]{{\text{\tiny\textnormal{#1}}}}
\newcommand{\aut}[1]{\operatorname{Aut}\left(#1\right)}
\newcommand{\Aut}[1]{\operatorname{Aut}\left(#1\right)}
\newcommand{\Rb}{{r_\sub{B}}}

\newcommand{\C}{\mathcal{C}}
\newcommand{\id}{\mathds{1}}
\newcommand{\PQDM}{\texttt{ZQDM}}
\newcommand{\VER}{\texttt{VER}}
\newcommand{\del}{\partial}
\newcommand{\R}{\mathbb{R}}
\newcommand{\Z}{\mathbb{Z}}
\newcommand{\ZZ}{{\mathbb{Z}_2}}
\newcommand{\dg}{\textup{d}}
\newcommand{\F}{\mathcal{F}}
\newcommand{\G}{\mathcal{G}}
\newcommand{\etal}{\emph{et\,al.}\xspace}
\newcommand{\Wlog}{\emph{w.l.o.g.}\xspace}
\newcommand{\Eig}{\operatorname{Eig}}

\newtheorem{theorem}{Theorem}
\newtheorem{definition}{Definition}
\newtheorem{lemma}{Lemma}
\newtheorem{proposition}{Proposition}
\newtheorem{corollary}{Corollary}

\definecolor{myred}{HTML}{df1b1b}
\definecolor{myblue}{HTML}{0f34c2}
\newcommand{\CR}[1]{{\color{myred}#1}}
\newcommand{\CB}[1]{{\color{myblue}#1}}

\graphicspath{{fig/}}

\newcommand{\rydberg}{blockade\xspace}

\begin{document}

\title{Topological order in symmetric blockade structures}

\author{Tobias F.\ Maier}
\author{Hans Peter Büchler}
\author{Nicolai Lang}
\email{nicolai.lang@itp3.uni-stuttgart.de}
\affiliation{%
    Institute for Theoretical Physics III 
    and Center for Integrated Quantum Science and Technology,\\
    University of Stuttgart, 70550 Stuttgart, Germany
}

\date{\today}


\begin{abstract}
    The bottom-up design of strongly interacting quantum materials with
    prescribed ground state properties is a highly nontrivial task, especially
    if only simple constituents with realistic two-body interactions are
    available on the microscopic level. Here we study two- and
    three-dimensional structures of two-level systems that interact via a
    simple blockade potential in the presence of a coherent coupling between
    the two states. For such strongly interacting quantum many-body systems, we
    introduce the concept of \emph{blockade graph automorphisms} to construct
    symmetric blockade structures with strong quantum fluctuations that lead to
    equal-weight superpositions of tailored states. Drawing from these results,
    we design a quasi-two-dimensional periodic quantum system that -- as we
    show rigorously -- features a topological $\mathbb{Z}_2$ spin liquid as its
    ground state. Our construction is based on the implementation of a local
    symmetry on the microscopic level in a system with only two-body
    interactions.
\end{abstract}



\maketitle

\section{Introduction}

Condensed matter physics is concerned with the explanation and prediction of
emergent phenomena on large scales from systems of many interacting,
microscopic degrees of freedom, using a sophisticated arsenal of experimental,
theoretical, and numerical techniques. Motivated by tremendous progress in the
preparation and control of many quantum degrees of freedom on the atomic
scale~\cite{Bloch2008,Browaeys_2020}, the ``inverse problem'' has recently come
into focus (\cref{fig:rationale}):
Given a ``toolbox'' of microscopic constituents (like atoms) that can be
controlled precisely and correlated via simple, tunable interactions, is it
possible -- and if so, how -- to robustly engineer a \emph{prescribed} quantum
many-body phase?
A better understanding of this bottom-up design of quantum materials opens an
alternative route to study emergent phenomena in quantum many-body physics.
This is emphasized by the pivotal role played by exactly solvable, though often
experimentally unrealistic models such has the toric code~\cite{Kitaev2003},
more general string net models~\cite{Levin2005}, and other parent
Hamiltonians~\cite{Fannes_1992, Schuch_2010} like the famous AKLT
construction~\cite{Affleck1988} or resonating valence bond
states~\cite{Anderson1973,Schuch2012}. The ``inverse problem'' is distinguished
(and complicated) by its restriction to a specific type of simple,
experimentally accessible interaction on the microscopic level that can be used
for the construction of a prescribed phase of matter. Solving this problem is
of particular interest for quantum phases with topological order~\cite{Wen2017}
due to potential applications as quantum memories~\cite{Dennis2002} and for
topological quantum computing~\cite{Nayak2008}. 

This paper is a contribution to tackle the ``inverse problem'' in a particular
setting, with focus on the construction of topologically ordered states of
matter. We consider a toolbox of elementary two-level systems that interact via
a simple blockade potential; such systems can be viewed as spin systems with a
strong Ising-type interaction in the presence of transversal and longitudinal
magnetic fields. We extend this framework by a versatile concept of symmetry,
and leverage this novel tool for the construction of a topological
$\mathbb{Z}_2$ spin liquid. To comply with the rationale of engineered quantum
matter, none of these results rely on numerical techniques or perturbative
arguments.

\begin{figure}[tb]
    \centering
    \includegraphics[width=0.95\linewidth]{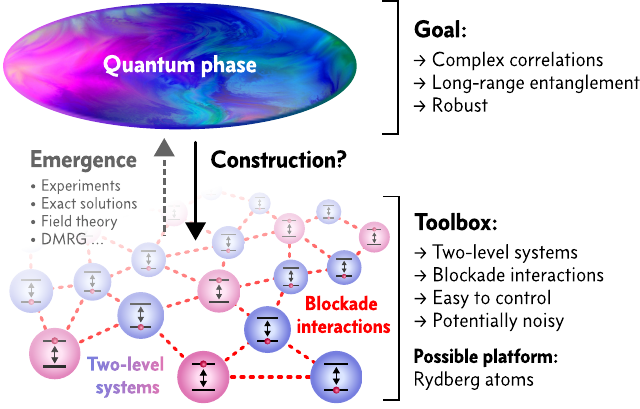}
    \caption{%
        \emph{Rationale.} The ``inverse problem'' of condensed matter physics
        studies the construction of prescribed, robust quantum phases from a
        given set of simple, microscopic constituents. Here we consider a
        toolbox motivated by (but not specific to) the Rydberg platform:
        two-level systems can be placed freely in two and three dimensions and
        interact via a simple blockade mechanism. Our contribution is an
        extension of the theoretical foundations of this toolbox to construct
        interesting quantum phases.
    }
    \label{fig:rationale}
\end{figure}

Our focus on this particular toolbox is motivated by the significant progress
that neutral atom platforms have seen in recent years~\cite{Browaeys_2020},
both on the experimental and the theoretical front: Control on the single-atom
level is facilitated by optical lattices and tweezers that allow for the
design of arbitrary spatial structures in two and three dimensions comprising
hundreds of atoms~\cite{Barredo2018, Bluvstein_2023} (though leveraging
three-dimensional structures for quantum simulation and processing remains
technically challenging). The internal states of atoms can be controlled by
lasers, microwaves, and electromagnetic fields with high fidelity, both
uniformly and with single-site addressability~\cite{Labuhn_2014,Omran2019},
while dipole-dipole and/or van der Waals interactions allow for controllable
interactions~\cite{Sibalic2018}. Van der Waals interactions give rise to the
famous \emph{Rydberg blockade mechanism}, where a single excited atom prevents
further excitations within a tunable blockade
radius~\cite{Jaksch2000,Tong2004,Singer2004,Gaetan2009,Urban2009}. The
simplicity of this strong interaction, combined with the high level of spatial
control over hundreds of atoms, makes the Rydberg platform a prime candidate
for quantum
simulation~\cite{Weimer2010,Georgescu2014,Schauss2015,Labuhn2016,Gross2017,Bernien_2017,Altman2021,Semeghini2021,Scholl2021}
and the design of artificial quantum matter. Recently, various two-dimensional
structures of Rydberg atoms were put forward and studied numerically, with the
goal to artificially design interesting quantum phases such as topological spin
liquids~\cite{Samajdar_2021,Verresen2021, Tarabunga_2022,  Slagle_2022,
Maity_2024, Wang_2025}, dimer models~\cite{Yan_2022, Zeybek_2023, Zeng2025},
fractonic phases~\cite{Verresen_2021b, Myerson_Jain_2022, Mac_do_2024}, lattice
gauge
theories~\cite{Surace_2020,Celi2020,Homeier_2023,Samajdar_2023,Feldmeier_2024,
Koeylueoglu2024,Cheng_2024,Shah_2025}, and glassy phases~\cite{Yan_2023}; these
studies were complemented by experimental results for systems of several
hundred atoms~\cite{Semeghini2021, Ebadi_2021, Manovitz_2025}.  However, the
experimental preparation of the true many-body ground state of such strongly
interacting systems has proven challenging, and (quasi)adiabatic preparation
schemes often fail to do so due to small gaps in some sectors of the
spectrum~\cite{Giudici2022,Sahay2022,Wang_2025}.
The same methods used for the design of artificial quantum matter can also be
applied to solve optimization problems by encoding them into the geometry of
tailored blockade structures, a procedure called ``geometric
programming''~\cite{Pichler2018b, Brady_2023, Vercellino_2023, Dalyac2024,
Bombieri2024, Naghmouchi_2024, Lanthaler2024, Farouk_2024, Dupont_2024,
Schiffer_2024, Byun2024, Leclerc_2024, Schuetz_2024}; experimental
implementations of small problem instances have been reported as
well~\cite{Byun2022, Kim2022, Ebadi2022, Dalyac2023, Jeong_2023,
Park2024, Cazals2025,de_Oliveira_2025}.

While the neutral atom platform certainly is the most advanced in regard to
theory development and experimental sophistication, there are other platforms
that allow for tailored blockade interactions in the quantum regime, like
superconducting qubit architectures~\cite{Menta_2025} and Rydberg excitons in
solid state systems~\cite{Hecktter_2021}.

The fact that the same type of microscopic interactions can be realized on
different experimental platforms suggests the development of a ``platform
agnostic'' toolbox for the design of artificial quantum matter. This toolbox
provides abstract two-level systems that interact via a simple blockade
potential; the pattern and layout of these interactions is then a central part of
the design process for a particular quantum phase. Recently, first steps
towards such a toolbox have been reported~\cite{Wurtz2022, Nguyen2022,
Stastny2023a, Lanthaler2023}. While these results explain how the frustration
of generic blockade structures can be leveraged to prepare nontrivial,
\emph{classical} ground state manifolds, a systematic understanding of the
effect of \emph{quantum fluctuations} on these strongly interacting systems is
lacking.

In this paper, we take a first step to incorporate quantum fluctuations
systematically into the design process of blockade structures to engineer
artificial quantum materials. While our focus on this particular toolbox is
motivated by the Rydberg platform, all our results are platform agnostic. Our
main contribution is the design of a quasi-two-dimensional periodic array of
two-level systems that -- by construction -- stabilizes an equal-weight
superposition of loop states when exposed to uniform quantum fluctuations. This
feature makes the ground state of this strongly interacting quantum many-body
system topologically ordered (in the universality class of the toric code),
while using only physically realistic, local, two-body interactions. It is
remarkable, and in line with the rationale of engineered quantum matter
advertised above, that these results can be established rigorously, without
relying on numerics -- and despite the fact that the constructed model is not a
renormalization fixed point (as compared to string net models~\cite{Levin2005}
like the toric code Hamiltonian~\cite{Kitaev2003}).

To achieve this feat, we extend the toolbox of \rydberg structures by the
concept of \emph{blockade graph automorphisms}, a -- not necessarily geometric
-- symmetry of blockade interactions. We show that particularly
symmetric (dubbed \emph{fully-symmetric}) \rydberg structures stabilize
equal-weight superpositions of their classical ground states when subject to
uniform quantum fluctuations. In combination with the known blueprints for the
design of \rydberg structures with prescribed classical ground state manifolds
(studied previously in Refs.~\cite{Wurtz2022, Nguyen2022, Stastny2023a,
Lanthaler2023}), this provides a powerful tool for the design of strongly
fluctuating quantum systems with nontrivial correlations and entanglement
patterns.

The remainder of this paper is structured as follows. In \cref{sec:setting}, we
introduce the systems under consideration, define the relevant notation, and
specify our goal. To make this paper self-contained, we provide in
\cref{sec:review} a brief review of prior results on the construction of
blockade structures. The discussion of new results starts in
\cref{sec:motivation} with two motivating examples that lead to the concept of
fully-symmetric structures defined in \cref{sec:sym}. In \cref{sec:fsu}, we
introduce and discuss a crucial example for this concept, the fully-symmetric
universal gate, which is the basis for the construction detailed in the second
half of the paper: First, in \cref{sec:firstview}, we discuss general features
and constraints that periodic (tessellated) fully-symmetric \rydberg structures
must satisfy. Our main result is then detailed in \cref{sec:local_aut}. We
start in \cref{subsec:toric} with a brief review of the toric code, followed by
a discussion of fundamental problems in \cref{subsec:nogo} that must be
overcome to realize its topological order with a fully-symmetric \rydberg
structure. These insights are used in \cref{sec:const} to engineer a periodic,
quasi-two-dimensional system that stabilizes an equal-weight superposition of
loop states by construction.  We study the symmetries responsible for this
feature in \cref{sec:localz2}, and comment on the geometric embedding of the
model in \cref{sec:embedding}. We complete our discussion of this strongly
interacting quantum many-body system in \cref{subsec:gs} with a study of its
many-body spectrum (\cref{subsubsec:general}), a rigorous proof of its ground
state topological order (\cref{subsubsec:topo}), and comments on its bulk gap
(\cref{subsubsec:gap}). We close in \cref{sec:outlook} with comments on open
questions and conclude in \cref{sec:summary}.

\begin{figure*}[tb]
    \centering
    \includegraphics[width=1.0\linewidth]{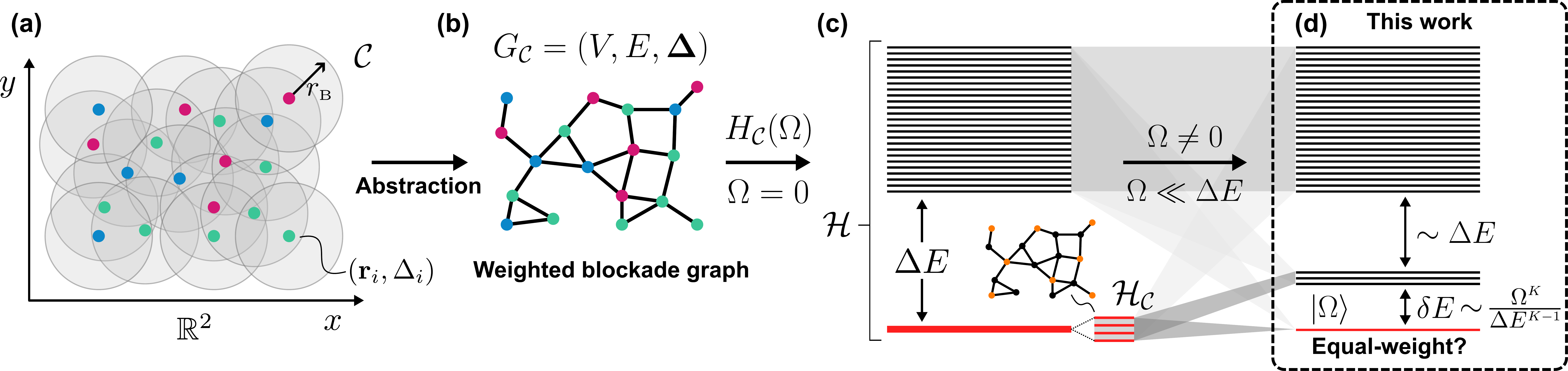}
    \caption{%
	\emph{Setting \& Objective.} 
    (a)~A two-dimensional arrangement $\C=(\vec r_i,\Delta_i)_{i\in V}$ of
    atoms $i\in V$ with position $\vec r_i$ and detuning $\Delta_i$ is governed
    by the Hamiltonian $H_\C$ that describes a blockade interaction with
    radius $\Rb$. The detunings are encoded by colors.
    (b)~Interactions and detunings are conveniently represented by a
    vertex-weighted blockade graph $G_\C$.
    (c)~For $\Omega=0$, the Hamiltonian gives rise to a degenerate low-energy
    eigenspace $\H_\C<\H$, separated from the excited states by a gap $\Delta
    E$. The ground states correspond to maximum-weight independent sets of the
    blockade graph.
    (d)~Our objective is to design blockade structures in which quantum
    fluctuations $\Omega\neq 0$ stabilize a ground state that is an
    equal-weight superposition of the states in~$\H_\C$.
    }
    \label{fig:setting}
\end{figure*}

\section{Setting \& Objective}
\label{sec:setting}

We start with a characterization of the systems we consider in this paper,
their formal description, and introduce the required mathematical notation.
With these foundations, we can formulate our objective precisely.

In this work, we focus on two- and three-dimensional arrangements of two-level
systems that interact via a blockade potential 
\begin{align}
    U(r):=
    \begin{cases}
        0&\text{for}\;r>\Rb\\
        \infty&\text{for}\;r\leq\Rb
    \end{cases}
    \label{eq:U}
\end{align}
with \emph{blockade radius} $\Rb$; the latter defines the only length scale of
the system and can be set to unity. The two-level systems can be realized by
electronic states of atoms, internal spin states, cavity modes, etc. For
simplicity, we refer to these systems as ``atoms'' throughout the paper,
although they are not necessarily realized as such. Every atom is assigned an
index $i\in V=\{1\dots N\}$, placed at position $\vec r_i\in\mathbb{R}^{3}$,
and described by the state $\ket{n}_i$ with $n\in\{0,1\}$, where we identify
$\ket{0}_i$ as the \emph{ground state} and $\ket{1}_i$ as the \emph{excited state}.

The quantum dynamics of atoms is achieved by coupling the ground state to the
excited state with strength $\Omega_i$, which we refer to as \emph{(quantum)
fluctuations}. In addition, each atom can be detuned by $\Delta_{i}\geq 0$.
The Hamiltonian that describes such systems is then
\begin{align}
    H_\C(\Omega)=
    \sum_{i<j}\,U(r_{ij})\,n_in_j
    +\sum_i\,\left(\Omega_i\,\sigma_i^x-\Delta_i\,n_i\right)
    \label{eq:H}
\end{align}
with $r_{ij}:=|\vec r_i-\vec r_j|$ the distance between atoms $i$ and $j$. This
Hamiltonian acts on the full Hilbert space $\H=(\mathbb{C}^2)^{\otimes N}$ with
representations $n_i=\ket{1}\bra{1}_i$ and
$\sigma_i^x=\ket{0}\bra{1}_i+\ket{1}\bra{0}_i$. Note that detunings $\Delta_i$
and fluctuations $\Omega_i$ can in principle be site-dependent; in the
following, we assume uniform fluctuations $\Omega_i\equiv\Omega$ (for
simplicity), but allow for non-uniform detunings $\Delta_i$. Throughout this
paper, we make use of only a few distinct values
$\Delta_i\in\{1\Delta,2\Delta,\ldots\}$, specified as integer multiples of a
unit detuning $\Delta$, in order to prepare degenerate ground state manifolds
for $\Omega=0$. The Hamiltonian~\eqref{eq:H}, together with the blockade
interaction~\eqref{eq:U}, is closely related to so called \emph{PXP models}
studied on the Rydberg platform~\cite{Lesanovsky2011,Verresen2021}.

Many-body quantum systems of this form are completely specified by the data
$\C\equiv(\vec r_i,\Delta_i)_{i\in V}$, which we refer to as a \emph{(blockade)
structure}, \cref{fig:setting}~(a), together with the uniform quantum
fluctuations $\Omega$. We refer to structures with $\Omega=0$ as
\emph{classical}, for then all operators in~\eqref{eq:H} commute and
eigenstates are given by products of the $\{\ket{0},\ket{1}\}$-basis. The
purpose of this paper is to deepen our understanding of the effects of quantum
fluctuations $\Omega\neq 0$ on such classical blockade structures.

Because of \cref{eq:U}, the interactions of blockade structures translate to
kinematic constraints that can be encoded by a \emph{vertex-weighted blockade
graph} $G_\C=(V,E,\bm\Delta)$, where an edge $e=(i,j)\in E$ between atoms
$i,j\in V$ indicates that they are in blockade, i.e., their distance $r_{ij}$
is smaller than the blockade radius $\Rb$.  In this case, the only accessible
states of the two atoms are $\ket{0}_i\ket{0}_j$, $\ket{0}_i\ket{1}_j$, and
$\ket{1}_i\ket{0}_j$; the doubly-excited state $\ket{1}_i\ket{1}_j$ gains
infinite energy from the interaction term in \eqref{eq:H} and is forbidden. The
detunings $\bm\Delta \equiv \{\Delta_i\}_{i\in V}$ are interpreted as weights of the
vertices and determine the energy of the accessible states.  An abstract graph
that can be realized in this way is called a \emph{unit disk (ball) graph} in
two (three) dimensions. Conversely, a layout of vertices that realizes a
prescribed graph as its blockade graph is a \emph{unit disk (ball) embedding}
of this graph (where the ``unit'' is the blockade radius $\Rb$).  Throughout
the paper, blockade graphs are drawn by solid edges connecting atoms that are
in blockade, while the colors of vertices encode their detuning,
\cref{fig:setting}~(b).

Provided a classical blockade structure $\C$ ($\Omega=0$), its ground state
manifold $\H_\C\equiv\operatorname{span}\{\ket{\vec n}\,|\,\vec n\in L_\C\}$ is
spanned by perfectly degenerate product states $\ket{\vec n}=\ket{n_1\dots
n_N}$ that are characterized by a subset of excitation patterns
$L_\C\subseteq\mathbb{Z}_2^N$. This manifold is separated by a gap $\Delta E$
from the excited states $\ket{\vec n}$ with patterns $\vec n\notin L_\C$
(typically $\Delta E\sim\min_i\Delta_i$), \cref{fig:setting}~(c).  Note that
these excited patterns still satisfy all blockade constraints.  Because of our
choice of detunings $\Delta_i\geq 0$, these states are typically characterized
by \emph{fewer} excited atoms, with the state $\ket{\vec 0}$ (all atoms in
their ground state) being the state with the \emph{highest} energy.  The ground
state patterns in $L_\C$ are the \emph{maximum-weight independent sets} (MWISs)
of the vertex-weighted blockade graph $G_\C$. An independent set of a graph is
a subset of vertices such that no two vertices are connected by an edge (this
accounts for the blockade). A \emph{maximum-weight} independent set is an
independent set that maximizes the total weight of vertices [which minimizes
the energy contribution of the $-\Delta_i n_i$ terms in~\eqref{eq:H}]. Clearly,
MWISs correspond to the ground states of $H_\C(\Omega=0)$. Finding the
maximum-weight independent sets $L_\C$ for a given blockade graph $G_\C$ is
known to be \NP-hard~\cite{Karp2009} (even for the subclass of unit disk
graphs~\cite{Clark1990}) and therefore essentially intractable for larger
structures.

Fortunately, this intractability is of no concern to us, since we are
\emph{not} interested in deriving the ground state pattern $L_\C$ for a given
structure $\C$. On the contrary, here (and in our previous
work~\cite{Stastny2023a}) we are interested in the \emph{inverse problem} of
engineering a suitable structure $\C$ for a given set of ground state patterns
$L_\C$.
Usually, these patterns are chosen such that they obey nontrivial logical
constraints to implement the solution of optimization
problems~\cite{Pichler2018b,Serret_2020,Dlaska2022,Nguyen2022,Ebadi2022,Lanthaler2023,Byun2024,Bombieri2024,Lanthaler2024},
or enforce local constraints necessary for topologically ordered quantum
phases~\cite{Samajdar_2021,Verresen2021,Tarabunga_2022,Slagle_2022,Stastny2023a,Maity_2024,Wang_2025}.
In this context, the set of bit strings $L_\C\subseteq\mathbb{Z}_2^N$ is
referred to as the \emph{language} of the structure~\footnote{%
    In theoretical computer science, a \emph{(formal) language} $L$ is any
    subset $L\subseteq\Sigma^*$ of finite strings constructed from some set of
    characters $\Sigma$ called an \emph{alphabet}; the elements of a formal
    language are called \emph{words}. $\Sigma^*$ is called the \emph{Kleene
    closure} and denotes the set of all finite strings of characters in
    $\Sigma$. In our case $\Sigma=\{0,1\}$ and $\Sigma^*$ denotes the set of
    all finite bit strings. Note that $\mathbb{Z}_2^N\subset\Sigma^*$ is the
    set of bit strings of uniform length $N$.
}, and the states $\ket{\vec n}$ with $\vec n\in L_\C$ as its \emph{logical
states} (we often use patterns $\vec n$ and states $\ket{\vec n}$
interchangeably). The purpose of a structure $\C$ is therefore to isolate
a prescribed subset of excitation patterns $L_\C$ (realized by the spectral
gap $\Delta E>0$), while remaining agnostic to the differences between
patterns within $L_\C$ (realized by the perfect degeneracy of $H_\C$). How
this inverse problem can be tackled systematically has been studied
previously~\cite{Nguyen2022,Stastny2023a,Lanthaler2023} and is not our main
focus here.  However, to make this paper self-contained, we provide a brief
review in \cref{sec:review} below.

So let us assume henceforth that we are given a structure $\C$ that realizes
some prescribed language $L_\C$ of logical states as its classical ground state
manifold $\H_\C$. In this paper, we are interested in the effect of weak
quantum fluctuations $|\Omega|\ll\Delta E\sim\Delta_i$ on the low-energy physics
of $H_\C(\Omega\neq 0)$. Most importantly, we are interested in the new ground
state(s) of this strongly interacting quantum many-body system,
\cref{fig:setting}~(d). Without loss of generality, we can expect these to be
of the form
\begin{align}%
    \label{eq:ground_state}
    \ket{\Omega}
    =\sum_{\vec n\in L_\C}\lambda_{\vec{n}}(\Omega) \ket{\vec n}
    +\sum_{d\geq 1}\left(\tfrac{\Omega}{\Delta E}\right)^d
    \sum_{\vec n\in L_\C^d}\eta_{\vec{n}}(\Omega)\ket{\vec n}
\end{align}%
with amplitudes $\lambda_{\vec{n}},\eta_{\vec{n}}\in\O(1)$. Here we introduced
the sets of excitation patterns $L_\C^d$ that have Hamming distance $d$ from
$L_\C^0\equiv L_\C$, i.e., $\vec n\in L_\C^d$ means that $d$ is the minimum
number of bits to flip such that $\vec n$ becomes a valid logical state in
$L_\C$.
For $\Omega=0$, the amplitudes $\lambda_{\vec{n}}$ are clearly arbitrary, but
once fluctuations set in, this is no longer the case in general. Presumably,
there is a unique ground state characterized by some particular superposition
$\{\lambda_{\vec n}\}$ of logical patterns from $L_\C$, dressed by patterns
$\vec n\in L_\C^{d}$ that violate the constraints of $L_\C$. To quantify this
dressing, we define $\Lambda^2:=\sum_{\vec n\in L_\C}|\lambda_{\vec{n}}|^2$ as
the \emph{logical weight} of $\ket{\Omega}$.

Since the classical structure $\C$ is designed to single out -- but otherwise
not distinguish -- the patterns $\vec{n}\in L_\C$, it is reasonable to demand
the same in the presence of quantum fluctuations. This means that we are
interested in structures where $\lambda_{\vec{n}}\equiv\lambda=\const$ for all
$\vec{n}\in L_\C$ in \cref{eq:ground_state} even for $\Omega\neq 0$. Thus we
seek blockade structures that stabilize \emph{equal-weight superpositions} of
logical states in the presence of quantum fluctuations.

This is a reasonable objective, as we would like the classical set of
configurations $L_\C$ to ``fluctuate as strongly as possible'' in the presence
of a uniform quantum fluctuations. For example, when blockade structures based
on Rydberg atoms are employed to solve optimization problems,
the hope is that quantum fluctuations between all allowed logical states in
$L_\C$ single out the classical solution -- a rationale similar to adiabatic
quantum computing~\cite{Albash2018}. Here we do not focus on optimization
problems, but are interested in the preparation of interesting quantum
many-body phases. For instance, Anderson's famous resonating valence bond
states \cite{Anderson1973} and the related dimer states \cite{Rokhsar1988} can
be idealized as equal-weight superpositions of an extensive set of constrained
product states \cite{Schuch2012}. Similarly, the topological quantum phase of
the paradigmatic toric code model \cite{Kitaev2003} is characterized by an
equal-weight superposition of ``loop states'' (see \cref{subsec:toric} below).

In the first part of this paper, we therefore develop methods to characterize
and construct generic blockade structures that stabilize equal-weight
superpositions (\cref{sec:motivation,sec:sym,sec:fsu}). Subsequently, we apply
these methods to construct a periodic blockade structure that actually prepares
an equal-weight loop condensate in its ground state (\cref{sec:const}), and
lastly, we show that the loss of logical weight $\Lambda^2<1$ due to ``non-loop
states'' does not destroy the topological order (\cref{subsec:gs}). All central
results that follow are constructive and/or analytically rigorous, without
resorting to numerical or approximate techniques.

\section{Review: Blockade structures}
\label{sec:review}

\begin{figure*}[tb]
    \centering
    \includegraphics[width=0.9\linewidth]{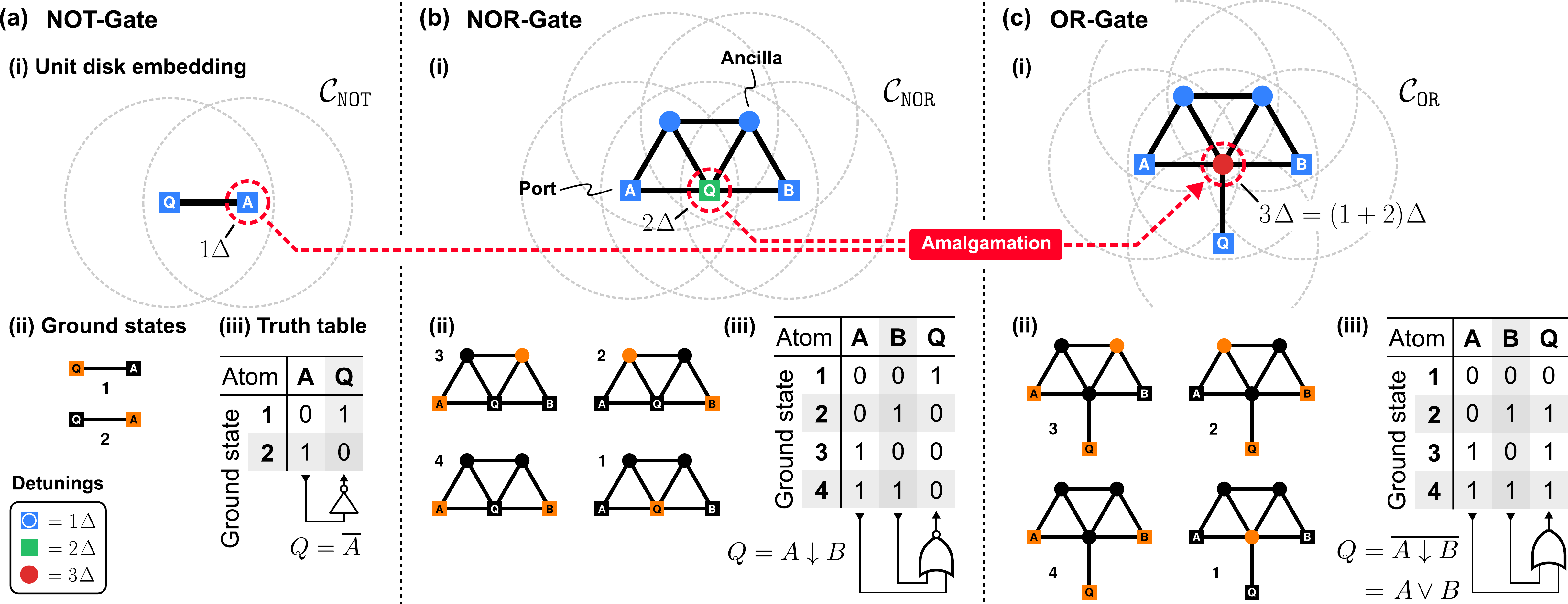}
    \caption{%
	\emph{Review of blockade structures.} 
    (a)~The simplest blockade structure $\C_\texttt{NOT}$ consists of two
    equally detuned atoms in blockade: $\Delta_Q=\Delta_A=1\Delta$ (blue
    squares)~(i). Its degenerate ground state manifold $\H_{\C_\texttt{NOT}}$
    is spanned by the two states $\ket{01}$ and $\ket{10}$ (ii). It therefore
    realizes the language $L_\texttt{NOT}=\{01,10\}$, which contains bit
    strings that correspond to the rows of the truth table of the
    \texttt{NOT}-gate (iii).
    (b)~The five-atom structure $\C_\texttt{NOR}$ consists of three ports
    (squares) and two ancillas (disks)~(i).  Two ports and both ancillas are
    detuned by $1\Delta$ (blue) while the central port is detuned by
    $\Delta_Q=2\Delta$ (green). As a consequence, the structure has a four-fold
    degenerate ground state manifold (ii). If one ignores the ancillas and
    labels the ground states by the excitation patterns of the ports, one finds
    the language $L_\texttt{NOR}=\{001,010,100,110\}$ which corresponds to the
    truth table of the universal \texttt{NOR}-gate (Not-\texttt{OR}) (iii).
    (c)~The \texttt{NOT}-structure $\C_\texttt{NOT}$ can be ``glued'' to the
    output port of the \texttt{NOR}-structure $\C_\texttt{NOR}$ by identifying
    this port atom with an input port atom of $\C_\texttt{NOT}$, thereby adding
    up their detunings to $3\Delta$ (red)~(i). This procedure is called
    \emph{amalgamation} and results in an \texttt{OR}-structure, as can bee
    seen from the ground state patterns (ii) that realize the truth table (iii)
    of an \texttt{OR}-gate. The existence of the \texttt{NOR}-structure and the
    possibility of amalgamation make the toolbox of blockade structures
    functionally complete.
    }
    \label{fig:review}
\end{figure*}

To make this paper self-contained, we continue with a brief review of previous
results on \emph{classical} blockade structures, i.e., without quantum
fluctuations ($\Omega=0$). Readers familiar with Ref.~\cite{Stastny2023a} can
skip this section and continue with \cref{sec:motivation}.

As emphasized above, the overarching goal is the \emph{construction} of a
blockade structure $\C$ for a given language $L$ (= set of bit strings) such
that the ground state patterns of $H_\C(\Omega=0)$ realize this language: $L_\C
= L$.  To this end, a versatile class of languages $L$ to consider are the
\emph{rows} of truth tables of Boolean functions. This is so, because any set
$L\subseteq\mathbb{Z}_2^n$ of bit strings $\vec x=(x_1,\ldots,x_n)$ of uniform
length $n$ can be encoded as a solution to a constraint of the form
$f(x_1,\ldots,x_n)=1$ for some Boolean function $f$. It is therefore convenient
to introduce the class of languages
\begin{align}
    L_f:=\left\{(\vec x,f(\vec x))\,|\,\vec x\in\mathbb{Z}_2^n\right\}
\end{align}
which contain all $2^n$ rows of the truth table of the Boolean function $f$.
If one finds a structure $\C_f$ such that $L_{\C_f}=L_f$, one can simply add an
additional detuning to the atom that encodes the value $f(\vec x)$ to produce a
new structure $\C_{f=1}$ that realizes the language
\begin{align}
    L_{f=1}:=\left\{\vec x\,|\,f(\vec x)=1,\;\vec x\in\mathbb{Z}_2^n\right\}\,.
\end{align}

It is therefore paramount to understand the construction of structures that
realize the language $L_f$ for a given Boolean function $f$. For example, the
simplest such function is the \texttt{NOT}-gate
\begin{align}
    \texttt{NOT}(x):=\bar x\,,
\end{align}
and the language derived from its truth table is simply
\begin{align}
    L_\texttt{NOT}=\{01,10\}\,.
\end{align}
Our task is to construct a blockade structure $\C_\texttt{NOT}$ such that
$L_{\C_\texttt{NOT}}=L_\texttt{NOT}$, i.e., $H_{\C_\texttt{NOT}}$ should have
the logical ground state manifold
$\H_\texttt{NOT}=\operatorname{span}\{\ket{01},\ket{10}\}$. This is of course
realized by a single blockade between two equally detuned atoms, see
\cref{fig:review}~(a). Throughout the paper, detunings are given as multiples
of the unit detuning $\Delta$ and color coded blue ($1\Delta$), green
($2\Delta$), red ($3\Delta$), and pink ($4\Delta$).

$\C_\texttt{NOT}$ is the most elementary, nontrivial blockade structure. A less
trivial Boolean function is the \texttt{NOR}-gate (Not-\texttt{OR})
\begin{align}
    \texttt{NOR}(x,y)\equiv x\downarrow y:=\overline{x\vee y}
\end{align}
with associated language
\begin{align}
    \label{eq:L_NOR}
    L_\texttt{NOR}=\{001,010,100,110\}\,,
\end{align}
where the last bit encodes the result of the \texttt{NOR}-gate applied to the
first two bits. One can show that no structure with three atoms can realize
this language~\cite{Stastny2023a}. However, if one allows for additional
\emph{ancilla} atoms, it becomes possible with the \emph{five}-atom structure
$\C_\texttt{NOR}$ depicted in \cref{fig:review}~(b).
The three atoms of $\C_\texttt{NOR}$ with states that map one-to-one to bit
patterns of $L_\texttt{NOR}$ are called \emph{ports} (henceforth depicted as
squares). The remaining two atoms do not contribute additional degrees of
freedom to the ground state manifold and are referred to as \emph{ancillas}
(henceforth depicted as disks). The excitation patterns of both ports and
ancillas in the ground state manifold of this structure realize the language
\begin{align}
    L_{\C_\texttt{NOR}}=\{001{\color{gray}00},010{\color{gray}01},100{\color{gray}10},110{\color{gray}00}\}\,,
\end{align}
where the states of ports are printed black and the states of the two ancillas
gray. If one ignores the ancilla states (which carry redundant information),
this is equivalent to the \texttt{NOR}-language~\eqref{eq:L_NOR}. When
referring to ground states, we often omit the ancilla states altogether and
simply identify $L_{\C_\texttt{NOR}}$ with $L_\texttt{NOR}$.

What makes the existence of $\C_\texttt{NOR}$ remarkable is the well-known fact
that \texttt{NOR}-gates are \emph{universal}: Any Boolean function
$f(x_1,\ldots,x_n)$ can be written as a Boolean circuit built only with
\texttt{NOR}-gates~\cite{Sheffer1913,Wernick1942}. This suggests that one can
construct a blockade structure for \emph{any} language $L_f$, if there is a
systematic way to implement such circuit decompositions by ``gluing'' several
\texttt{NOR}-structures together. This is made possible by a construction
called \emph{amalgamation}; the rules are very simple: Two structures that
realize Boolean gates can be joined by identifying the atoms that correspond to
the in- and output ports to be connected (in the circuit sense), and
\emph{adding up} the detunings that the two structures assign to the identified
atoms. One easily verifies that the new structure implements the truth table of
the Boolean function that the two connected gates define.

This procedure is exemplified in \cref{fig:review}~(c), where the
\texttt{NOT}-structure $\C_\texttt{NOT}$ from \cref{fig:review}~(a) is
``glued'' to the output port of the \texttt{NOR}-structure $\C_\texttt{NOR}$
from \cref{fig:review}~(b), producing an \texttt{OR}-structure
$\C_\texttt{OR}$. It is straightforward to check that the ground state manifold
of this amalgamation implements the truth table of an \texttt{OR}-gate: 
\begin{align}
    \label{eq:L_OR}
    L_\texttt{OR}=\{000,011,101,111\}\,.
\end{align}
One can even show that this particular six-atom structure is minimal in the
number of atoms~\cite{Stastny2023a}, although amalgamations typically produce
non-minimal structures.

In summary, the existence of a \texttt{NOR}-structure, and the possibility to
combine structures via amalgamation, essentially proves that any language $L$
of length-$n$ bit strings can be realized by some blockade structure $\C$ with
$N\geq n$ atoms in the sense that $L_\C=L$ if one ignores ancillas. A rigorous
derivation, including some technicalities, is presented in
Ref.~\cite{Stastny2023a}. We stress that this proof works in strictly two
dimensions and is constructive, so that $\C$ can be engineered efficiently in
principle. However, the structures derived in this way are typically very
convoluted, and numerical optimization techniques must be used to identify
functionally equivalent structures with fewer atoms (preferably as few as
possible). In Ref.~\cite{Stastny2023a} we compiled an exhaustive list of such
minimal structures that realize all elementary Boolean gates (\texttt{AND},
\texttt{OR}, \texttt{XOR},~\ldots).

This result makes the blockade toolbox functionally complete and serves as
foundation for the present paper: Given an arbitrary Language $L$, we can
always assume the existence of a classical blockade structure $\C$ with
degenerate ground state manifold $\H_\C$, spanned by states with excitation
patterns in $L$.

\section{Motivation: Logic primitives}
\label{sec:motivation}

We now turn to structures with quantum fluctuations ($\Omega\neq 0$). To
illustrate the concepts introduced in \cref{sec:setting}, we consider two
simple examples: one with and one without an equal-weight superposition of
logical states as ground state. These discussions motivate subsequent
generalizations, which then facilitate the main results of the paper.

\begin{figure*}[tb]
    \centering
    \includegraphics[width=0.90\linewidth]{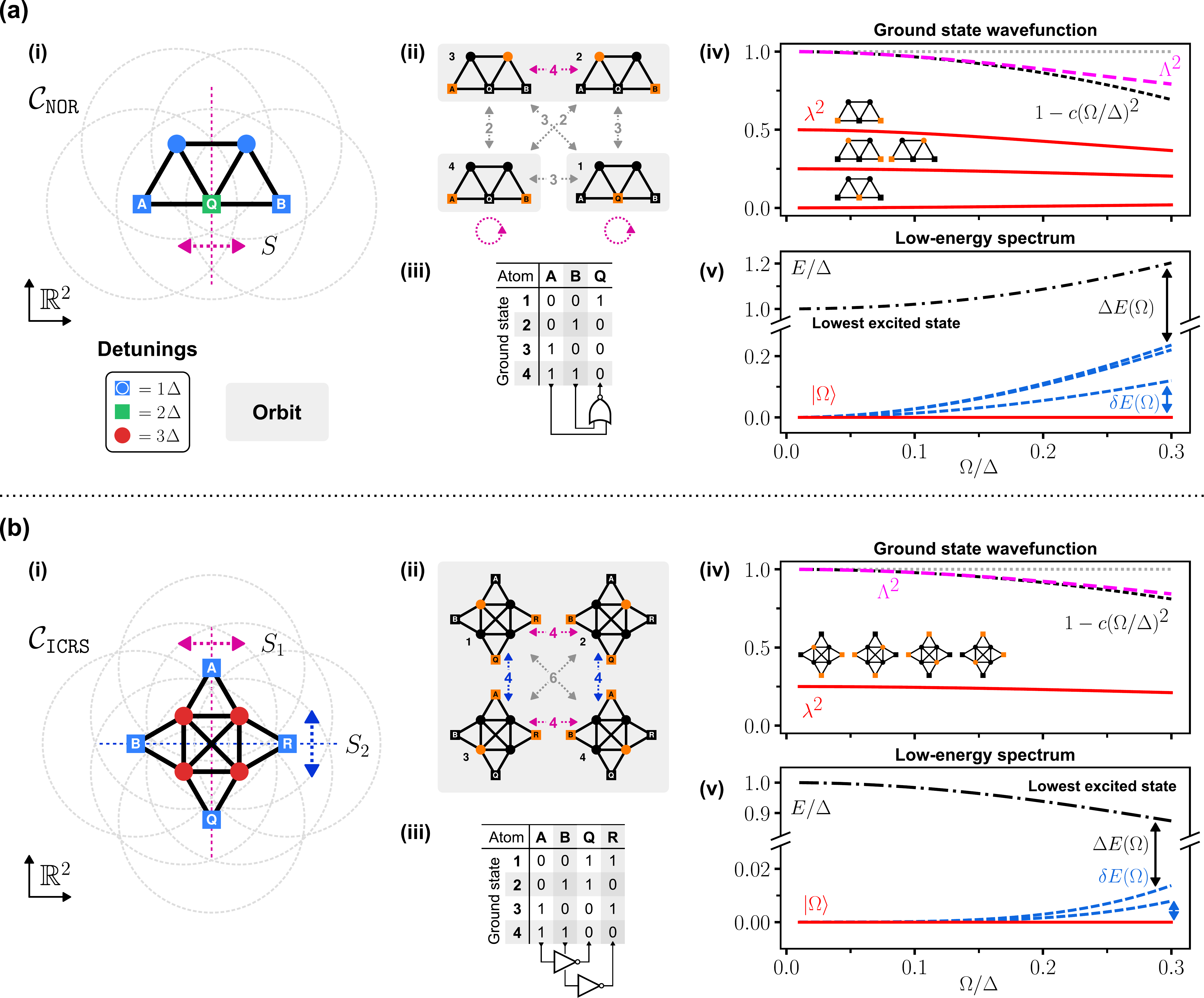}
    \caption{%
	\emph{Motivation: Logic gates.} 
    (a)~(i)~The \texttt{NOR}-gate introduced in Ref.~\cite{Stastny2023a}.  It
    is related to a triangle-based family of primitives that exhaust all
    fundamental logic gates. 
    (ii)~The structure has four degenerate ground states that realize the truth
    table of a \texttt{NOR}-gate (iii). The numbers at arrows in (ii) denote
    Hamming distances between excitation patterns. The gray boxes indicate the
    orbits (for a definition see main text) in which the excitation patterns
    transform under the symmetry $S$ highlighted in (i). 
    (iv-v)~With quantum fluctuations, the ground state of $\C_\texttt{NOR}$ is
    \emph{not} equal-weight, but favors one logic configuration above all
    others. 
    (b)~By contrast, the (inverted) crossing $\C_\texttt{ICRS}$ has an
    equal-weight ground state for $\Omega\neq 0$ (iv-v). It is also highly
    symmetric (i-ii): The symmetries $S_1$ and $S_2$ generate graph
    automorphisms that transform all ground state patterns into each other
    [gray box in (ii)].
    }
    \label{fig:motivation}
\end{figure*}

We start with the blockade structure $\C_\texttt{NOR}$ for a \texttt{NOR}-gate
introduced in Ref.~\cite{Stastny2023a} and reproduced in
\cref{fig:motivation}~(a-i). As reviewed in \cref{sec:review}, classically it
has four degenerate ground states
\begin{align}
    L_\texttt{NOR}=\{\;
    \underset{\substack{\uparrow\\\mathrm{A}}}{0}
    \underset{\substack{\uparrow\\\mathrm{B}}}{0}
    \underset{\substack{\uparrow\\\mathrm{Q}}}{1}
    {\color{gray}00}
    ,\,
    010{\color{gray}01},\,
    100{\color{gray}10},\,
    110{\color{gray}00}
    \;\}
    \label{eq:LNOR}
\end{align}
that are easily identified as the rows of the truth table of a Boolean
\texttt{NOR}-gate, \cref{fig:motivation}~(a-ii) and (a-iii). The black bits in
\cref{eq:LNOR} denote the states of the \emph{ports} $A$, $B$, and $Q$, and
identify the ground states uniquely. By contrast, the gray bits carry redundant
information and encode the states of \emph{ancillas} (which we can omit when
referring to ground states).

When switching on quantum fluctuations $\Omega$, the degeneracy of the ground
state manifold is lifted with a gap $\delta E\sim \Omega^2/\Delta E$ and
produces a unique ground state $\ket{\Omega}$, \cref{fig:motivation}~(a-v). The
exponent of $\Omega$ follows from the leading order perturbation, which
corresponds to the smallest Hamming distance between ground state
configurations in \cref{fig:motivation}~(a-ii). Exact diagonalization reveals
the relative weights
$\{\lambda_{001},\lambda_{010},\lambda_{100},\lambda_{110}\}$ of the four
logical states in $L_\texttt{NOR}$, and the logical weight
$\Lambda^2=\lambda_{001}^2+\lambda_{010}^2+\lambda_{100}^2+\lambda_{110}^2$ in
the logical space $\H_\texttt{NOR}$, see \cref{fig:motivation}~(a-iv). We find
that $\lambda_{100}=\lambda_{010}\neq\lambda_{001}\neq\lambda_{110}$, i.e., the
new ground state is \emph{not} equal-weight and the quantum fluctuations
discriminate between the four logic states. Straightforward numerics reveals
that this problem affects the complete family of Boolean gates introduced in
Ref.~\cite{Stastny2023a}. This suggests that, unsurprisingly, equal-weight
superpositions are not typical.

Next, we consider the 8-atom blockade structure $\C_\texttt{ICRS}$ introduced
in Refs.~\cite{Nguyen2022,Stastny2023a} and depicted in
\cref{fig:motivation}~(b). It has four degenerate ground states, labeled by the
configurations of four ports
\begin{align}
    L_\texttt{ICRS}=\{\;
    \underset{\substack{\uparrow\\\mathrm{A}}}{0}
    \underset{\substack{\uparrow\\\mathrm{B}}}{0}
    \underset{\substack{\uparrow\\\mathrm{Q}}}{1}
    \underset{\substack{\uparrow\\\mathrm{R}}}{1}
    ,\,
    0110,\,
    1001,\,
    1100\;
    \}\,,
\end{align}
which can be interpreted as copying two independent bits from ports A and B to
ports Q and R after inverting them. Geometrically, the structure realizes an
(inverted) crossing of two Boolean wires, a crucial primitive to realize
arbitrary Boolean circuits in the plane~\cite{Nguyen2022,Stastny2023a}. An
analogous analysis of $\C_\texttt{ICRS}$ for $\Omega\neq 0$ reveals that
$\lambda_{0011}=\lambda_{0110}=\lambda_{1001}=\lambda_{1100}$, i.e., the
inverted crossing already achieves our goal and stabilizes a ground state with
maximal fluctuations between the four logic states,
\cref{fig:motivation}~(b-iv).

A quick comparison of the geometries and degeneracies in \cref{fig:motivation}
suggests a relation between \emph{symmetries} of the structure and equal-weight
superpositions in the ground state. The \texttt{NOR}-structure
$\C_\texttt{NOR}$ features a single mirror symmetry about its central in-plane
axis. This symmetry acts on the excitation patterns in $L_\texttt{NOR}$ and
induces \emph{three} distinct sets of excitation patterns that are invariant
under the symmetry: $O_1=\{010,100\}$, $O_2=\{001\}$, and $O_3=\{110\}$,
matching the three values of ground state amplitudes in
\cref{fig:motivation}~(a-iv). By contrast, the \texttt{ICRS}-structure features
the full dihedral group $D_4$ as symmetry group. Its action on
$L_\texttt{ICRS}$ produces a single invariant set $O_1=\{0011,0110,1001,1100\}$
that (presumably) enforces the equal weights of logical amplitudes in the
ground state of $\C_\texttt{ICRS}$, \cref{fig:motivation}~(b-iv). In the
following, we refer to sets of excitation patterns that are invariant under a
symmetry (and do not contain proper subsets that are also invariant) as
\emph{orbits}.

\section{Symmetric blockade structures}
\label{sec:sym}

Our next goal is to formalize this notion of \emph{symmetry} of a blockade
structure, and make its relation to equal-weight superpositions rigorous. To
this end, we need some concepts from graph theory.
\begin{definition}[Graph automorphism]
    \label{df:graph_autom}
    Given a vertex-weighted graph $G=(V,E,\bm\Delta)$, an \emph{automorphism}
    $\phi$ of $G$ is a one-to-one map $\phi:V\to V$ such that
    $\Delta_i=\Delta_{\phi(i)}$, and $(\phi(i),\phi(j))\in E$ if and only if
    $(i,j)\in E$. The set of all automorphisms of $G$ forms the
    \emph{automorphism group} $\aut{G}$.
\end{definition}
Pictorially, an automorphism can be thought of as a permutation of vertices
(dragging edges along) after which the graph ``looks the same'' if one ignores
the vertex labels and allows edges to pass through each other; it is the
natural notion of a \emph{symmetry} of a graph.

The automorphisms $\A_\C:=\aut{G_\C}$ of blockade graphs act naturally on
excitation patterns $\vec n\in\mathbb{Z}_2^N$ via ${\phi\cdot(n_1\dots
n_N)}:=(n_{\phi(1)}\dots n_{\phi(N)})$; for a pattern $\vec n$, the set
$\A_\C\cdot\vec n:=\{\phi\cdot\vec n\,|\,\phi\in\A_\C\}$ is called the
\emph{orbit} of $\vec{n}$. This group action on excitation patterns then
induces a unitary representation $U_\phi$ on the atom Hilbert space
$\mathcal{H}$ via $U_\phi\ket{\vec n}:=\ket{\phi\cdot\vec n}$ (and linear
extension).

With these concepts at hand, we can make two immediate observations. First, the
unitary representation $U_\phi$ is a \emph{symmetry} of the blockade
Hamiltonian~\eqref{eq:H} for arbitrary uniform $\Omega$:
\begin{align}%
    U_\phi H_\C U_\phi^\dag = H_\C\qquad\forall\phi\in\A_\C\,.
\end{align}%
Second, this implies that the set of ground state patterns $L_\C$ is
\emph{invariant} under the group action of automorphisms: $\A_\C\cdot
L_\C=L_\C$ (for details, see \cref{app:symmetries}).

Clearly, every orbit of a group action is an invariant set, but the converse is
not true in general: invariant sets are the \emph{disjoint union} of orbits.
Hence, the set of logical patterns can be written as
\begin{align}
    L_\C=\textstyle\bigsqcup_k O_k\,,
\end{align}
where each orbit $O_k$ is invariant under the action of blockade graph
automorphisms separately. We saw this happen for the \texttt{NOR}-gate in
\cref{fig:motivation}~(a-ii). 

By contrast, the \texttt{ICRS}-structure has the peculiar property that $L_\C$
is \emph{identical} to an orbit (= is transitive), see
\cref{fig:motivation}~(b-ii). This motivates the following definition.
\begin{definition}[Fully-symmetric blockade structure]
    A blockade structure $\C$ is called \emph{fully-symmetric} if
    $\A_\C\cdot\vec n = L_\C$ for $\vec n\in L_\C$, i.e., $L_\C$ is an orbit
    under the action of~$\A_\C$.
\end{definition}
With this definition, we can finally generalize our observation that the
fully-symmetric \texttt{ICRS}-structure features an equal-weight ground state
as follows:
\begin{proposition}
    \label{prop:1} 
    Let $\C$ be a finite, fully-symmetric blockade structure. Then the ground
    state $\ket{\Omega}$ of $H_\C(\Omega)$ is \emph{unique} for $\Omega\neq 0$
    and has the form 
    \begin{align}%
        \label{eq:prop_omega}
        \ket{\Omega}
        =\lambda(\Omega)\sum_{\vec n\in L_\C} \ket{\vec n}
        +\sum_{d\geq 1}\left(\tfrac{\Omega}{\Delta E}\right)^d
        \sum_{\vec n\in L_\C^d}\eta_{\vec{n}}(\Omega)\ket{\vec n}.
    \end{align}%
    In particular, this state satisfies $U_\phi \ket{\Omega} = \ket{\Omega}$
    for every blockade graph automorphism $\phi \in \A_\C$.
\end{proposition}
We note that, in contrast to the equal-weight property, the uniqueness of the
ground state is not a consequence of the symmetry, but rather a generic
property of the blockade Hamiltonian~\eqref{eq:H}. A detailed proof of
\cref{prop:1} is given in \cref{app:proof_1}. In a nutshell: The uniqueness
follows from the form of \cref{eq:H} via the Perron-Frobenius theorem.
Consequently, the ground state forms a one-dimensional (possibly nontrivial)
representation of $U_\phi$, which then requires $\lambda_{\vec
n}=\lambda=\const$ for all $\vec n\in L_\C$ due to $L_\C$ being an orbit.

Before we proceed, it is important to emphasize that, for blockade
potentials~\eqref{eq:U}, the symmetries $U_\phi$ are \emph{not} geometric
(Euclidean) symmetries of the structure $\C$ (despite the suggestive
\cref{fig:motivation}), but topological features of the blockade graph $G_\C$.
For example, the following geometrically distinct embeddings of the
\texttt{NOR}-structure from \cref{fig:motivation}~(a) have the same blockade
graphs $G_{\C_\texttt{NOR}^1}=G_{\C_\texttt{NOR}^2}$ and therefore the same
automorphism groups $\A_{\C_\texttt{NOR}^1}=\A_{\C_\texttt{NOR}^2}$:
\begin{center}
    \includegraphics[width=0.95\linewidth]{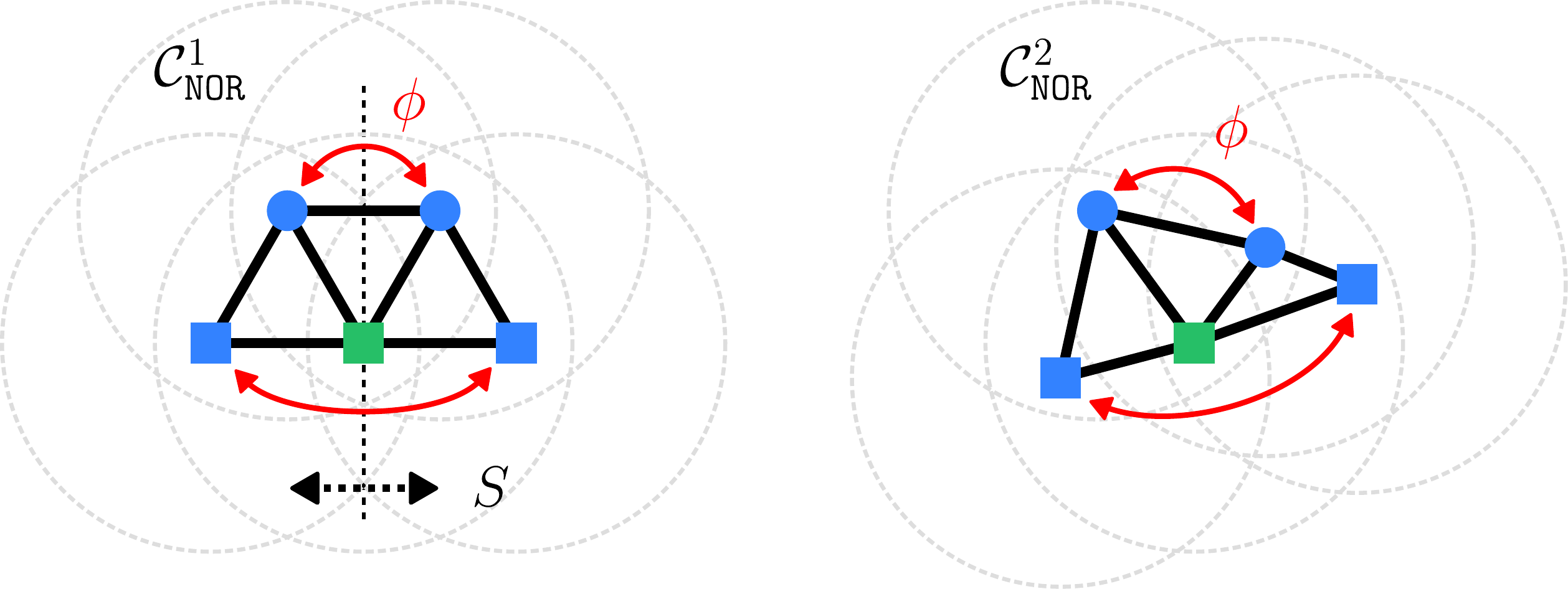}
\end{center}
This demonstrates that blockade graph automorphisms (and thereby symmetry
groups of blockade Hamiltonians) are more general than the \emph{geometric}
symmetries of a structure, as they need not be reflected in these symmetries!
For a generic structure, it is also not true that all blockade graph
automorphisms can be realized by the Euclidean symmetries of \emph{some
special} unit disk or sphere embedding (as is the case for the example above), see
\cref{app:counterexample} for an explicit counterexample.  Another
counterexample is the tessellated structure that we design in \cref{sec:const}
below: the size of its automorphism group grows exponentially with the number
of unit cells, whereas the space group of the underlying lattice is only
polynomial in size.

With that said, \cref{prop:1} translates our goal from finding blockade
structures with strong quantum fluctuations into finding fully-symmetric
structures. So far, we only know of the inverted crossing \texttt{ICRS} as an
example. The obvious follow-up question is whether there are fully-symmetric
realizations of \emph{logic gates} like \texttt{NOR}?

\section{Fully-symmetric universal gate}
\label{sec:fsu}

We start by noting that all triangle-based logic gates (studied in
Ref.~\cite{Stastny2023a}) suffer from similar asymmetries as the
\texttt{NOR}-gate discussed above, i.e., none of them are fully-symmetric. We
show now that -- quite surprisingly -- the \texttt{ICRS}-structure can be
extended such that (1) it remains fully-symmetric and (2) it realizes several
logic primitives at once; we call this the \emph{fully-symmetric universal
(\texttt{FSU}) gate} $\C_\texttt{FSU}$.

\begin{figure*}[tb]
    \centering
    \includegraphics[width=1.0\linewidth]{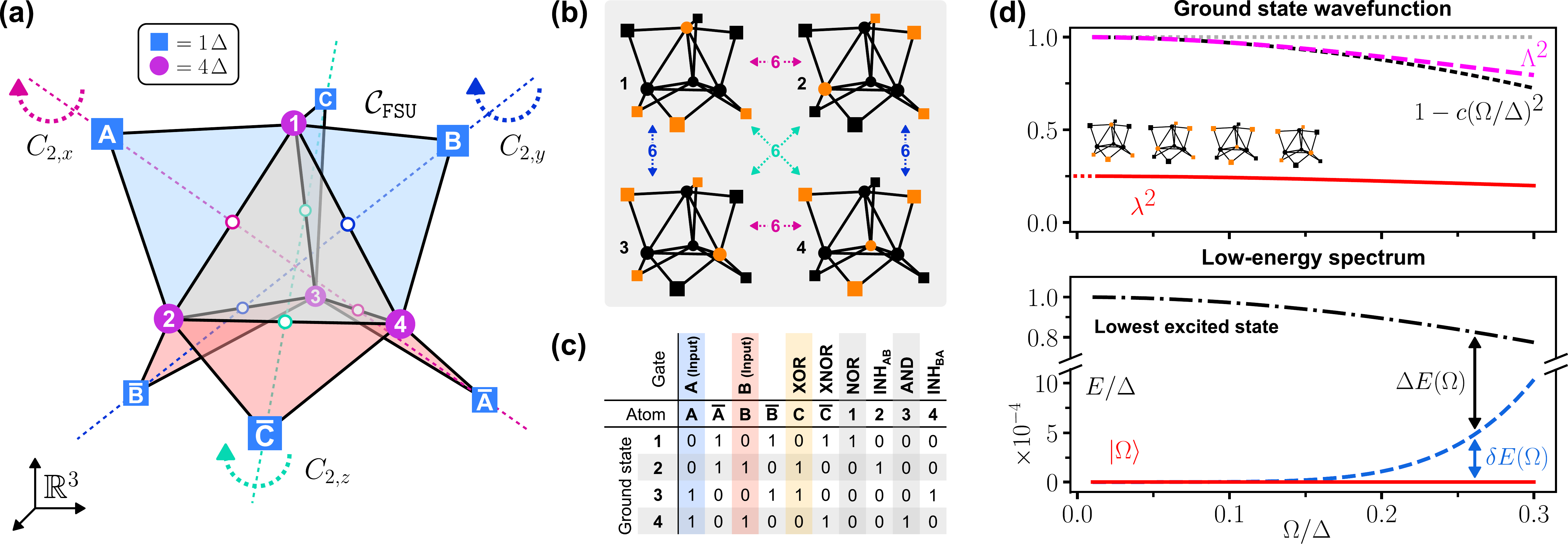}
    \caption{%
    \emph{Fully-symmetric universal (\texttt{FSU}) gate.} 
    (a)~The \texttt{FSU}-gate $\C_\texttt{FSU}$ is constructed from the
    inverted crossing by embedding it in 3D and attaching two additional atoms.
    The resulting structure has the full tetrahedral symmetry as automorphism
    group. Here we show the Klein four-group $V=\{1,C_{2,x},C_{2,y},C_{2,z}\}$,
    where $C_{2,\alpha}$ denotes the $\pi$-rotation about one of the three
    symmetry axes $\alpha=x,y,z$.
    (b)~The four degenerate ground states form a single orbit under the action
    of $V$, thereby ensuring their equal-weight superposition in the ground
    state for $\Omega\neq 0$. Furthermore, the Hamming distances (numbers on
    arrows) of all transitions are equal [cf.\ \cref{fig:motivation}~(b-ii)].
    (c)~The gate implements most Boolean primitive gates by choosing different
    atoms as ports. The six atoms on the ``wings'' form logically inverted
    antipodal pairs. We highlight the gate choice for the \texttt{XOR}-gate
    which becomes important below. $\texttt{INH}_{xy}=\neg x\wedge y$ denotes
    the ``inhibition gate'' where $x$ inhibits $y$.
    (d)~Weight of the logical states for $\Omega\neq 0$ and logical overlap in
    the ground state (upper panel) and low-energy spectrum (lower panel).
    }
    \label{fig:fsu}
\end{figure*}

The construction is illustrated in \cref{fig:fsu}~(a): We first embed the
\texttt{ICRS}-structure in three dimensions, putting the central four atoms on
the corners of a tetrahedron. We then augment the structure by two additional
atoms to produce a structure with full tetrahedral symmetry
$\A_\texttt{FSU}\simeq S_4$ ($S_4$ denotes the permutation group of four
elements). The latter has many subgroups; here we focus on the Klein four-group
$V=\{1,C_{2,x},C_{2,y},C_{2,z}\}$ which contains the $\pi$-rotations
$C_{2,\alpha}$ about the three symmetry axes through opposite \emph{edges} of
the tetrahedron. The ground state patterns $L_\texttt{FSU}$ are depicted in
\cref{fig:fsu}~(b) and transform as a single orbit under $V$, thereby producing
an equal-weight superposition for $\Omega\neq 0$ according to \cref{prop:1},
see \cref{fig:fsu}~(d). What makes this structure ``universal'' is the fact
that depending on which atoms are interpreted as ports, the structure
implements several fundamental Boolean gates at once, \cref{fig:fsu}~(c),
including a \texttt{NOR}-gate (which is a universal gate in Boolean algebra).

An exhaustive search \footnote{To generate all possible blockade graphs, we
used the tool \texttt{geng} which is a part of the package \texttt{nauty}. To
compute the automorphism groups, we also utilized the package
\texttt{nauty}~\cite{MCKAY201494}.} shows that there is no fully-symmetric
realization of \texttt{AND}- and \texttt{NOR}-gates with less than 8 atoms.
That this bound is tight follows from the \texttt{ICRS}-structure in
\cref{fig:motivation}~(b) which has 8 atoms and realizes these two gates by
reinterpreting certain ancillas as output ports. We also found that
fully-symmetric \texttt{XOR}- and \texttt{XNOR}-gates cannot be realized with
less than 10 atoms; this makes the \texttt{FSU}-structure in \cref{fig:fsu} a
minimal fully-symmetric realization of these two gates (we did not show that
this realization is unique, though). Lastly, we showed that there are no
fully-symmetric \texttt{NAND}- and \texttt{OR}-gates with 10 atoms or less
(these two gates are not realized by the \texttt{FSU}-structure).

We conclude with two observations: First, using the three atoms on the upper
``wings'' of \cref{fig:fsu}~(a) as ports yields a \texttt{XOR}-gate $\oplus$,
characterized by configurations with an \emph{even} number of excited atoms.
And second, the three atoms on the lower ``wings'' are the Boolean inverse of
their opposite (upper) counterparts, and therefore realize an
\texttt{XNOR}-gate $\odot$. The \texttt{FSU}-structure thus implements the
Boolean identities
\begin{align}
    \label{eq:fsu_feature}
    \underbrace{A\oplus B = C}_{\text{Upper ``wings''}}
    \;\Leftrightarrow\;
    \overline{A\oplus B}=\overline{C}
    \;\Leftrightarrow\;
    \underbrace{\overline{A}\odot\overline{B}=\overline{C}}_{\text{Lower ``wings''}}
\end{align}
in a fully-symmetric fashion. This feature is crucial for the last part of the
paper where we discuss our main result.

\section{A first view on tessellations}
\label{sec:firstview}

One of the most interesting features of blockade structures is that they can
enforce local constraints on translationally invariant systems to prepare
non-factorizable low-energy Hilbert spaces (see
Refs.~\cite{Verresen2021,Samajdar_2021,Stastny2023a,Zeng2025} for details and
examples). Quantum fluctuations within these nontrivial ground state manifolds
might then stabilize topologically ordered quantum phases, though this is by no
means guaranteed~\cite{Giudici2022,Sahay2022,Rourke2023}. So far, we only
considered small structures that realize Boolean logic gates. It is therefore a
reasonable next step to ask whether \emph{tessellated} (i.e.\ scalable and
periodic) structures can be fully-symmetric, as this would potentially provide
a robust mechanism to stabilize interesting quantum phases.

To understand the necessary requirements, we start by studying a generic
condition on the automorphism group of such structures. In group theory,
Burnside's lemma constrains the orbit structure of arbitrary group actions, and
thereby imposes a lower bound on the size of the automorphism group of any
fully-symmetric (FS) blockade structure (\cref{app:burnside}):
\begin{align}
    \label{eq:burnside}
    \underbrace{\left|\bigslant{L_\C}{\A_\C}\right|}_{\text{\#Orbits}}
    \stackrel{\text{FS}}{=}1
    \quad\stackrel{\text{Burnside}}{\Rightarrow}\quad
    |\A_\C| \geq |L_\C|\,.
\end{align}
That is, fully-symmetric structures must have at least as many automorphisms as
ground state configurations. This is very restrictive in the context of quantum
phases, were we want to implement nontrivial, \emph{local} degrees of freedom
in a scalable system, i.e., asymptotically $|L_\C|=\dim\H_\C\sim d^N$ for some
``quantum dimension'' $1<d<2$. This means that the automorphism group of
fully-symmetric tessellations must grow \emph{exponentially} with the system
size. In particular, the automorphisms inherited from the group of lattice
symmetries (point group and translations) cannot be sufficient, as this
subgroup is only polynomial in size.

With this knowledge, we can now revisit some of the recently proposed
tessellated blockade structures on the Rydberg platform to check whether they
are fully-symmetric. Unfortunately, it turns out that neither the structures
proposed by Verresen \etal~\cite{Verresen2021} (targeting toric code
topological order) nor by Zeng \etal~\cite{Zeng2025} (targeting quantum dimer
models) have automorphism groups large enough to stabilize equal-weight
superpositions (\cref{app:verresen,app:pichler}). The same is true for the
structures proposed by some of us in Ref.~\cite{Stastny2023a} to realize the
$\ZZ$-loop space that underlies the toric code~\cite{Kitaev2003,Kitaev2006}
(\cref{app:stastny}). Just as we saw for the logic gates above, fully-symmetric
blockade structures seem to be special and hard to come by.

Clearly, the most straightforward path to an exponentially large automorphism
group is to construct systems with \emph{local} automorphisms (``gauge
automorphisms''). Unfortunately, \cref{eq:burnside} is an \emph{implication}
and not an equivalence: Even if $|\A_\C|\geq |L_\C|$ is satisfied, the ground
state excitation patterns $L_\C$ can still split in different orbits, and the
system fails to be fully-symmetric. This happens, for example, when the local
automorphisms act only on ancillas and leave the ports invariant.

In summary, the construction of fully-symmetric tessellated blockade structures
has not been achieved yet, and it is not obvious how to proceed without more
specific knowledge about the extensive ground state configurations $L_\C$ to be
realized. This is why we focus in the remainder of the paper on a particular
problem: the construction of a tessellated blockade structure with $\ZZ$ toric
code topological order.

\section{$\ZZ$ topological order and local automorphisms}
\label{sec:local_aut}

\subsection{Reminder: Toric code topological order}
\label{subsec:toric}

We start with a brief review of the toric code to set the stage and specify
our goal. Readers familiar with the toric code~\cite{Kitaev2003} and the
results from Ref.~\cite{Stastny2023a} can safely skip to \cref{subsec:nogo}.

We consider a honeycomb lattice with one qubit (spin-$\tfrac{1}{2}$) per
\emph{edge} $e$, represented by the local Pauli matrices $\sigma_e^\alpha$
with $\alpha=x,y,z$. The toric code Hamiltonian is defined as
\begin{align}%
    \label{eq:tc_h}
    H_\sub{TC}=-J_A\sum_{\text{Vertices $s$}}A_s-J_B\sum_{\text{Faces $p$}}B_p
\end{align}%
with $J_A,J_B > 0$, \emph{star operators} $A_s=\prod_{e\in s}\sigma^z_e$,
and \emph{plaquette operators} $B_p=\prod_{e\in p}\sigma^x_e$. Here $e\in s$
denotes edges emanating from vertex $s$ and $e\in p$ denotes edges bounding
face $p$:
\begin{center}
    \includegraphics[width=0.55\linewidth]{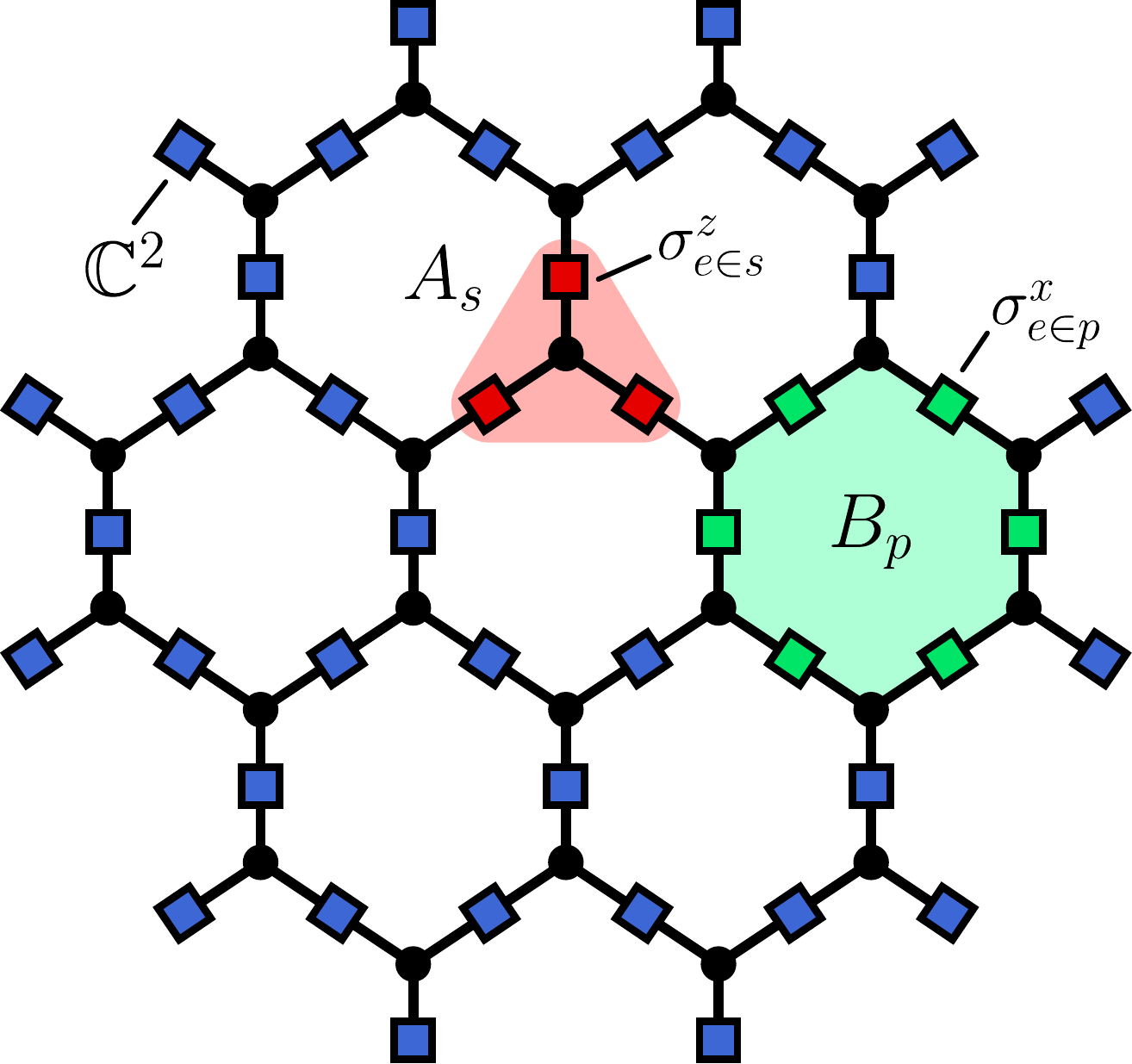}
\end{center}
Note that typically the toric code is defined on a square lattice. Here we
focus on the honeycomb lattice because its trivalent vertices make the
constructions below more natural. (By contracting the vertical edges of the
honeycomb lattice one can easily transition to a square lattice afterwards.)

To construct the ground state(s) of~\eqref{eq:tc_h}, one can exploit the fact that all
local terms $A_s$ and $B_p$ commute pairwise. We can therefore start by setting
$J_B=0$ and first construct the ground state manifold of the star operators
$A_s$ alone, and subsequently diagonalize the plaquette operators $B_p$ in this
submanifold:

\begin{itemize}

\item\textbf{Step 1:} The ground state manifold for $J_B=0$ is characterized by
    states with $+1$ eigenvalues of all star operators $A_s$~\footnote{%
        To see this, note that $A_s^2=\id$ implies eigenvalues $\pm 1$. Since
        all operators $A_s$ commute, they can be diagonalized simultaneously.
        The coefficient $-J_A$ in \eqref{eq:tc_h} then selects states with
        eigenvalue $+1$ on all sites as ground state (for $J_A>0$).
}.  
This subspace is spanned by product states with an \emph{even} number of
excited states $\ket{1}_e$ adjacent to each vertex and suggests a ``string
picture'': we interpret edges with qubits in state $\ket{1}_e$ as being
occupied by a string ($\sigma_e^z\ket{1}_e=-\ket{1}_e$). Then the constraint
imposed by star operators is simply that strings cannot terminate on vertices.
The ground state manifold for $J_B=0$ is therefore spanned by (exponentially
many) \emph{loop} configurations like this:
\begin{center}
    \includegraphics[width=0.55\linewidth]{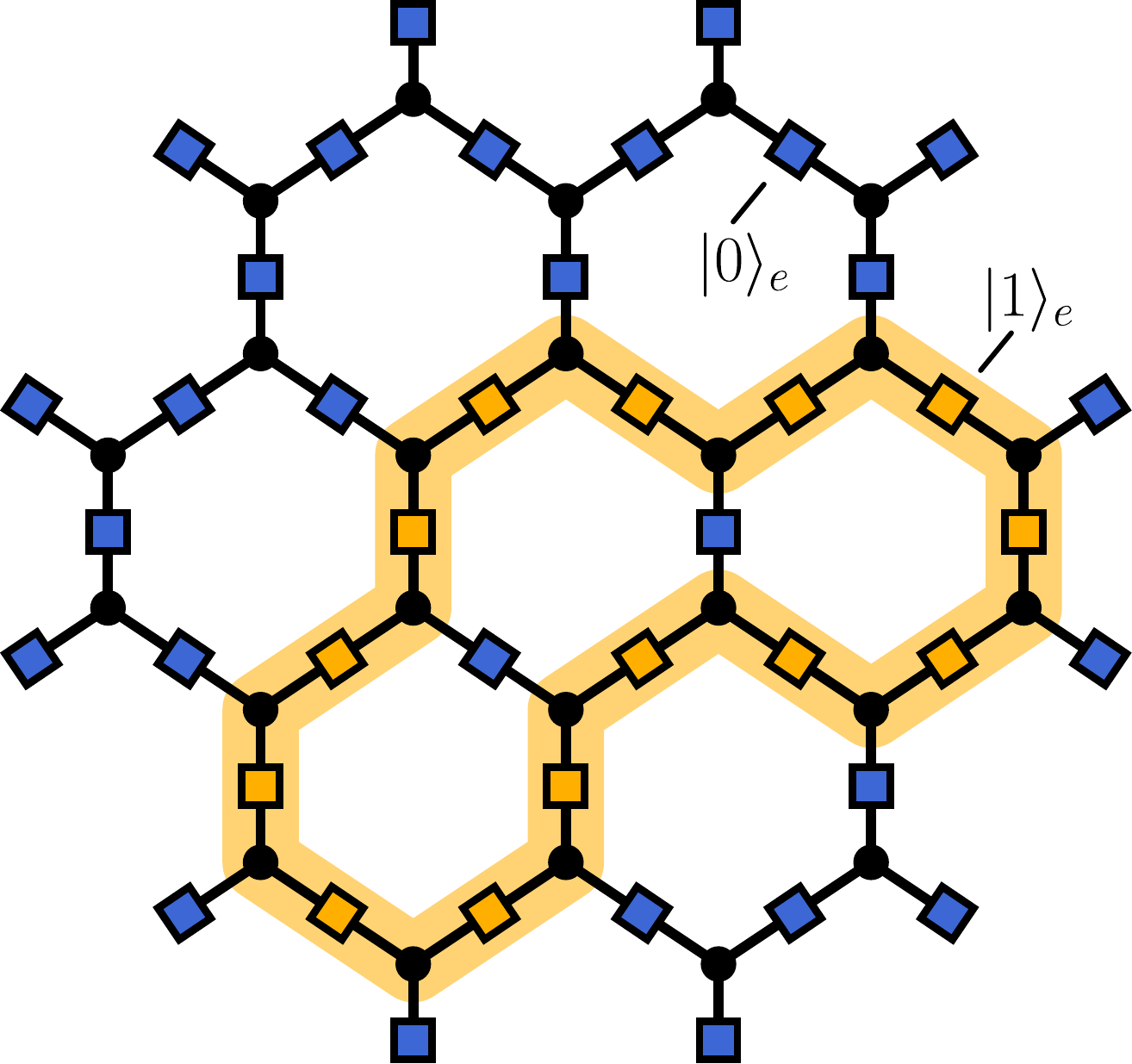}
\end{center}
Note that the boundary conditions determine whether strings are allowed to
terminate on the boundary. Above we show a finite patch with ``rough''
boundaries on which strings can end. By contrast, ``smooth'' boundaries are
obtained by cutting off the dangling edges; on such boundaries strings cannot
terminate. For periodic boundaries, strings can also not terminate anywhere.
The following discussion is independent of boundary conditions unless noted
otherwise.

Let us denote the product states of all qubits by $\ket{\vec n}$ and define
$L_\sub{Loop}$ as the set of all patterns that correspond to valid loop
configurations (respecting the chosen boundary conditions); this defines the
subspace $\H_\sub{Loop}=\operatorname{span}\{\ket{\vec n}\,|\,\vec n\in
L_\sub{Loop}\}$ spanned by these ``loop states''.

\item\textbf{Step 2:} To construct the ground state $\ket{\Omega_\sub{TC}}$
    of~\eqref{eq:tc_h}, we must also satisfy $B_p\ket{\Omega_\sub{TC}}=
    +1\,\ket{\Omega_\sub{TC}}$
    for all faces $p$ (since $B_p^2=\id$). It is easy to see that the action of
    $B_p$ on a loop state $\ket{\vec n}\in \H_\sub{Loop}$ is to add the
    elementary loop around face $p$ (modulo-2) to $\vec n\in L_\sub{Loop}$,
    thereby creating another loop configuration: $B_p\ket{\vec n}=\ket{\vec
    n'}\in \H_\sub{Loop}$ with $\vec n'\in L_\sub{Loop}$. This mapping is a
    bijection from $L_\sub{Loop}$ to itself since $B_p$ is invertible. It is
    also not hard to see that, for open boundaries (all rough or smooth), one can
    transform any loop configuration into any other by a sequence of $B_p$
    operators. 

    It follows that the only state invariant under \emph{all} plaquette
    operators is the \emph{equal-weight} superposition of all loop states
    (we omit the normalization):
    \begin{align}%
        \label{eq:tc_gs}
        \ket{\Omega_\sub{TC}}\propto\sum_{\vec n\in L_\sub{Loop}}\ket{\vec n}\,.
    \end{align}%
\end{itemize}
\cref{eq:tc_gs} is the (unique) ground state of \cref{eq:tc_h} and the
renormalization fixed point (characterized by vanishing correlation length) of
the $\mathbb{Z}_2$ toric code topological order~\cite{Levin2005}.  For other
boundary conditions (alternating smooth and rough boundaries or periodic
boundaries) there is one independent ground state for each element of the
(relative) homology group of the spatial manifold on which the lattice is
embedded~\cite{Bravyi1998a}. This topological degeneracy is not important for
the following arguments since (gapped) quantum phases are bulk features and
independent of boundary conditions.

As a condensate of extended objects ($\mathbb{Z}_2$-loops), the state
$\ket{\Omega_\sub{TC}}$ features a particular pattern of \emph{long-range
entanglement}~\cite{Chen2010}; this makes~\eqref{eq:tc_gs} the simplest example
of a \emph{string net condensate}, a family of many-body quantum states that
describe two-dimensional topological orders~\cite{Levin2005}. The particular
(abelian) toric code order in~\eqref{eq:tc_gs} makes it a potential substrate
for topological quantum memories~\cite{Dennis2002,Acharya2024}. The presence of
long-range entanglement can be certified and characterized by a non-vanishing
topological entanglement entropy~\cite{Kitaev2006a,Levin2006} (which is not
necessarily universal within a topological
phase~\cite{Williamson2019,Kim2023}).

Our goal in the remainder of this paper is to construct a tessellated blockade
structure with a many-body ground state $\ket{\Omega}$ in the same topological
quantum phase as $\ket{\Omega_\sub{TC}}$. The return of this endeavor would be
the realization of the toric code topological order using only the physically
accessible two-body interactions of the blockade Hamiltonian~\eqref{eq:H},
instead of the unrealistic four-body interactions in the toric code
Hamiltonian~\eqref{eq:tc_h}.

Before we proceed, it is important to remember that two ground states belong to
the same quantum phase if they can be adiabatically connected by a gapped
Hamiltonian~\cite{Chen2010}, i.e., the ground state $\ket{\Omega}$ of the
blockade structure (to be constructed) and $\ket{\Omega_\sub{TC}}$ are not
required to be the \emph{same} state. Nonetheless, the aforementioned criterion
ensures that they are characterized by the same pattern of long-range
entanglement~\cite{Chen2010} (recall that phases generally correspond to stable
renormalization fixed points~\cite{Wen2010,Sachdev2015}).

\subsection{Warm-up: An instructive failure}
\label{subsec:nogo}

To motivate our rather involved construction below, it is advisable to consider
a more straightforward construction first. This approach was already suggested
by some of us in Ref.~\cite{Stastny2023a} and, as already mentioned in
\cref{sec:firstview}, \emph{fails} to prepare the equal-weight
superposition~\eqref{eq:tc_gs}. The point of the following discussion is to
identify a potential \emph{reason} for this failure, which then inspires our
(successful) construction in \cref{sec:const} below.

In order to stabilize a state that is (locally) unitarily equivalent to
\cref{eq:tc_gs} using the toolbox of blockade structures, one must solve two
problems:
\begin{itemize}

    \item[\textbf{P1}] Construct a structure $\C_\sub{Loop}$ such that, for
        $\Omega=0$, the ground state manifold $\H_\sub{Loop}$ is spanned by
        degenerate loop configurations $\vec n\in
        L_{\C_\sub{Loop}}\stackrel{!}{=}L_\sub{Loop}$. Note that we can use
        additional ancilla atoms to achieve this, as long as these atoms do not
        contribute degrees of freedom to the ground state manifold. Solving
        this problem corresponds to the constraint $A_s=+1$ in the toric code
        (\textbf{Step 1}).

    \item[\textbf{P2}] Ensure that, by switching on uniform quantum
        fluctuations $\Omega\neq 0$, a state $\ket{\Omega}$ within the same
        quantum phase as $\ket{\Omega_\sub{TC}}$ is stabilized as the (unique)
        ground state of the blockade structure. A promising feature would
        certainly be an equal-weight superposition of states $\ket{\vec n}$ for
        $\vec n\in L_{\C_\sub{Loop}}$ (potentially with admixtures from other
        states). Solving this problem corresponds to the constraint $B_p=+1$ in
        the toric code (\textbf{Step 2}).

\end{itemize}
Both of these are complicated ``inverse problems'' that seem hard -- if not
impossible -- to tackle. Fortunately, the solution of \textbf{P1} is
straightforward using the toolbox from Ref.~\cite{Stastny2023a} reviewed in
\cref{sec:review}:

It is reasonable to start by placing an atom on each link of the honeycomb
lattice. If we set the detunings of all atoms to zero (white boxes), the ground
state manifold contains all $2^N$ excitation patterns and matches the full
state space of the conventional toric code (with one spin-$\tfrac{1}{2}$ per
edge):
\begin{center}
    \includegraphics[width=0.55\linewidth]{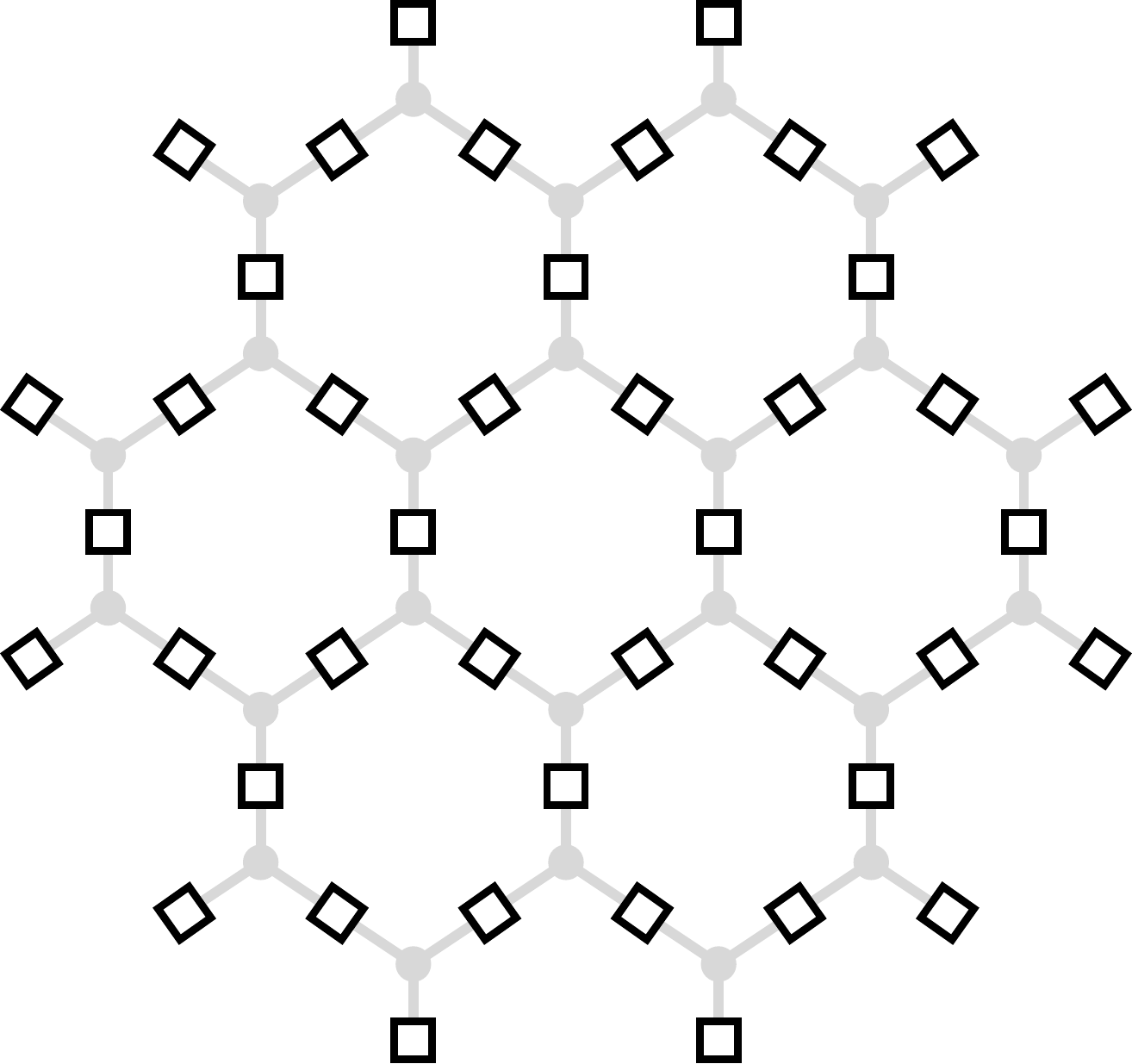}
\end{center}

Next, we must introduce detunings and blockades (and potentially ancillas) such
that only loop patterns remain in the ground state manifold. A key insight is
that the loop constraint $A_s=+1$ on a \emph{trivalent} vertex with state
$\ket{x,y,z}$ on its three emanating edges is satisfied if and only if $x\oplus
y=z$, where ``$\oplus$" denotes a \texttt{XOR}-gate (modulo-2 addition). We can
therefore pick any blockade structure $\C_\texttt{XOR}$ that realizes a
\texttt{XOR}-gate, put a copy on each vertex of the honeycomb lattice, and then
amalgamate the three ports of the gates with the adjacent atoms on the edges:
\begin{center}
    \includegraphics[width=0.55\linewidth]{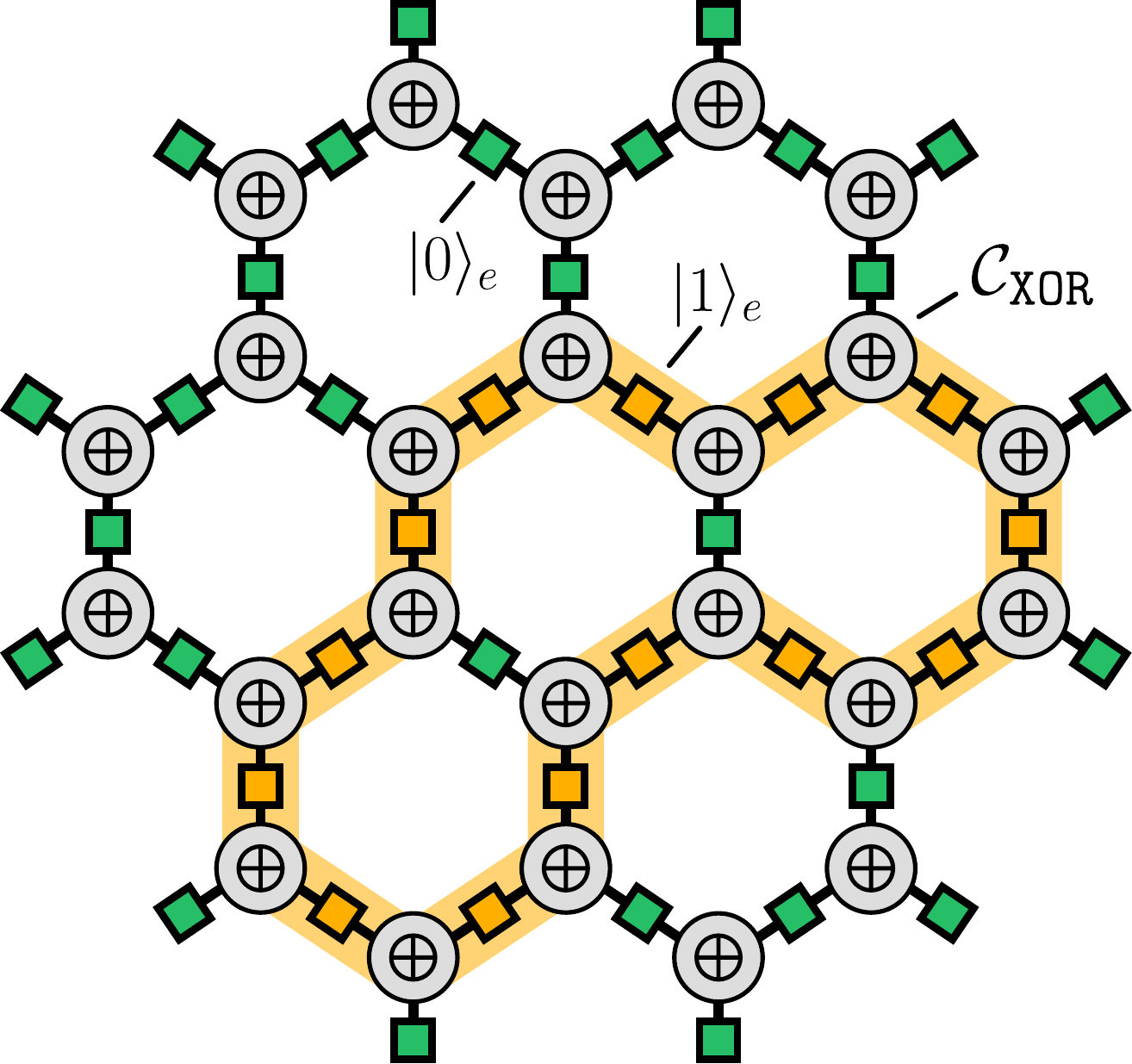}
\end{center}
Note that the constraint $x\oplus y=z$ is completely symmetric in the three
variables; hence it is not important which edges play the role of input and
output ports of the logic gates. The exemplary loop pattern shown clearly
satisfies the $\texttt{XOR}$-constraints on all vertices.

This is the approach to tackle \textbf{P1} that has been proposed in
Ref.~\cite{Stastny2023a}, where a specific (minimal) realization for the vertex
structure $\C_\texttt{XOR}$ was given. (We do not repeat this realization here
as we do not need it in the following.)

We now turn to \textbf{P2}. Since the classical part of the
Hamiltonian~\eqref{eq:H} (blockade interactions and detunings) is fixed by the
construction above (and the choice of a particular $\C_\texttt{XOR}$), all we
can do is to switch on weak, uniform fluctuations $\Omega\neq 0$ and hope for
the best. Given the general structure of Hamiltonian~\eqref{eq:H}, there is simply no
freedom left to tailor the effects of quantum fluctuations to serve our goal.
Note that this is the approach of most previous works that proposed
realizations of quantum phases by leveraging the Rydberg blockade
mechanism~\cite{Verresen2021,Samajdar_2021,Stastny2023a,Zeng2025}.

The only tool at our disposal is the concept of \emph{full symmetry} introduced
in \cref{sec:sym}. Unfortunately, as we rigorously show in \cref{app:stastny},
the particular realization of $\C_\texttt{XOR}$ proposed in
Ref.~\cite{Stastny2023a} leads to a tessellation that is \emph{not}
fully-symmetric due to a lack of automorphisms. While this does not preclude an
equal-weight superposition of states in $\H_\sub{Loop}$, it certainly makes it
unlikely (as equal-weight superpositions are not typical). To support this
claim, we derived a perturbative, low-energy effective Hamiltonian of this
particular model~\cite{Maier2023} and performed numerical simulations using
iDMRG~\cite{tenpy2024} to study its ground state properties. While there are
quantum fluctuations that induce a certain amount of local entanglement, the
results suggest that there is no long-range entanglement (certified by a
vanishing topological entanglement entropy). Together, these findings support
our suspicion that the many-body ground state of this particular blockade
structure is \emph{not} topologically ordered. Furthermore, we take the stance
that the only workable path towards solving \textbf{P2} is the construction of
a fully-symmetric tessellation.

To do so, it is important to understand whether the failure to be
fully-symmetric is caused by the particular choice of the vertex structure
$\C_\texttt{XOR}$, or whether it is rooted in the general architecture
suggested above (the proof in \cref{app:stastny} exploits the internal
structure of $\C_\texttt{XOR}$). As we show in \cref{app:proof_2}, the
architecture is to blame:
\begin{proposition}
    \label{prop:2}
    Consider a periodic, tessellated blockade structure on an arbitrary,
    regular lattice with \emph{one} port per edge that is shared between
    adjacent vertex structures; let $L$ be the set of ground state patterns.
    Then \emph{local} automorphisms must leave the ports on edges invariant,
    and therefore act as the identity on $L$.
\end{proposition}
The proof is rather technical; the gist, however, is quite simple: The global
topology of the blockade graph is severely restricted by the fact that
vertex structures interact only via a single atom (or single link) on the edges.
One can then show that automorphisms must leave ports invariant if they act
trivially on ports in the vicinity. In particular, the assumption that an
automorphism acts trivially on ports that surround a finite region implies that
it acts trivially everywhere. Note that \cref{prop:2} does not imply that
tessellated structures with single ports on edges \emph{cannot} be
fully-symmetric; however, it shows that \emph{local} automorphisms cannot be
used to establish full symmetry.

In summary, we identified two problems to solve: The $\ZZ$-loop constraint
(\textbf{P1}, $A_s=+1$) can be satisfied on the classical level using the
previously developed, versatile toolbox for blockade structures. However, the
loop condensation (\textbf{P2}, $B_p=+1$) due to uniform quantum fluctuations
must be enforced by -- preferably local -- blockade graph automorphisms that
render the structure fully-symmetric. Unfortunately, the most straightforward
architecture that solves \textbf{P1} necessarily fails to solve \textbf{P2}.
This insight inspires the construction below and leads to the main result of
the paper.

\subsection{Construction of a fully-symmetric $\ZZ$-loop blockade structure}
\label{sec:const}

In the toric code (\cref{subsec:toric}), it is the plaquette operators $B_p$
that act transitively on the space of loop configurations $L_\sub{Loop}$ and
enforce the equal-weight superposition~\eqref{eq:tc_gs}. We therefore should
aim for a blockade structure with local ``plaquette automorphisms'' that do the
same.  However, \cref{prop:2} forbids exactly that if we construct the
tessellation from vertex structures connected by single links. This makes
sense, since automorphisms can only \emph{permute} atoms; but, if all atoms on
the edges surrounding a face are in state $\ket{0}_e$, there is no permutation
that can excite them to $\ket{1}_e$, i.e., create a fundamental loop. Hence we
must encode the fundamental two-dimensional degrees of freedom on edges not in
single atoms (as above) but in \emph{pairs} of atoms that are in blockade:
\begin{center}
    \includegraphics[width=0.55\linewidth]{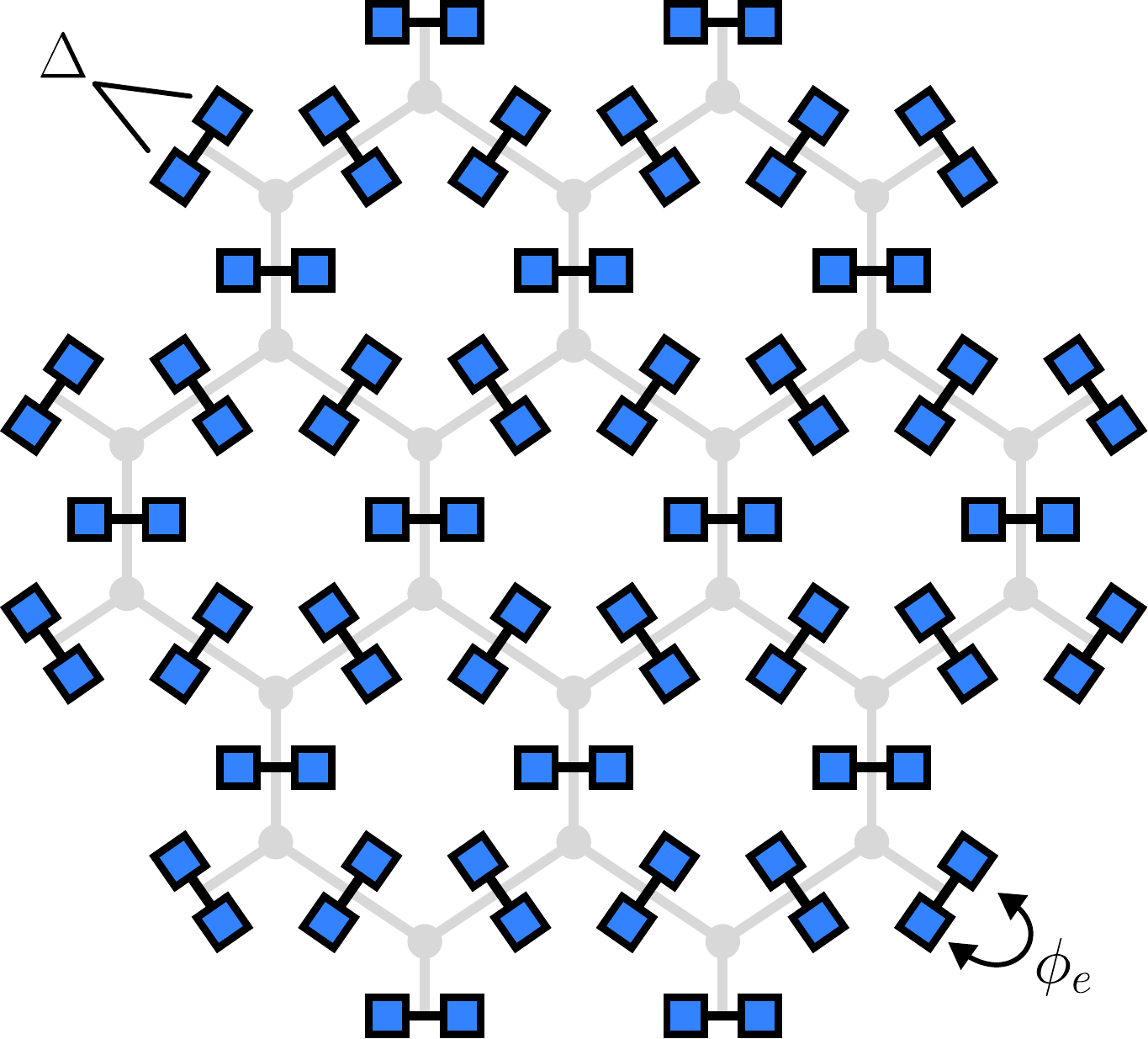}
\end{center}
If we detune all atoms by the same amount $\Delta > 0$ (blue filled boxes),
this produces again a degenerate ground state manifold of dimension $2^N$
(i.e., one qubit per edge), only that now these qubits are encoded in the two
states $\ket{01}_e$ and $\ket{10}_e$ of the two atoms in blockade. Notably, the
structure has a huge automorphism group; among these are \emph{local}
``edge-automorphisms'' $\phi_e$ that permute the two atoms on edge $e$. Under
these automorphisms, the set of all ground state patterns $L_N:=\{0,1\}^N$
transforms as a single orbit. Hence, this structure is fully-symmetric and its
ground state for $\Omega\neq 0$ is the equal-weight superposition of all $2^N$
bit patterns in $L_N$, i.e.,
$\ket{\Omega}\propto\bigotimes_e\ket{+}_e+(\Omega/\Delta E)\sum\ldots$ with
$\ket{+}_e=\tfrac{1}{\sqrt{2}}(\ket{01}_e+\ket{10}_e)$ and where $\sum\ldots$
includes configurations with completely unexcited edges $\ket{00}_e$.

This is not yet what we want, but now that we have enough local automorphisms,
we can try to add suitably chosen vertex structures that \ldots
\begin{enumerate}
  \item[(1)] enforce the $\ZZ$-loop constraint, and 
  \item[(2)] preserve a local subgroup of ``edge-automorphisms'' (namely: plaquette automorphisms). 
\end{enumerate}

To this end, it is convenient to separate the $2N$ atoms into two sublattices,
with one atom of each per edge (which atom belongs to which sublattice is
conventional). In the following, we label the two sublattices by red and blue
bounded boxes and blockades. Because of the blockade between the atoms on each
edge, it is sufficient to know the ground state pattern of, say, the red
sublattice to identify the pattern uniquely. Henceforth we say that an edge is
occupied by a string if the red atom is excited:
\begin{align}
    \text{String}\; \leftrightarrow\;\ket{\CR{1}\CB{0}}_e
    \quad\text{and}\quad
    \text{Empty}\; \leftrightarrow\;\ket{\CR{0}\CB{1}}_e\,.
    \label{eq:convention}
\end{align}
In the following, we omit the fill color that indicates the detuning of atoms
whenever it is not relevant for the discussion. An exemplary state in the
ground state manifold looks like this [orange (white) atoms are (un)excited]:
\begin{center}
    \includegraphics[width=0.55\linewidth]{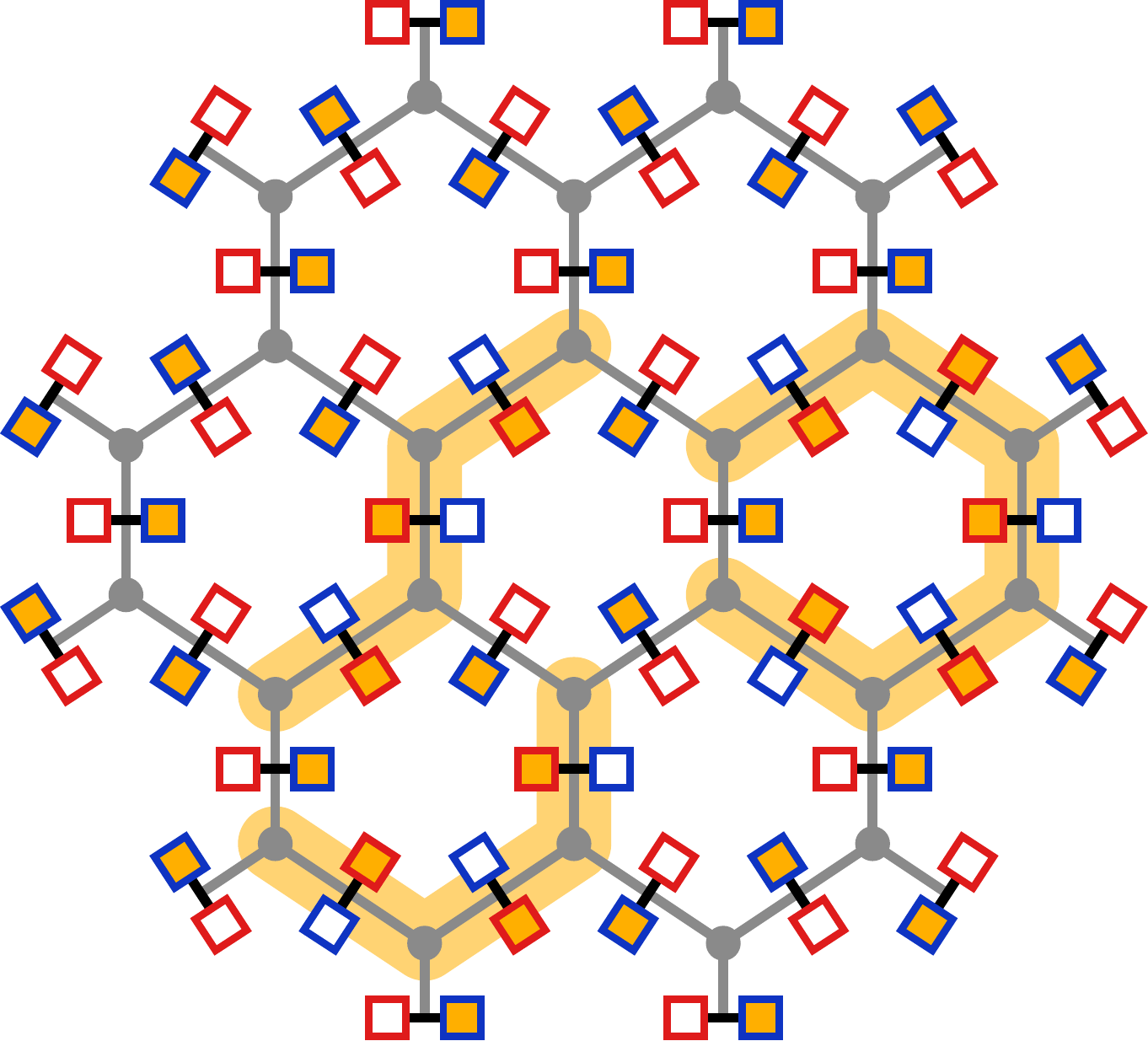}
\end{center}
Note that, so far, open strings are still part of the ground state manifold.  We
now turn to the vertex structure that lifts these patterns to excited states
and enforces a loop constraint.

Let us focus on a single vertex with adjacent variables $(\CR x,\CB{\bar x},\CR
y, \CB{\bar y},\CR z,\CB{\bar z})$, where $\CR x,\CR y,\CR z$ denote the string
occupations of the three adjacent edges. The $\ZZ$-loop constraint, according
to our convention~\eqref{eq:convention}, is then equivalent to a Boolean
\texttt{XOR}-constraint on the red sublattice,
\begin{align}
    \CR{x}\oplus\CR{y}\stackrel{!}{=}\CR{z}
    \quad\Leftrightarrow\quad\text{All permutations of $\CR{x},\CR{y},\CR{z}$}\,.
\end{align}
The red atoms must therefore be ``glued together'' by a \texttt{XOR}-structure,
just as before. However, consistency then requires that, for the atoms on the blue
sublattice,
\begin{align}
    \CR{x}\oplus\CR{y}\stackrel{!}{=}\CR{z}
    \quad\Leftrightarrow\quad
    \overline{\CR{x}\oplus\CR{y}}\stackrel{!}{=}\bar{\CR{z}}
    \quad\Leftrightarrow\quad
    \CB{\bar x}\odot\CB{\bar y}\stackrel{!}{=}\CB{\bar z},
    \label{eq:logic_const}
\end{align}
where ``$\odot$" denotes a \texttt{XNOR}-gate (exclusive-\texttt{NOR}). To enforce
the loop constraint, we therefore need a vertex structure with three inverted
\emph{pairs} of ports that realizes both a \texttt{XOR}- and a
\texttt{XNOR}-gate in a unified way. This is achieved by the
\texttt{FSU}-structure introduced in~\cref{sec:fsu}; hence we can finalize our
construction by placing a copy of $\C_\texttt{FSU}$ on every vertex and
amalgamating their port pairs on the edge atoms:
\begin{center}
    \includegraphics[width=1.0\linewidth]{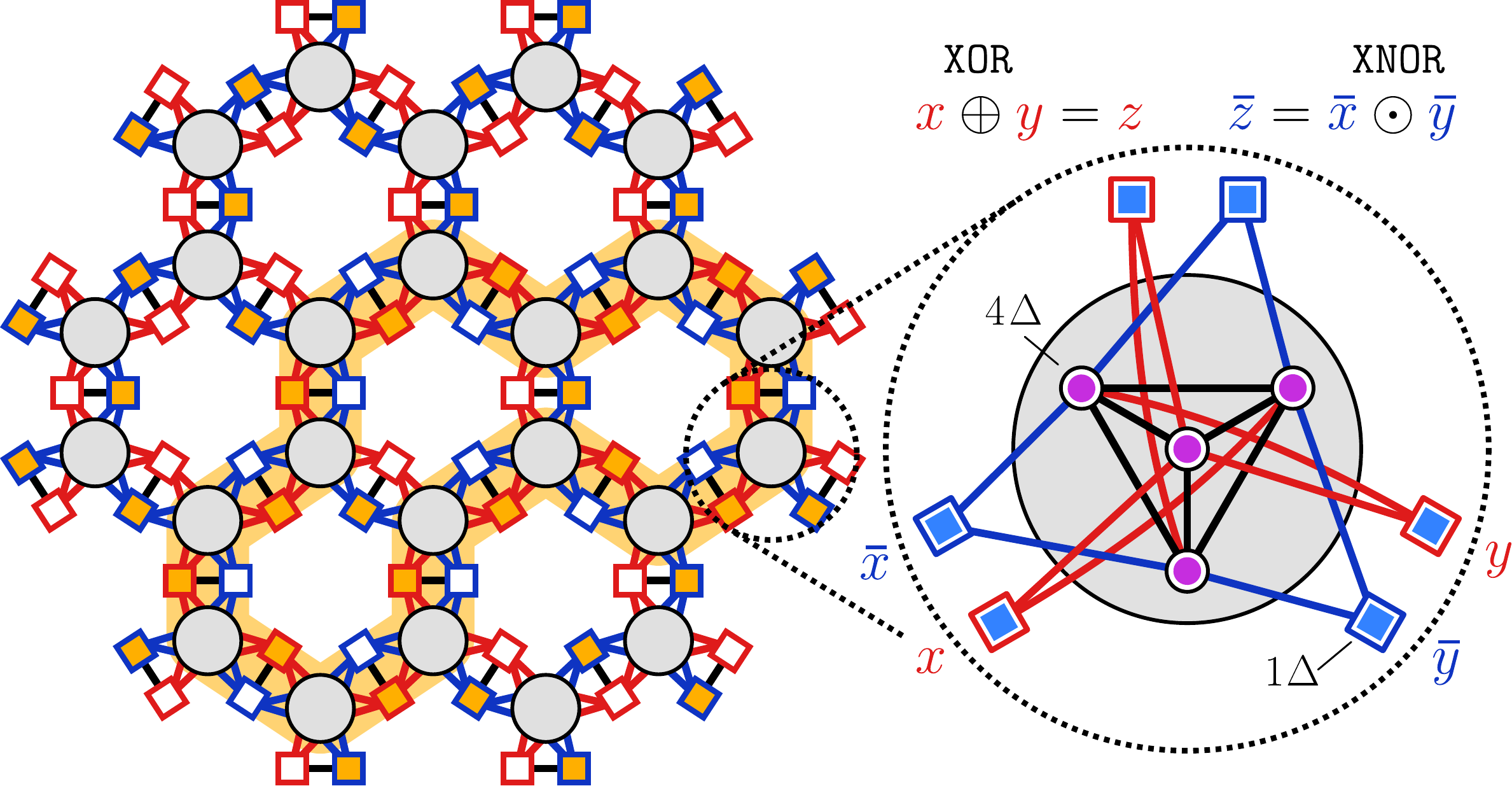}
\end{center}
[Recall \cref{fig:fsu}~(a) and \cref{eq:fsu_feature} and identify red (blue)
blockades with the upper (lower) ``wings''.]

Note that one needs mirror images of the \texttt{FSU}-structure for one of the
two vertices in each unit cell of the honeycomb lattice to ensure that only
ports of the same sublattice (= color) are amalgamated. 
Note also that the \texttt{FSU}-structure automatically enforces the constraint 
that states of
atoms of port pairs are inverted. Thus, the blockades between pairs of edge
atoms (black edges) are actually no longer needed and we can delete them from
the tessellation without affecting the loop constraint.
For now, we ignore whether this construction can be realized as unit disk or
unit ball embedding in two- or three-dimensional space, and simply project
the necessary blockades into the plane (the sketch above is not a unit disk
embedding!); we return to the question of embeddability in
\cref{sec:embedding}.

By construction, this tessellation only supports closed loop configurations in
its ground state manifold; an extended patch, including all ancillas,
blockades, and detunings, is depicted in~\cref{fig:fsu_tessellation}. We refer
to this tessellated blockade structure as $\C_\sub{Loop}$ in the following. 

\begin{figure*}[t]
    \centering
    \includegraphics[width=1.0\linewidth]{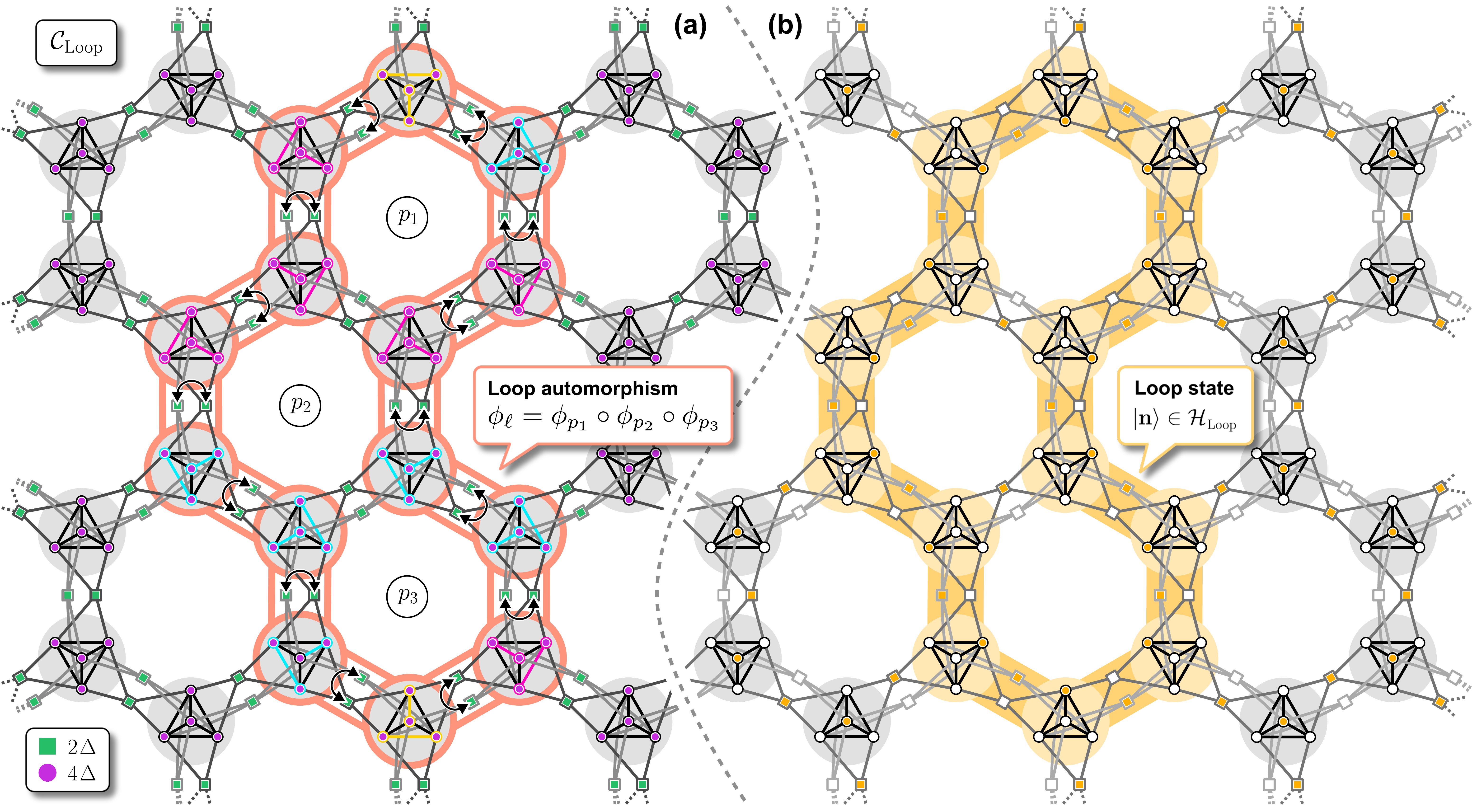}
    \caption{%
    \emph{Tessellated fully-symmetric loop structure $\C_\sub{Loop}$.} 
    Result of the construction discussed in \cref{sec:const}. The tessellation
    $\C_\sub{Loop}$ of \texttt{FSU}-vertices is shown with light gray (formerly
    red), dark gray (formerly blue), and black blockades. In (a) the fill color
    of atoms denotes the \emph{detuning} (green: $2\Delta$, pink: $4\Delta$),
    in (b) the \emph{state} of the atom (white: unexcited, orange: excited).
    (a)~Exemplary loop automorphism $\phi_\ell$ constructed as product of three
    adjacent plaquette automorphisms $\phi_{p_1},\phi_{p_2},\phi_{p_3}$. The
    automorphism decomposes into a chain of vertex- and edge permutations,
    cf.~\cref{fig:fsu_symmetry}. Permutations $\varphi_e$ of edge atoms
    (squares = ports) are denoted by arrows, permutations $\varphi_s^\alpha$ of
    atoms on vertices (circles = ancillas) by colored blockades of the
    tetrahedra. Note that only atoms with equal detunings are mapped onto each
    other. It is straightforward to check that the application of $\phi_\ell$
    leaves the complete blockade graph invariant; it therefore corresponds to a
    loop symmetry $U_\ell$ of the corresponding blockade Hamiltonian
    $H_\sub{Loop}\equiv H_{\C_\sub{Loop}}$. 
    (b)~Exemplary excitation pattern $\vec n\in L_\sub{Loop}=L_{\C_\sub{Loop}}$
    that corresponds to a ground state $\ket{\vec
    n}\in\mathcal{H}_\sub{Loop}\equiv\mathcal{H}_{\C_\sub{Loop}}$ for
    $\Omega=0$. Along the yellow loop, the light gray edge atoms are excited
    (which satisfy \texttt{XOR}-constraints on vertices), whereas everywhere
    else the dark gray edge atoms are excited (which satisfy
    \texttt{XNOR}-constraints on vertices). Note that each of the four possible
    vertex states correlates with one of the four vertex ancillas being
    excited. It is straightforward to see that the permutations of $\phi_\ell$
    in (a) would act on the pattern in (b) by removing the loop, thereby
    producing another valid ground state in $L_\sub{Loop}$. This property makes
    $\C_\sub{Loop}$ fully-symmetric.
    }
    \label{fig:fsu_tessellation}
\end{figure*}

\subsection{Automorphisms of the tessellated blockade structure}
\label{sec:localz2}

\begin{figure}[tb]
    \centering
    \includegraphics[width=1.0\linewidth]{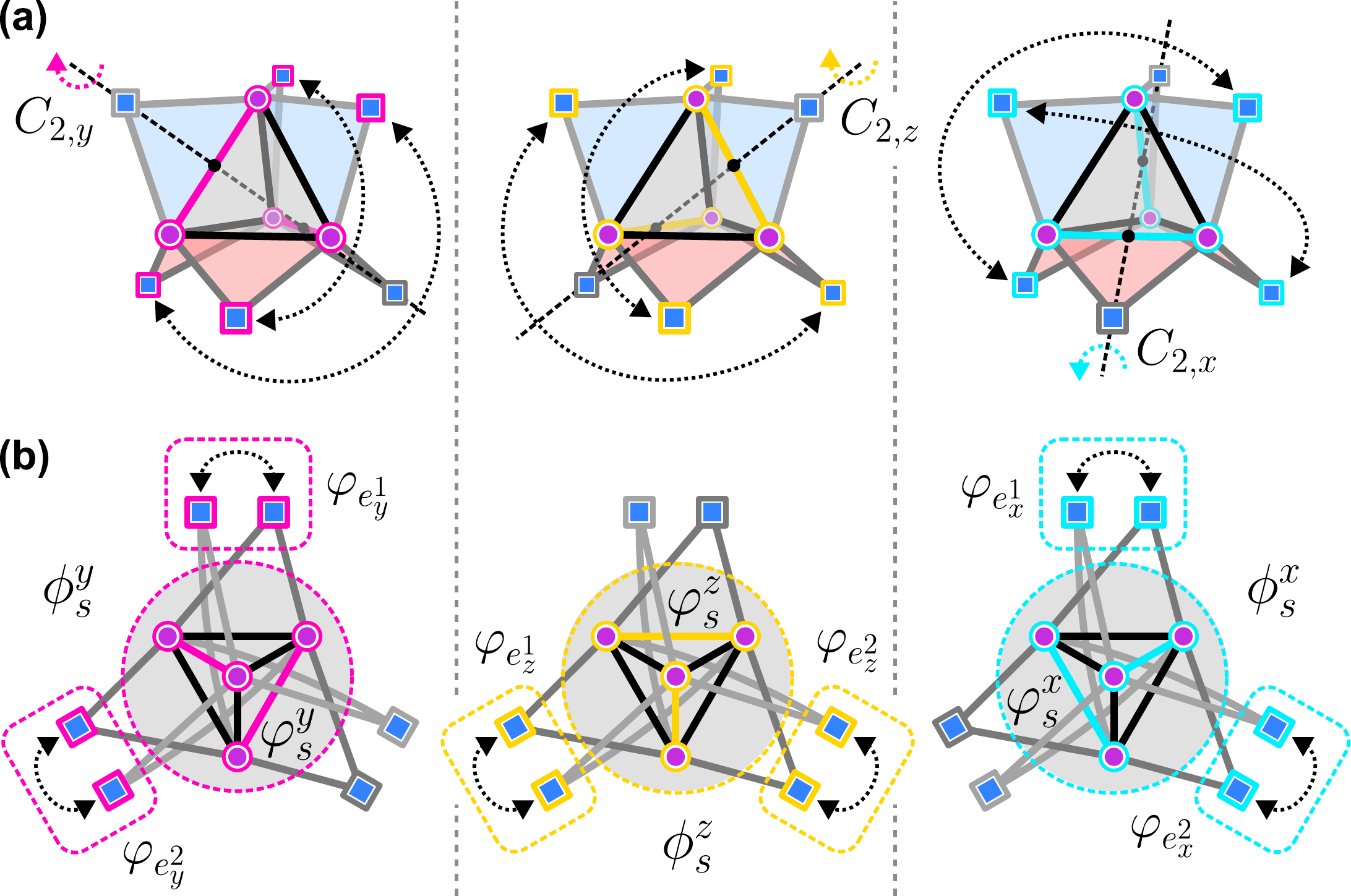}
    \caption{%
    \emph{Automorphisms of an independent \texttt{FSU}-vertex.} 
    (a)~The automorphisms $\phi^\alpha_s$ ($\alpha=x,y,z$) of an independent
    \texttt{FSU}-structure on vertex $s$ can be visualized as $\pi$-rotations
    $C_{2,\alpha}$ of the central tetrahedron about the symmetry axis through
    opposite edges [recall \cref{fig:fsu}~(a)]. These symmetries form the Klein
    four-group $V$ in which $(C_{2,\alpha})^2=1$ and $C_{2,x}\cdot C_{2,y}\cdot
    C_{2,z}=1$.
    (b)~In a flattened (not unit disk) representation, each automorphism
    corresponds to a permutation of ports (squares), indicated by dotted
    arrows, and of two pairs of ancillas (circles) along the colored blockades
    of the central tetrahedron. Note that only atoms with the same detunings
    (fill color) are mapped onto each other, and that one pair of ports always
    remains invariant (gray ports). Here the blockades to ports are colored
    light and dark gray instead of red and blue (as in the text). It is
    convenient to decompose the automorphisms $\phi^\alpha_s$ into edge and
    vertex permutations $\varphi_e$ and $\varphi^\alpha_s$, respectively (see
    text).
    }
    \label{fig:fsu_symmetry}
\end{figure}

We continue with the automorphisms of $\C_\sub{Loop}$ to establish its full
symmetry. As a preliminary step, we focus on an arbitrary
\texttt{FSU}-structure on vertex $s$, including its three adjacent edges (but
ignoring its three neighboring vertex structures for now). As discussed in
\cref{sec:fsu}, this structure has a large automorphism group, and we are
interested in the subgroup that -- in a three-dimensional embedding -- can be
identified with the three rotations about the symmetry axes through opposing
edges of the central tetrahedron [recall \cref{fig:fsu}~(a) and (b)].  We show
these automorphisms and their permutations of vertex and edge atoms explicitly
in \cref{fig:fsu_symmetry}. There we also introduce the decomposition of the
three automorphisms
\begin{align}
    \label{eq:fsu_aut}
    \phi^\alpha_s\equiv\varphi_{e_\alpha^1}\circ\varphi^\alpha_s\circ\varphi_{e_\alpha^2}\,,
    \quad\alpha\in\{x,y,z\}\,,
\end{align}
into permutations $\varphi_e$ that act only on atoms on edge $e$ (ports), and
permutations $\varphi^\alpha_s$ that act only on the tetrahedral atoms
(ancillas) on vertex $s$. Here $e_\alpha^1$ and $e_\alpha^2$ denote the two
edges that are singled out by the choice $\alpha=x,y,z$ (see
\cref{fig:fsu_symmetry}).
Henceforth we denote general \emph{permutations} of atoms by $\varphi$ to
distinguish them from \emph{automorphisms} $\phi$. For example, none of the
three factors in \cref{eq:fsu_aut} is an automorphism of a single
\texttt{FSU}-vertex, only their product is. 

Before we continue with the full tessellation $\C_\sub{Loop}$, let us point out
that the three automorphisms~\eqref{eq:fsu_aut} are reflected in the structure
of the truth table of the loop constraints~\eqref{eq:logic_const}. Indeed, from
Boolean algebra, we know the identities
{\tikzcdset{every label/.append style = {font = \small}}%
    \begin{center}
        \begin{tikzcd}[column sep=1.5cm, row sep=1.2cm, inner sep=20pt,
            every cell/.style={inner sep=8pt}]
            \CR{x}\oplus\CR{y}=\CR{z} 
            \arrow[r, "\phi^x" description, Leftrightarrow] 
            \arrow[d, "\phi^y" description, Leftrightarrow] 
            \arrow[rd, Leftrightarrow] 
        & \CR{x}\oplus\CR{\bar y}=\CR{\bar z} 
        \arrow[d, "\phi^y" description, Leftrightarrow] 
        \arrow[ld, "\phi^z=\phi^x\circ\phi^y" description, Leftrightarrow] 
        \\
        \CR{\bar x}\oplus\CR{y}=\CR{\bar z} 
        \arrow[r, "\phi^x" description, Leftrightarrow] 
        & \CR{\bar x}\oplus\CR{\bar y}=\CR{z}
        \end{tikzcd}
\end{center}}%
\noindent and analogous relations for the \texttt{XNOR}-constraint. This
demonstrates that the four consistent states on the red sublattice
(\texttt{XOR}-constraint) are mapped to each other by the three automorphisms
of the \texttt{FSU}-structure -- which is the hallmark of its full symmetry.

The decomposition~\eqref{eq:fsu_aut} shows that, for each independent
\texttt{FSU}-structure on vertex $s$, there are three nontrivial automorphisms
$\phi^\alpha_s$ ($\alpha=x,y,z$), each of which affects \emph{two} adjacent
edges. Since these edges are \emph{shared} between adjacent vertices in
$\C_\sub{Loop}$, applying a vertex automorphism on one vertex necessitates
another automorphism on two of its neighbors. Therefore the subgroup
$\A_\sub{Loop}\subset\A_{\C_\sub{Loop}}$ of automorphisms of the complete
tessellation that is generated by such transformations also obeys a $\Z_2$-loop
constraint (in addition to the loop constraint on \emph{excitation patterns} in
the ground state!). That is, for every closed loop
\begin{align}
    \ell\equiv s_1\xrightarrow{e_1}s_2\xrightarrow{e_2}s_3\ldots\,,
\end{align}
given as a sequence of visited vertices $s_1,s_2,\ldots$ and traversed edges
$e_1,e_2,\ldots$, there is a corresponding automorphism
$\phi_\ell\in\A_\sub{Loop}$ of the form
\begin{align}
    \label{eq:loopaut}
    \phi_\ell
    &=\varphi^{\alpha_1}_{s_1}\circ\varphi_{e_1}\circ\varphi^{\alpha_2}_{s_2}\circ\varphi_{e_2}\circ \ldots
\end{align}
where $\alpha_i\in\{x,y,z\}$ denotes one of the three possibilities how the
loop visits vertex $s_i$ (recall \cref{fig:fsu_symmetry}). 

Note that all edge permutations commute trivially:
$\varphi_e\circ\varphi_{\tilde e}=\varphi_{\tilde e}\circ\varphi_{e}$; that
also all vertex permutations commute $\varphi_s^\alpha\circ\varphi_{\tilde
s}^\beta=\varphi_{\tilde s}^\beta\circ\varphi_{s}^\alpha$ (even on the same
vertex $s=\tilde s$) follows because the Klein four-group is abelian
(\cref{fig:fsu_symmetry}). Furthermore, both permutations $\varphi_e$ and
$\varphi_s^\alpha$ decay into 2-cycles (= transpositions of disjoint pairs of
atoms) so that $(\varphi_e)^2=\id$ and $(\varphi_s^\alpha)^2=\id$ and thereby
$(\phi_\ell)^2=\id$. In conclusion, $\A_\sub{Loop}$ is an abelian subgroup of
the full automorphism group $\A_{\C_\sub{Loop}}$, and can be interpreted as a
$\Z_2$ vector space with the symmetric difference of loops (= composition of
automorphisms) as addition (also known as \emph{cycle space} in graph
theory~\cite{Diestel2025}).

Pictorially, $\phi_\ell$ swaps all pairs of inverted ports along $\ell$ via
$\varphi_e$, and compensates for these swaps by corresponding vertex
permutations $\varphi^\alpha_s$. The blockade structure of the
\texttt{FSU}-vertex ensures that the result is an isomorphic blockade graph.
An example of such a loop automorphism is depicted in
\cref{fig:fsu_tessellation}~(a). It is important to appreciate that only the
\emph{combined} permutations of all tetrahedral vertex atoms and the inverted
pairs on edges along the closed loop makes this transformation an automorphism
of the blockade graph. In particular, the construction fails if the path of
edges does not close!

The smallest of these loop automorphisms are the \emph{plaquette automorphisms}
on a hexagon $p\equiv(abcde\!f)$, which can be decomposed into a product of
six vertex and six edge permutations:
\begin{align}
    \phi_p &=\includegraphics[width=4.3cm,valign=c]{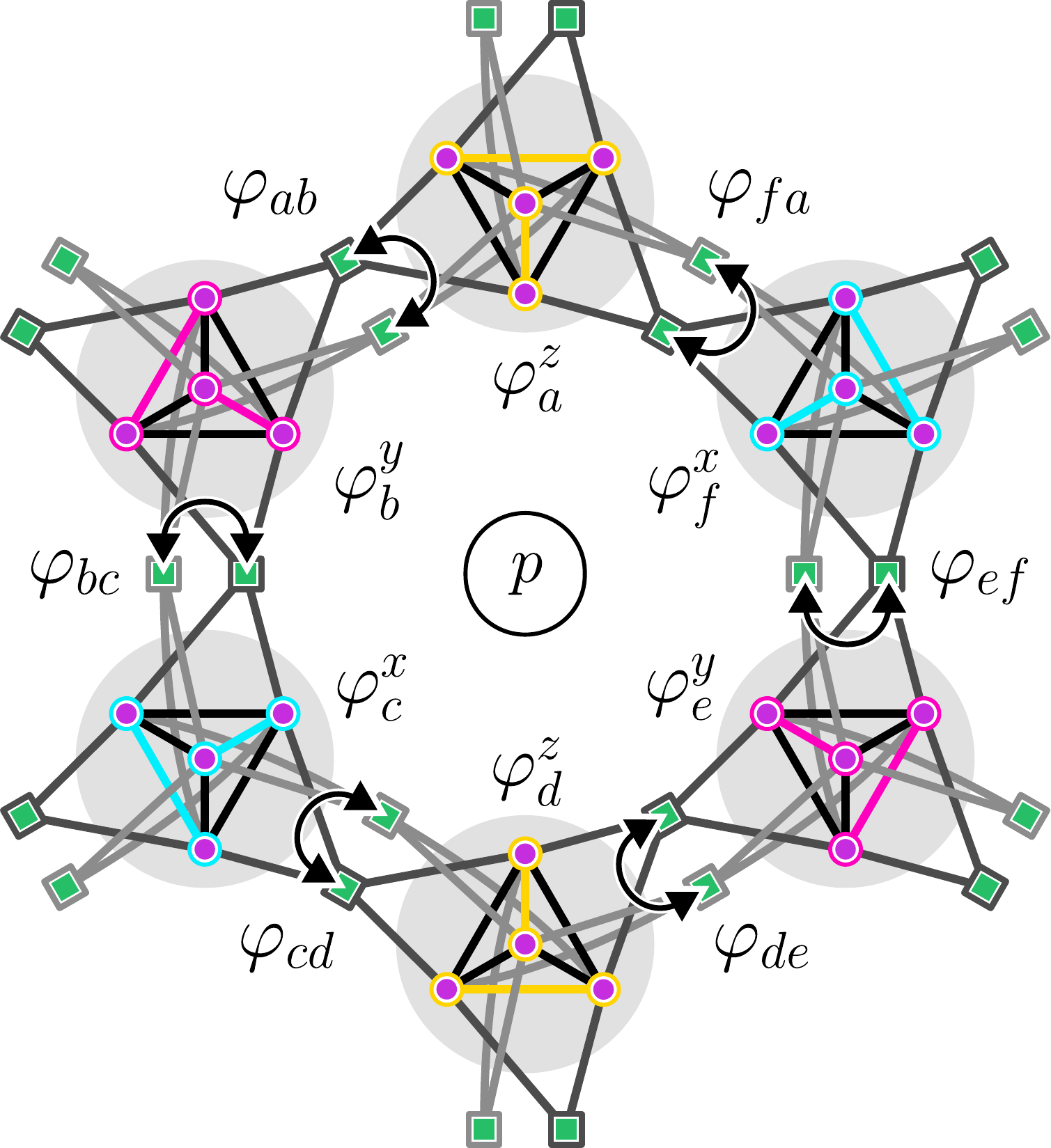}
\end{align}

Multiplication of several such plaquette automorphisms yields loop
automorphisms $\phi_\ell=\prod_{p\in P}\phi_p$, where $P$ denotes a set of
plaquettes that produce the loop $\ell=\del P$ as their boundary $\del P$.
Loops $\ell$ that can be constructed in this way are called \emph{homologically
trivial}. For example, the loop automorphism in \cref{fig:fsu_tessellation}~(a)
can be constructed as product of three elementary plaquette automorphisms.

Note that it is not necessary that $\ell$ is homologically trivial for
$\phi_\ell$ to be an automorphism. For example, with periodic boundaries,
single non-contractible loops around the torus \emph{cannot} be constructed as
boundaries of plaquettes, but still give rise to loop automorphisms that
satisfy the $\Z_2$-loop constraint on every vertex.

The group $\A_\sub{Loop}$ of loop automorphisms is therefore generated by
plaquette automorphisms and homologically nontrivial loop automorphisms (if any
exist), formally:
\begin{align}
    \A_\sub{Loop}\simeq
    \langle\{\phi_p\,|\,\text{Plaquettes $p$}\}\rangle
    \oplus H_1(\Sigma)\,.
\end{align}
Here, $H_1(\Sigma)$ denotes the first homology group (over $\Z_2$) of the
2-manifold $\Sigma$ on which the honeycomb lattice is embedded. For example,
with periodic boundaries, it is $\Sigma=\mathbb{T}$ a torus, and
$H_1(\mathbb{T})\simeq\Z_2\times
\Z_2\simeq\{1,\ell_x,\ell_y,\ell_x\oplus\ell_y\}$ where $\ell_x$ and $\ell_y$
denote two non-contractible loops around the torus. On lattices with
boundaries, $H_1(\Sigma)$ must be replaced by relative homology
groups~\cite{Bravyi1998a}. Here we are not interested in such details and
consider a planar patch of the tessellation with the same type of boundary
everywhere; it is then $H_1(\Sigma)=\{1\}$ as all loops are generated by
elementary plaquettes.

With these insights -- and our convention~\eqref{eq:convention} to identify
strings -- it is now evident that $\varphi_e$ acts on ground state patterns
like $\sigma^x_e$ acts on spins in the toric code: it creates or annihilates a
string on edge $e$. Loop automorphisms therefore create or annihilate extended
loops; in particular, plaquette automorphisms $\phi_p$ act on ground state
patterns by modulo-2 addition of elementary loops; see
\cref{fig:fsu_tessellation}~(b) for an example.
As a consequence, every loop pattern $\vec n\in L_\sub{Loop}\equiv
L_{\C_\sub{Loop}}$ can be constructed from the ``no-loop-configuration'' $\vec
0$ by application of some loop automorphism $\phi\in \A_\sub{Loop}$, i.e.,
$L_\sub{Loop}=\A_\sub{Loop}\cdot\vec 0$; this makes $\C_\sub{Loop}$
fully-symmetric, as intended.

As discussed in \cref{sec:sym}, every automorphism translates to a unitary
symmetry of the blockade Hamiltonian $H_\sub{Loop}\equiv H_{\C_\sub{Loop}}$. We
denote the unitary representation of plaquette automorphisms $\phi_p$ as $U_p$.
Since $\A_\sub{Loop}$ is abelian, these operators form a local, abelian
subgroup of the full symmetry group of the Hamiltonian $H_\sub{Loop}$,
\begin{align}
    \com{U_p}{U_q}=0
    \quad\text{and}\quad
    \com{H_\sub{Loop}}{U_p}=0\,,
    \label{eq:UH}
\end{align}
for all plaquettes $p$ and $q$. Note that this remains true for $\Omega\neq 0$
due to the uniformity of the quantum fluctuations. Since $\phi_p^2=\id$, it
immediately follows that $U_p^2=\id$ and therefore $U_p^\dag=U_p$. It is appropriate
to think of these operators as the analogues of the plaquette operators $B_p$ of
the toric code (at least on the ground state manifold). \cref{eq:UH} will be
crucial for our study of the many-body ground state in \cref{subsec:gs} below.

\subsection{Spatial embedding}
\label{sec:embedding}

\begin{figure*}[tb]
    \centering
    \includegraphics[width=0.9\linewidth]{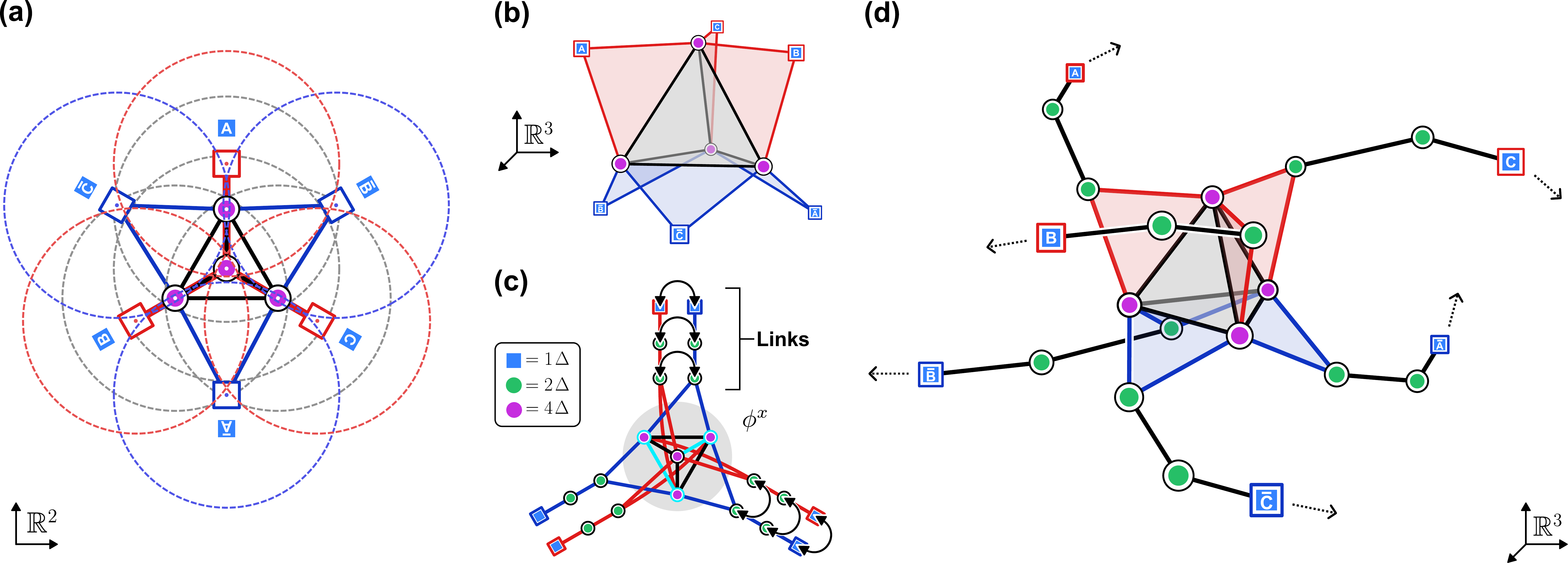}
    \caption{%
    \emph{Embedding of $\C_\sub{Loop}$.} 
    (a)~The \texttt{FSU}-vertex alone is a (tight) unit disk graph in
    $\mathbb{R}^2$; the downside is that the inverted pairs of ports are
    geometrically opposite to each other.
    (b)~A unit ball embedding in $\mathbb{R}^3$ is more natural and more
    robust; it suffers from the same problem, though. (We omit the balls for
    the sake of clarity, the exact coordinates of atoms are listed in
    \cref{app:embedding}.)
    (c)~One can extend the inverted ports of the \texttt{FSU}-vertex by links
    of \emph{equal length} without breaking its blockade graph automorphisms,
    as illustrated exemplarily by the (extended) vertex automorphism $\phi^x$
    (cf.~\cref{fig:fsu_symmetry}).
    (d)~This is useful to ``bend'' the inverted pairs of ports to one side of
    the structure, as needed for the tessellated structure $\C_\sub{Loop}$.
    Indeed, with this trick, one can embed an extended, functionally
    equivalent, tessellated structure $\C_\sub{Loop}^\sub{ext}$ as a unit ball
    graph in $\mathbb{R}^3$. Note that this embedding is unoptimized (in its
    number of atoms and robustness of blockades) and therefore quite
    impractical.
    }
    \label{fig:embedding}
\end{figure*}

Before we turn to the ground state properties of $H_\sub{Loop}$, we briefly
discuss potential embeddings of the loop structure $\C_\sub{Loop}$. We stress
that an optimized embedding that would bring the structure closer to
realization is not (yet) our priority. Here we simply demonstrate that (a
slight modification of) $\C_\sub{Loop}$ can \emph{in principle} be realized as
a quasi-two-dimensional blockade structure of atoms, and leave further
improvements to future studies (like reducing the number of atoms and/or
increasing the robustness of the blockades).

First, note that a single \texttt{FSU}-vertex can -- perhaps surprisingly -- be
embedded as a \emph{unit disk} graph in two dimensions, see
\cref{fig:embedding}~(a). This embedding, however, is rather tight and not very
useful from a practical point of view (where geometrically robust structures
are preferable). Another, more problematic downside is that the inverted pairs
of port atoms that belong to \emph{one} edge of $\C_\sub{Loop}$ are located
\emph{opposite} of each other. It is therefore impossible to tessellate the
plane with this embedding (cf.~\cref{fig:fsu_tessellation}). A \emph{unit ball}
embedding in three dimensions is clearly more natural (and more robust), see
\cref{fig:embedding}~(b). The problem of opposite inverted ports remains,
though.

A crucial insight to solve this problem is illustrated in
\cref{fig:embedding}~(c): One can ``extend'' the ports of the
\texttt{FSU}-vertex by attaching arbitrarily long \emph{link structures} (see
also Ref.~\cite{Stastny2023a}) -- as long as the length of links that carry
inverted signals is equal. It is easy to see that this modification preserves
the relevant automorphisms of the structure, so that an extended tessellation
$\C_\sub{Loop}^\sub{ext}$ with more atoms (but not more degrees of freedom) per
edge can be constructed. While the extension of a single vertex also works with
strictly two-dimensional unit disk embeddings, it is unclear whether many
vertices can be tessellated and amalgamated without breaking the local
automorphisms. The problem that arises in such a planar setup is that the
opposing, inverted links must somehow end up on the same edge of the lattice.
This can only be achieved by crossing structures [recall
\cref{fig:motivation}~(b)]. However, these crossings ``entangle'' the blockade
graphs of links from \emph{different} edges, and therefore break the plaquette
automorphisms. We suspect that a purely two-dimensional construction might be
impossible, though we were not able to prove this.

Fortunately, these problems can be circumvented in three dimensions, where
links that belong to different edges can simply pass each other without
blockades, see \cref{fig:embedding}~(d). Thus, it is indeed possible to realize
a suitably extended structure $\C_\sub{Loop}^\sub{ext}$ (which is functionally
equivalent to $\C_\sub{Loop}$) as a unit ball graph in three dimensions. Note
that the extension of this structure in the third dimension is bounded and does
not scale with the lattice size, hence we refer to this embedding as
\emph{quasi-two-dimensional}. It seems reasonable to expect that such an
extended structure $\C_\sub{Loop}^\sub{ext}$ is not a unit disk graph (and
therefore cannot be realized in strictly two dimensions); however, we were not
able to prove this rigorously. (Note that deciding whether a given graph is a
unit disk graph is \NP-hard~\cite{Breu1998,Kuhn2004,McDiarmid2013}. That is, an
efficient proof must exploit the particular structure of the graph.)

In summary, the blockade structure in \cref{fig:fsu_tessellation} (or a slight
generalization thereof) \emph{can} be realized as a blockade structure in three
dimensions. However, finding the simplest, experimentally most accessible
realization remains an open question. In the remainder of the paper, we focus
again on the unmodified $\C_\sub{Loop}$ and ignore its embedding (all results
are also valid for embeddable extensions $\C_\sub{Loop}^\sub{ext}$).

\subsection{Ground state properties}
\label{subsec:gs}

In this last section, we study the ground state properties of the blockade
Hamiltonian $H_\sub{Loop}=H_\sub{Loop}(\Omega)$ for nonzero quantum
fluctuations $\Omega\neq 0$. Here we prove our main result, namely that its
many-body ground state $\ket{\Omega}$ is in the same quantum phase as the loop
condensate $\ket{\Omega_\sub{TC}}$ of the toric code. This result is rigorous
and does not rely on numerical evidence.

\subsubsection{General features}
\label{subsubsec:general}

Let us consider a finite patch of the tessellation $\C_\sub{Loop}$ in
\cref{fig:fsu_tessellation} with or without periodic boundaries. As detailed in
\cref{sec:localz2}, $\C_\sub{Loop}$ is fully-symmetric so that \cref{prop:1}
applies. Namely, it guarantees that the many-body ground state $\ket{\Omega}$
of $H_\sub{Loop}$ for $\Omega\neq 0$ is unique and given by an equal-weight
superposition of loop states from $\H_\sub{Loop}\equiv\H_{\C_\sub{Loop}}$,
dressed by excited states that are suppressed in powers of $\Omega/\Delta E$,
recall \cref{eq:prop_omega}:
\begin{align}
    \label{eq:toric_code_state}
    \ket{\Omega}
    =
    \underbrace{\lambda(\Omega)\mkern-10mu\sum_{\vec n\in L_\sub{Loop}}\mkern-10mu\ket{\vec n}}_{=:\Lambda\ket{\Omega_0}}
    +
    \underbrace{%
        \sum_{d\geq 1}\left(\tfrac{\Omega}{\Delta E}\right)^d\mkern-10mu
        \sum_{\vec n\in L_{\sub{Loop}}^d}\mkern-10mu
        \eta_{\vec{n}}(\Omega)\ket{\vec n}
    }_{=:\sqrt{1-\Lambda^2}\ket{\Omega_\eta}}
\end{align}
with $\BraKet{\Omega_0}{\Omega_\eta}=0$.

Note that the form~\eqref{eq:toric_code_state} is true for arbitrary
$\Omega\neq 0$, and not only for small perturbations. This already indicates
that one \emph{cannot} deduce from \cref{eq:toric_code_state} the presence of
topological order -- despite the suggestive equal-weight superposition of loop
states in $\ket{\Omega_0}$! Furthermore, the \emph{uniqueness} of
$\ket{\Omega}$ is a generic feature of finite-size blockade Hamiltonians like
$H_\sub{Loop}$ for $\Omega\neq 0$; in particular, it is independent of their
boundary conditions. Since we expect $H_\sub{Loop}$ to feature topological
ground state degeneracies for periodic boundaries, we must conclude that this
uniqueness can be due to finite-size gaps that vanish exponentially fast in the
thermodynamic limit. An immediate corollary of \cref{prop:1} is then that no
blockade Hamiltonian of the form~\eqref{eq:H} can realize the renormalization
fixed point~\eqref{eq:tc_h} of the toric code (as the latter has zero
correlation length and therefore perfect topological degeneracy on finite-size
systems).

Next, because of \cref{eq:UH} and $U_p^2=\id$ and $U_p^\dag=U_p$, we can
label the many-body eigenstates of $H_\sub{Loop}=H_\sub{Loop}(\Omega)$,
\begin{align}
    \label{eq:eigenstates_E}
    H_\sub{Loop}\ket{E^{\vec\xi}_\Omega,\vec\xi}&=E^{\vec\xi}_\Omega\ket{E^{\vec\xi}_\Omega,\vec\xi}\,,
\end{align}
by a sign $(-1)^{\xi_p}=\pm 1$ per plaquette:
\begin{align}
    \label{eq:eigenstates_U}
    U_p\ket{E^{\vec\xi}_\Omega,\vec\xi}&=(-1)^{\xi_p}\ket{E^{\vec\xi}_\Omega,\vec\xi}\,.
\end{align}
We refer to $\xi_p\in\{0,1\}$ as the ($\Z_2$-)\emph{flux} through plaquette
$p$, in analogy to the toric code. If we set the ground state energy to
$E^{\vec\xi^0}_\Omega=0$, we have in particular
$\ket{\Omega}=\ket{0,\vec\xi^0}$. The flux sector $\vec{\xi}^0$ follows from
\cref{prop:1}, which states that $U_p \ket{\Omega} = \ket{\Omega}$ for all
plaquette operators $U_p$; together with \cref{eq:eigenstates_U} this implies
that $\vec\xi^0=\vec 0$.

In summary, the ground state $\ket{\Omega}=\ket{0,\vec 0}$ of $H_\sub{Loop}$ is
in the \emph{flux-free sector}, again in analogy to the toric code where
$B_p\ket{\Omega_\sub{TC}}=\ket{\Omega_\sub{TC}}$. Moreover, due to
\cref{eq:eigenstates_U}, the full many-body spectrum of $H_\sub{Loop}$ splits
into \emph{flux sectors} labeled by flux patterns $\vec\xi$. With this
knowledge, we turn now to the question of topological order in $\ket{\Omega}$.

\subsubsection{Topological order}
\label{subsubsec:topo}

It is obvious that $\ket{\Omega_0}$ is in the toric code phase, as it is
locally unitary equivalent to the loop condensate~\eqref{eq:tc_gs}. (Remember
that the additional ancillas on the vertices do not contribute degrees of
freedom to the states in $\H_\sub{Loop}$.) What remains to be shown is that
this order is stable under the addition of $\ket{\Omega_\eta}$ for small but
finite fluctuations $\Omega\neq 0$.

\begin{figure}[tb]
    \centering
    \includegraphics[width=1.0\linewidth]{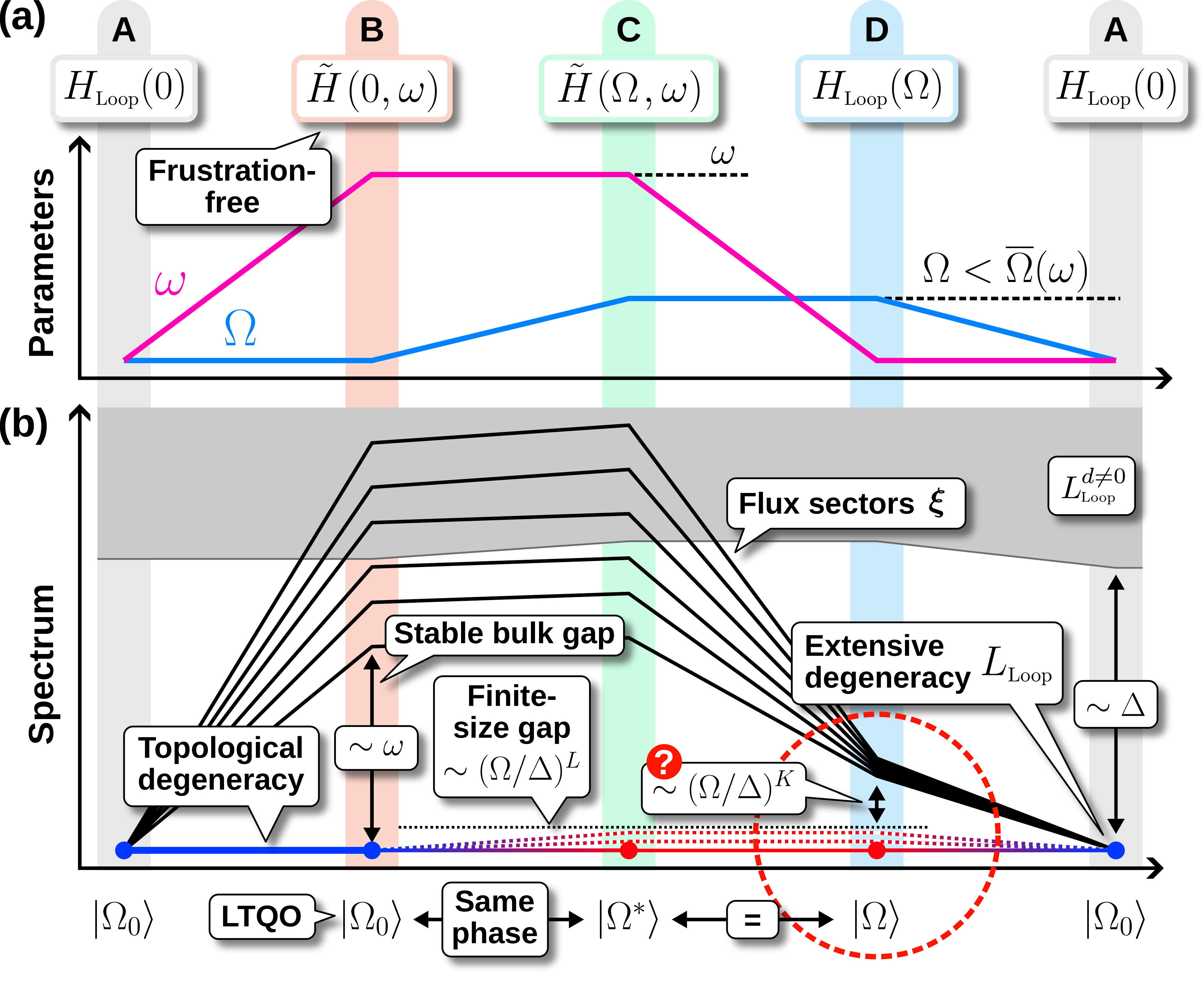}
    \caption{%
    \emph{Spectral properties of $\tilde H(\Omega,\omega)$.} 
    (a)~Family of Hamiltonians $\tilde H(\Omega,\omega)$ parametrized by
    $\Omega$ and $\omega$. Notably, $\tilde H(0,\omega)$ is frustration-free,
    an important ingredient for gap stability. 
    (b)~Schematic spectrum of $\tilde H(\Omega,\omega)$ along this path. For
    $\Omega=0=\omega$ (\textbf{A}), $\tilde H(0,0)=H_\sub{Loop}(0)$ is purely
    classical with exponentially degenerate ground state space $\H_\sub{Loop}$.
    For $\Omega=0<\omega$ (\textbf{B}), $\tilde H(0,\omega)$ is
    frustration-free and has a fixed point ground state $\ket{\Omega_0}$ with
    local topological quantum order (LTQO). Depending on boundary conditions,
    there can be a (non-extensive) topological ground state degeneracy. The
    bulk gap of this Hamiltonian is stable against arbitrary, weak, local
    perturbations, so that for $0<|\Omega|<\ol\Omega(\omega)$ (\textbf{C}),
    $\tilde H(\Omega,\omega)$ is still gapped, with a new ground state
    $\ket{\Omega^*}$ in the same quantum phase. Switching off $\omega\searrow
    0$ (\textbf{D}) leads to $H_\sub{Loop}(\Omega)$ with the ground state
    $\ket{\Omega}$. Due to symmetry, it is $\ket{\Omega^*}=\ket{\Omega}$, which
    allows us to characterize the quantum phase of $\ket{\Omega}$. We are
    ultimately interested in the physics marked by the dashed circle. The
    question mark indicates that the bulk gap at \textbf{D} remains unproven.
    }
    \label{fig:spectrum}
\end{figure}

Since we do not know how to solve $H_\sub{Loop}$ exactly, we must take a detour
to solve this problem. To this end, we introduce the auxiliary Hamiltonian
\begin{align}%
    \tilde H(\Omega,\omega):=H_\sub{Loop}(\Omega)+\frac{\omega}{2}\sum_p (\id-U_p)\,,
    \label{eq:Hext}
\end{align}%
where we artificially add (experimentally unrealistic) plaquette fluctuations
of loops with strength $\omega$. We stress that such fluctuations are already
\emph{perturbatively} present in $H_\sub{Loop}(\Omega)$ alone, although of much
weaker, unknown strength. The new term in \cref{eq:Hext} penalizes each flux
($\xi_p=1$) by an additional energy $\omega$. Our goal is to use $\tilde
H(\Omega,\omega)$ to find a parametric path from a system with an exactly known,
topologically ordered ground state to $H_\sub{Loop}(\Omega)=\tilde H(\Omega,0)$
with ground state
$\ket{\Omega}=\Lambda\ket{\Omega_0}+\sqrt{1-\Lambda^2}\ket{\Omega_\eta}$, such
that the gap remains open along the path (for finite but sufficiently small
$\Omega$); this would establish the topological order of $\ket{\Omega}$ and
demonstrate that $\ket{\Omega_\eta}$ does not modify the pattern of long-range
entanglement~\cite{Chen2010}.

To do so, we proceed along the three steps sketched in \cref{fig:spectrum} (for
mathematical details, see \cref{app:spectrum}):
\begin{itemize}

    \item\textbf{Step 1: A~\textrightarrow~B}
    
    We start with the classical Hamiltonian $H_\sub{Loop}(0)=\tilde H(0,0)$
    (\textbf{A}). By construction, this Hamiltonian has an extensively
    degenerate ground state manifold $\H_\sub{Loop}$ spanned by the loop
    configurations $L_\sub{Loop}$; its gap $\sim\Delta$ is given by the
    detunings.

    We now ramp up the artificial plaquette fluctuations $0<\omega <\Delta$
    (with $\omega\sim\Delta$) while keeping $\Omega =0$. Since $[{U_p},{\tilde
    H(0,\omega)}]=0$, the spectrum of $\tilde H(0,\omega)$ also decomposes into
    flux sectors. In particular, ground state(s) belong to the zero-flux sector
    within $\H_\sub{Loop}$. Their construction is analogous to the toric code
    (\cref{subsec:toric}), and for open, uniform boundaries, one obtains the
    equal-weight loop condensate
    \begin{align}
        \label{eq:toric_code_state_exact} 
        \ket{\Omega_0}\propto\sum_{\vec n\in L_\sub{Loop}}\ket{\vec n}\,.
    \end{align}
    For homologically nontrivial boundaries (like periodic boundaries), there
    is one independent equal-weight superposition for each element of
    $H_1(\Sigma)$ and these are perfectly degenerate. (On the torus, for
    example, there would be four perfectly degenerate ground states.)

    The ground state manifold of $\tilde H(0,\omega)$ is therefore no longer
    extensively degenerate and within the topologically ordered toric code
    phase. The gap of $\tilde H(0,\omega)$ is of order $\omega$ and determines
    the energy of elementary flux excitations (\textbf{B}).

    \item\textbf{Step 2: B~\textrightarrow~C}

    Next, we would like to switch on $\Omega\neq 0$ to reach $\tilde
    H(\Omega,\omega)$. In general, this will modify the spectrum, eigenstates,
    and degeneracy of the ground state manifold. In particular, the ground
    state itself changes,
    \begin{align}
        \text{\textbf{B}}\quad
        \ket{\Omega_0}\xrightarrow[\omega>0]{\Omega=0\;\to\;\Omega\neq 0}\ket{\Omega^\ast}
        \quad\text{\textbf{C}}
    \end{align}
    in an unknown way [as we can no longer diagonalize $\tilde
    H(\Omega,\omega)$]. The crucial insight is that, as long as we can
    guarantee that the gap of $\tilde H(0,\omega)$ does not close when going to
    $\tilde H(\Omega,\omega)$, we can conclude that the new ground state
    $\ket{\Omega^\ast}$ belongs to the same quantum phase as
    $\ket{\Omega_0}$~\cite{Chen2010}, i.e., it is still in the toric code phase.

    Thus, the problem reduces to the stability of the gap of $\tilde
    H(0,\omega)$ under weak local fluctuations $0<|\Omega|\ll\omega$. We stress
    that in general, bulk gaps are \emph{not} stable under arbitrary local
    perturbations -- even if the latter are small compared to the
    gap~\cite{Bravyi2010}.

    Fortunately, it can be shown rigorously~\cite{Michalakis2013} that a gap
    \emph{is} stable under arbitrary local, weak perturbations, given the
    unperturbed Hamiltonian [here: $\tilde H(0,\omega)$] is
    \emph{frustration-free}, and its ground state(s) [here: $\ket{\Omega_0}$]
    are \emph{topologically ordered}. (More precisely, $\tilde{H}(0,\omega)$ 
    must satisfy a
    condition called \emph{local topological quantum order}
    [LTQO]~\cite{Bravyi2010,Bravyi2011,Michalakis2013}.) In \cref{app:spectrum}
    we explicitly show that $\tilde H(0,\omega)$ and $\ket{\Omega_0}$ indeed
    satisfy these conditions. Note that in the perturbed system $\tilde
    H(\Omega,\omega)$, ground state degeneracies can be lifted, but these are
    provably finite-size and vanish exponentially with system
    size~\cite{Michalakis2013}.

    In conclusion, we can switch on weak quantum fluctuations $0<|\Omega| <
    \bar\Omega(\omega)$ without closing the bulk gap. Here,
    $\bar\Omega(\omega)$ denotes some $\omega$-dependent upper bound that
    guarantees the stability of the gap~\cite{Michalakis2013}. This leads us to
    the new Hamiltonian $\tilde H(\Omega,\omega)$ that is provably in the same
    gapped quantum phase as the toric code -- though with a new, no longer
    exactly known ground state $\ket{\Omega^*}$.

    \item\textbf{Step 3: C~\textrightarrow~D}

    In the last step, we would like to switch off the auxiliary fluctuations to
    transition from $\tilde H(\Omega,\omega)$ to $\tilde
    H(\Omega,0)=H_\sub{Loop}(\Omega)$. Generally, we would again expect a
    modification of the ground state:
    \begin{align}
        \text{\textbf{C}}\quad
        \ket{\Omega^\ast}\xrightarrow[\Omega\neq 0]{\omega>0\;\to\;\omega=0}\ket{\Omega}
        \quad\text{\textbf{D}}
    \end{align}
    The problem is that we cannot consider $\omega\sim\Delta \gg |\Omega|$ as a
    small perturbation. Moreover, $\tilde H(\Omega,\omega)$ is \emph{not}
    frustration-free, so that no gap stability can be inferred. In particular,
    we cannot conclude that $H_\sub{Loop}(\Omega)$ has a bulk gap (see
    \cref{subsubsec:gap} below).

    This is where the local $\Z_2$ symmetries of $H_\sub{Loop}(\Omega)$ pay
    off: The eigenbasis $\ket{E^{\vec\xi}_\Omega,\vec \xi}$ of
    $H_\sub{Loop}(\Omega)$ is also an eigenbasis of $\tilde H(\Omega,\omega)$.
    Indeed, \cref{eq:eigenstates_E,eq:eigenstates_U,eq:Hext} imply
    \begin{align}
        \label{eq:Ht_eig}
        \tilde H(\Omega,\omega)\ket{E^{\vec\xi}_\Omega,\vec \xi}
        =
        (E^{\vec\xi}_\Omega+\omega|\vec\xi|)\ket{E^{\vec\xi}_\Omega,\vec \xi}
    \end{align}
    with $|\vec\xi|=\sum_p\xi_p$ the total number of plaquettes with nonzero
    flux. The only effect of the $\omega$-term in \cref{eq:Hext} is to shift
    the eigenenergies of eigenstates in different flux sectors, without
    modifying the eigenstates themselves.

    In \cref{subsubsec:general} we showed that the ground state
    $\ket{\Omega}=\ket{0,\vec 0}$ is unique and in the flux-free sector of
    $H_\sub{Loop}(\Omega)$. Hence, \cref{eq:Ht_eig} implies that $\tilde
    H(\Omega,\omega)\ket{\Omega}=0$, but also that
    $\min{(E^{\vec\xi}_\Omega+\omega|\vec\xi|)}\stackrel{\omega>0}{=}\min{(E^{\vec\xi}_\Omega)}=0$,
    i.e., $\ket{\Omega}$ must be in the ground state manifold of $\tilde
    H(\Omega,\omega)$. Since the ground state manifold of
    $H_\sub{Loop}(\Omega)$ is non-degenerate, and eigenenergies can only remain
    constant or decrease for $\omega\searrow 0$, the ground state manifold of
    $\tilde H(\Omega,\omega)$ must also be non-degenerate. Hence we have shown
    that actually
    \begin{align}
        \text{\textbf{C}}\quad
        \ket{\Omega^\ast} = \ket{\Omega}
        \quad\text{\textbf{D}}\,,
    \end{align}
    which finally proves that $\ket{\Omega}$ is in the topologically ordered
    quantum phase of the toric code.

\end{itemize}

In conclusion, we managed to transfer our knowledge about the (fixed point)
topological order of $\ket{\Omega_0}$ to the unknown state $\ket{\Omega}$,
\begin{align}
    \lefteqn{\overbrace{\phantom{\ket{\Omega_0}\;\sim\;\ket{\Omega^\ast}}}^{\text{Same phase}}}
    \ket{\Omega_0}\;\sim\;
    \underbrace{\ket{\Omega^\ast}\;=\;\ket{\Omega}}_\text{Same state}
    \;=\;
    \Lambda\ket{\Omega_0}+\sqrt{1-\Lambda^2}\ket{\Omega_\eta}\,,
\end{align}
where the addition of $\ket{\Omega_\eta}$ perturbs $\ket{\Omega_0}$ away from
the fixed point to $\ket{\Omega}$, without leaving the toric codes phase for
small but finite, uniform quantum fluctuations
$0<|\Omega|<\bar\Omega(\omega\sim\Delta)$. We achieved this by exploiting the
gap stability of frustration-free, topologically ordered systems, in
conjunction with the local $\Z_2$ symmetry of the Hamiltonian $H_\sub{Loop}$ of
the tessellated blockade structure~$\C_\sub{Loop}$.

\subsubsection{Bulk gap}
\label{subsubsec:gap}

We have now rigorously established that the tessellated blockade structure
$\C_\sub{Loop}$ has a (potentially degenerate) topologically ordered ground
state manifold for finite but small fluctuations $\Omega\neq 0$. However, we
have \emph{not} shown that the system has a finite spectral gap in the
thermodynamic limit. Proving the existence or absence of spectral gaps of
interacting quantum many-body systems is notoriously
difficult~\cite{Lieb1961,Cubitt2015,Kastoryano2018}. Here we conclude our
discussion with arguments that might \emph{suggest} (or help to establish) the
existence of a bulk gap.

Firstly, one can show (\cref{app:Inv_Schriffer_Wolff}) that the large symmetry
group $\A_{\sub{Loop}}$ severely restricts the form of the \emph{effective}
Hamiltonian on the low-energy subspace $\H_\sub{Loop}$. A local
Schrieffer-Wolff transformation~\cite{Bravyi_2011,Datta1996} of
$H_\sub{Loop}(\Omega)$ yields 
\begin{align}%
    \label{eq:eff}
    H_\sub{Loop}^\sub{eff}=
    -J_\sub{eff}\tfrac{\Omega^K}{\Delta E^{K-1}}\sum_{\text{Faces $p$}} U_p
    +\O\left(\tfrac{\Omega^{K+2}}{\Delta E^{K+1}}\right)\,,
\end{align}%
up to a constant shift in energy. The order $K$ depends on the specific
implementation of the structure (e.g., the length of links connecting
vertices). If the effective coupling does not vanish, it is expected to be
positive $J_\sub{eff}>0$, so that the Hamiltonian~\eqref{eq:eff} is consistent
with the zero-flux ground state~\eqref{eq:toric_code_state}. In this case,
$H_\sub{Loop}^\sub{eff}$ has a finite gap of order
$J_\sub{eff}{\Omega^K}/{\Delta E}^{K-1}$. However, this result does not
\emph{prove} a finite gap of $H_\sub{Loop}(\Omega)$ as there is no guarantee
that the perturbation series converges. Furthermore, for a local
Schrieffer-Wolff transformation, we were not able to exclude the possibility 
that cancellations
in the perturbation series pin the effective coupling $J_\sub{eff}$ to zero. It
is worth mentioning, though, that for a global Schrieffer-Wolff transformation
on a finite system (say, a single plaquette), one can show rigorously that no
cancelations occur and $J_\sub{eff}$ is finite.

Secondly, it is well-known that ground states of local, gapped Hamiltonians
have only short-range correlations~\cite{Hastings2006}. Conversely, common
wisdom associates gapless systems with algebraically decaying correlations. Our
arguments in \cref{subsubsec:topo} prove that $\ket{\Omega}$ is indeed
short-range correlated, as it is the ground state of the gapped Hamiltonian
$\tilde H(\Omega,\omega)$. Unfortunately, without additional input, one cannot
go the other way and infer the existence of a gap from a short-range correlated
ground state, because gapless excitations may not couple to local
operators~\cite{FernandezGonzalez2012,FernandezGonzalez2014}.

Lastly, the (presumed) gap induced by the quantum fluctuations $\Omega$
separates the zero-flux sector of the ground state from higher flux sectors
(\cref{fig:spectrum} \textbf{D}); it therefore corresponds to the ``mass'' of
flux excitations -- which are non-dynamical due to the local $\Z_2$ symmetry.
As such, it is analogous to the mass of static charge pairs (mesons) in a pure
$\ZZ$ lattice gauge theory in two
dimensions~\cite{Wegner1971,Fradkin1978,Fradkin1979}. Via well-known duality
transformations, the energy gap can then be reinterpreted as the difference
between ground state energies of the two-dimensional transverse-field Ising
model with and without antiferromagnetic defect lines. This suggests that
methods and results from these related fields might provide further insight
into the spectral properties of the blockade structure~$\C_\sub{Loop}$.

In conclusion, there is anecdotal evidence for a finite bulk gap
$\sim\Omega^K/\Delta^{K-1}$ that separates the zero-flux sector of
$\ket{\Omega}$ from states with flux excitations. However, a rigorous proof of
its existence (or absence) is missing.

\section{Outlook \& Comments}
\label{sec:outlook}

There are several loose ends throughout the paper, and a few obvious follow-up
questions that need more work:

\textbf{Spectral gap.} The question whether our proposed blockade Hamiltonian
has a bulk gap remains open. It would be interesting to determine whether any 
of the (few)
known methods to rigorously establish the existence of a bulk gap applies to
our system. To strengthen (or debunk) our anecdotal evidence for a gap, one
could also resort to numerical techniques like DMRG. However, this approach
is complicated by the 36 atoms required to implement a single plaquette of the
honeycomb lattice. While the blockade interaction reduces the number of states
in the Hilbert space significantly, it hardly compensates for the exponential
resource requirements due to the rather large number of atoms in a unit cell.

\textbf{Planar tessellations.} It remains unclear whether our proposed
tessellated blockade graph necessarily requires three spatial dimensions to be
realized as a blockade structure. While it seems unlikely that a functionally
equivalent planar embedding exists, we are not aware of a rigorous argument to
support this claim. We also showed that none of the previously known
tessellations we checked feature local blockade graph automorphisms. Since all
these models are realizable as unit disk graphs, one might conjecture that the
third dimension is indeed \emph{needed} for a fully-symmetric tessellated
structure (at least if they implement non-factorizable ground state manifolds).
A counterexample to this conjecture would be a strictly planar, fully-symmetric
tessellation that is functionally equivalent to our quasi-two-dimensional
model.

\textbf{Other \texttt{XOR}/\texttt{FSU}-based tessellations.} In this work, we
employed the \texttt{XOR}-constraint -- realized by the \texttt{FSU}-structure
-- to implement a periodic system with local $\mathbb{Z}_2$-loop constraints.
However, \texttt{XOR}-constraints are more versatile than that, as they simply
realize a local $\sigma^z_i\sigma^z_j\sigma^z_k$ interaction. Using a
fully-symmetric version of a copy-structure~\cite{Stastny2023a}, it is possible
to connect \texttt{FSU}-structures by double-links (``ladders'') in
functionally different ways, without breaking the full symmetry of the
tessellation. For instance, it is straightforward to construct a blockade
version of the \emph{triangular-plaquette model} (also known as
\emph{Newman-Moore model}~\cite{Newman1999}), where the full symmetry is
realized by extensive ``fractal'' automorphisms on a triangular
lattice~\cite{Devakul2019}; the quantum fluctuations then realize (a low-energy
version of) the \emph{quantum} triangular-plaquette model, see
e.g.~Ref.~\cite{Sfairopoulos2024} and the references therein.

\textbf{Fully-symmetric circuits.} It would be interesting to explore how
generic the paradigm of fully-symmetric blockade structures actually is. For
example, is there a toolbox that allows for the amalgamation of fully-symmetric
elementary building blocks and can it be used to design fully-symmetric structures
for any prescribed set of ground state patterns? This would be potentially
useful for applications like geometric programming~\cite{Nguyen2022,Wurtz2022}
where optimization problems are encoded in blockade structures and the solution
relies on strong fluctuations over a vast search space. For example, one could
start with the \texttt{FSU}-gate and use double-links (``ladders'') to relay
Boolean variables, and then try to bootstrap a framework akin to
\emph{differential signaling} and \emph{balanced circuits} known from
electrical engineering.

\textbf{Alternative platforms.} As stressed previously, our results apply to
any framework with controllable two-level systems that interact via a simple
blockade potential. While neutral atom platforms that leverage Rydberg states
are currently the most promising candidate for the implementation of such
systems, it would be interesting to explore alternative platforms as
well~\cite{Hecktter_2021, Menta_2025}. For example, if one could implement
blockade graphs \emph{directly} (e.g.~using techniques from superconducting
circuit QED~\cite{Carusotto2020}), one could evade embeddability issues that
arise from the geometric nature of the van der Waals interaction between atoms
(see the next paragraph). 


\textbf{Van der Waals interactions.} In this work, we focused on ideal blockade
Hamiltonians that describe abstract two-level systems. Depending on the
experimental platform used to implement these models, platform-specific
modifications of the Hamiltonian must be taken into account. For instance, the
Rydberg blockade only \emph{approximates} a sharply decaying van der Waals
interaction. It is nontrivial but crucial for experiments on this platform to
study the effects of the long-range tails of these interactions on the proposed
blockade structures~\cite{Rourke2023}. For example, such interactions typically
break the symmetries provided by blockade graph automorphisms. It is reasonable
to expect that these symmetries can survive if they can be realized as
\emph{Euclidean symmetries} of the spatial structure. While this is often
possible for small, finite structures like a single \texttt{FSU}-gate, it
cannot be achieved for every structure (see \cref{app:counterexample}); this
includes in particular tessellations like $\C_\sub{Loop}$ with extensively
generated automorphism groups.  Generally speaking, the design of symmetric
blockade structures with long-range interactions is more involved (and more
restrictive) because geometry enters the stage. However, \emph{if} blockade
graph automorphism can be realized as isometries of Euclidean space, then
\cref{prop:1} can be readily adapted to van der Waals interactions (see
\cref{app:symmetries,app:proof_1}).

\section{Summary}
\label{sec:summary}

In this paper, we studied the effect of uniform quantum fluctuations on
tailored structures of two-level systems (atoms, spins, cavity modes etc.) that
interact via a blockade potential (so called \emph{blockade structures}). There
are three main results to highlight:

First, we introduced the concept of \emph{blockade graph automorphisms} to
describe (not necessarily geometric) symmetries of strongly interacting
blockade Hamiltonians with quantum fluctuations. We then defined
\emph{fully-symmetric} blockade structures as those where these symmetries act
transitively on the set of excitation patterns in the classical ground state
manifold. Our first main contribution is the rigorous proof that uniform
quantum fluctuations of fully-symmetric structures stabilize a \emph{unique}
ground state with \emph{equal-weight} contributions from each ground state of
the classical system. This result provides a versatile tool to engineer
blockade structures that experience strong quantum fluctuations by design.

Second, to illustrate these concepts, we constructed the fully-symmetric
universal (\texttt{FSU}) gate, a three-dimensional structure of 10 two-level
systems that is fully-symmetric and realizes various Boolean gates at once.
Among these are \texttt{XOR}- and \texttt{XNOR}-constraints between pairs of
two-level systems that are forced into inverted states within the ground state
manifold. We then showed step-by-step how this \texttt{FSU}-structure can be
used to construct a quasi-two-dimensional periodic tessellation of two-level
systems on the honeycomb lattice that enforces a $\Z_2$ Gauss law on its
vertices and features local $\Z_2$ symmetries on its plaquettes. The action of
these symmetries on the ``loop states'' (which satisfy the Gauss law) makes
this blockade structure fully-symmetric, and therefore stabilizes an
equal-weight superposition of loop states as its many-body ground state. This
loop condensate is reminiscent of the toric code, and suggests that the proposed
tessellation stabilizes topological order. This construction is our second main
contribution: It both showcases a highly nontrivial application of the concept
of full symmetry, and thereby results in the first known blockade structure
that provably stabilizes an equal-weight loop condensate.

Finally, we showed how prior results on the gap stability of topologically
ordered systems, in conjunction with the local $\Z_2$ symmetry of our system,
can be exploited to draw conclusions about the topological order of its
many-body ground state. In particular, we proved that the ground state of our
tessellated blockade structure belongs to the quantum phase of the toric code.
This proof is our third main contribution: It showcases a technique that allows
access to rigorous results on strongly interacting, tessellated blockade
structures -- without full knowledge of the eigenstates -- to establish the
topological order of the ground state. To the best of our knowledge, this is
the first system based on two-body blockade interactions where this is known
rigorously, without relying on numerical evidence or perturbative techniques.

Combined, these results extend the toolbox for the bottom-up design of tailored
quantum matter on platforms that feature simple two-body blockade interactions.


\begin{acknowledgments}
    We acknowledge funding from the Federal Ministry of Education and Research
    (BMBF) under the grants QRyd\-Demo, MUNIQC-Atoms, and the Horizon Europe
    programme \texttt{HORIZON-CL4-2021-DIGITAL-EMERGING-01-30} via the project
    \texttt{101070144} (EuRyQa). We thank Jos\'{e} Garre Rubio for discussions
    on potential applications of our toolbox to the quantum Newman-Moore model.
    T.M. thanks Malena Bauer for stimulating discussions on graph theory.
\end{acknowledgments}


\section*{Data availability}

The data that support the findings of this article are openly available \cite{data}.



\bibliographystyle{./bib/bibstyle.bst}
\bibliography{./bib/bibliography.bib}

\clearpage

\appendix

\section{Miscellaneous proofs \& derivations}

\subsection{Automorphisms as symmetries of blockade Hamiltonians}
\label{app:symmetries}

As already defined in \cref{df:graph_autom}, an automorphism $\phi \in
A_\C:=\aut{G_\C}$ acts on excitation patterns $\vec n\in\mathbb{Z}_2^N$ via
$(\phi\cdot\vec{n})_i = n_{\phi(i)}$.  This induces a unitary representation on
the full Hilbert space $\mathcal{H}$ of the blockade structure,
\begin{align}
    U_\phi\ket{\vec{n}} := \ket{\phi\cdot\vec{n}}.
\end{align}
As $\phi$ is bijective, it acts as a permutation of the configurations $\vec
n\in\mathbb{Z}_2^N$. Thus the matrix representation of $U_\phi$ in the basis
$\ket{\vec{n}}$ is a permutation matrix. In particular, this implies that
$U_\phi$ is unitary.

In the remainder of this section, we show that the operators $U_\phi$ are
symmetries of the blockade Hamiltonian~\eqref{eq:H}. To this end, we first
reformulate the problem in terms of the blockade graph
$G_{\C}=(V,E,\boldsymbol{\Delta})$, as introduced in \cref{sec:setting}.

The natural measure of distance on a graph is the graph metric
(cf.~\cref{app:graphs}). In the following, we denote the distance of vertices
$i,j \in V$ with respect to the graph metric as $d_{ij}$. By construction, two
vertices $i,j \in V$ are in blockade iff $(i,j) \in E$, which is equivalent to
$d_{ij} = 1$. Thus we can write the blockade potential in terms of the graph
metric as
\begin{align}
    \label{eq:U_app}
    U(d) := \begin{cases}
        0, & d > 1\\
        U_0, & d = 1.
    \end{cases}
\end{align}
In the limit $U_0 \to \infty$ we recover the blockade potential as given in
\cref{eq:U}. However, it is mathematically more convenient to consider a finite
but arbitrary large blockade strength $U_0$. Note that this does not change the
classical degeneracy of the considered blockade structures if $U_0$ is large
enough. Thus, in the following we consider the generalized blockade Hamiltonian
\begin{align}
    \label{eq:PXP_gen}
    H_\C := \sum_{i<j}\,U(d_{ij})\,n_in_j
    +\sum_i\,\left(\Omega_i\,\sigma_i^x-\Delta_i\,n_i\right),
\end{align}
with $U$ given by \cref{eq:U_app}. The remainder of this section is dedicated
to the proof of the following proposition:
\begin{proposition}
    \label{prop:PXP_inv}

    For any $\phi \in \A_\C$, the generalized blockade
    Hamiltonian~\eqref{eq:PXP_gen} is invariant under $U_\phi$:
    $U_\phi^\dagger H_\C U_\phi = H_\C$.

\end{proposition}
\begin{proof}
    For this proof, we denote the operator $\hat{n}_i$ with a hat to avoid
    confusion with its eigenvalues $n_i$. We compute the action of
    $U_\phi^\dagger \hat{n}_i U_\phi$ and $U_\phi^\dagger \sigma_i^x U_\phi$ on
    basis states $\ket{\vec{n}}$:

    For $\hat{n}_i$, we obtain
    \begin{subequations}
        \begin{align}
            U_\phi^\dagger \hat{n}_i U_\phi \ket{\vec{n}} &= U_\phi^\dagger \hat{n}_i  \ket{\phi\cdot \vec{n}}\\
            &= U_\phi^\dagger (\phi\cdot\vec{n})_i  \ket{\phi\cdot\vec{n}}\\
            &= n_{\phi(i)}  \ket{\vec{n}}\\
            &= \hat{n}_{\phi(i)}  \ket{\vec{n}}.
        \end{align}
    \end{subequations}
    
    The operators $\sigma_i^x$ flip the state of one atom in an excitation
    pattern $\ket{\vec{n}}$.  Thus $\sigma_i^x$ also has a natural action on
    the excitation patterns
    \begin{align}
        (\sigma_i^x \cdot \vec{n})_k := \begin{cases}n_k & k \neq i\\ 
        1-n_i & k = i\end{cases}.
    \end{align}
    This is consistent with its usual operator definition, in the sense that
    $\sigma_i^x\ket{\vec{n}} = \ket{\sigma_i^x \cdot \vec{n}}$. We now show
    that $\sigma_i^x\cdot(\phi\cdot\vec{n}) = \phi\cdot(\sigma_{\phi(i)}\cdot
    \vec{n})$.  To this end, consider the components of the left expression
    \begin{subequations}
    \begin{align}
        (\sigma_i^x\cdot(\phi\cdot\vec{n}))_k 
        &= \begin{cases}(\phi\cdot\vec{n})_k & k \neq i\\ 1-(\phi\cdot\vec{n})_i & k = i \end{cases}
        \\
        &= \begin{cases} n_{\phi(k)} & k \neq i\\ 1-n_{\phi(i)} & k = i \end{cases}.
    \end{align}
    \end{subequations}
    On the other hand, we obtain for the components of the right expression
    \begin{subequations}
        \begin{align}
            (\phi\cdot(\sigma_{\phi(i)}\cdot \vec{n}))_k 
            &= (\sigma_{\phi(i)}\cdot \vec{n})_{\phi(k)}\\
            &= \begin{cases} n_{\phi(k)} & \phi(k) \neq \phi(i)\\ 1 - n_{\phi(i)} & \phi(k) = \phi(i)\end{cases}\\
            &= \begin{cases} n_{\phi(k)} & k \neq i\\ 1 - n_{\phi(i)} & k = i\end{cases}\\
            &= (\sigma_i^x\cdot(\phi\cdot\vec{n}))_k.
        \end{align}
    \end{subequations}
    Here we used that $\phi$ is a bijection, hence $\phi(i) = \phi(k)$ is
    equivalent to $i = k$. 

    With this preparation, we obtain
    \begin{align}
        \sigma_i^x U_\phi \ket{\vec{n}}
        &= \ket{\sigma_i^x\cdot(\phi\cdot\vec{n})}
        \nonumber\\
        &= \ket{\phi\cdot(\sigma_{\phi(i)}\cdot \vec{n})}
        = U_\phi \sigma_{\phi(i)}^x \ket{\vec{n}}.
    \end{align}
    Since the states $\ket{\vec{n}}$ constitute a basis of $\H$, it follows
    that $U_\phi^\dagger \sigma_i^x U_\phi = \sigma_{\phi(i)}^x$. 

    With these preliminary results, we calculate the action of $U_\phi$ on the
    Hamiltonian $H_\C$. We obtain:
    \begin{widetext}

        \begin{subequations}
            \begin{align}%
                U_\phi^\dag H_\C U_\phi
                &=\sum_{i<j}\,U(d_{ij})\,n_{\phi(i)}n_{\phi(j)}
                +\sum_i\,\left(\Omega\,\sigma_{\phi(i)}^x-\Delta_i\,n_{\phi(i)}\right)
                \\
                &=\sum_{k<l}\,U(d_{\phi^{-1}(k)\phi^{-1}(l)})\,n_{k}n_{l}
                +\sum_k\,\left(\Omega\,\sigma_{k}^x-\Delta_{\phi^{-1}(k)}\,n_{k}\right)
                \\
                &=\sum_{k<l}\,U(d_{kl})\,n_{k}n_{l}
                +\sum_k\,\left(\Omega\,\sigma_{k}^x-\Delta_{k}\,n_{k}\right)
                \\
                &=H_\C 
            \end{align}%
        \end{subequations}
    \end{widetext}
    For the second equality, we reordered the summation, which is possible
    since $\phi$ is a bijection. For the third equality, we used that $\phi$ is
    an automorphism of the vertex-weighted blockade graph $G_\C$, which implies
    $\Delta_{\phi^{-1}(k)} = \Delta_{k}$ and $d_{\phi^{-1}(k)\phi^{-1}(l)} =
    d_{kl}$. 
\end{proof}

In most physical realizations of blockade structures, the blockade
potential~\eqref{eq:U} is only an approximation of the true interaction. On the
Rydberg platform, the true potential is a van der Waals interaction $U(r) =
C_6/r^6$~\cite{Saffman_2010}. In this case, the Hamiltonian~\eqref{eq:H} cannot
be solely encoded by a blockade graph. However, \cref{prop:PXP_inv} can be
readily generalized: As the van der Waals potential depends on the actual
positions of the atoms, the automorphisms must be replaced by bijective maps on
$V$ that induce \emph{isometries} on the set of atom positions, i.e.,
$|\vec{r}_i - \vec{r}_j| \overset{!}{=} |\vec{r}_{\phi(i)} -
\vec{r}_{\phi(i)}|$. It is then straightforward to check that with these
modifications the proof for \cref{prop:PXP_inv} still holds. In this case, the
$U_\phi$ are still symmetries of the Rydberg Hamiltonian, including van der
Waals interactions.

\subsection{Fully-symmetric blockade structures}
\label{app:proof_1}

Here we provide a detailed proof of \cref{prop:1}. We prove the two parts of
the proposition separately: First, we show that the ground state of the
Hamiltonian $H_\C$ is unique, as this is independent of the symmetry. Second,
we use the symmetries introduced in \cref{app:symmetries} to show the equality
of the coefficients $\lambda_{\vec{n}}$.

\begin{lemma}
    \label{lm:uniquenes}

    For $\Omega \neq 0$, the Hamiltonian $H_\C$ [\cref{eq:PXP_gen}] has a
    unique ground state $\ket{\Omega}$. This ground state can be expanded in
    the basis of excitation patterns as
    \begin{align}
        \ket{\Omega} = \sum_{\vec{n}} C_{\vec{n}}(\Omega) \ket{\vec{n}}
    \end{align}
    with coefficients that satisfy $C_{\vec{n}}(\Omega) \neq 0$ for all $\vec{n}$.

\end{lemma}

For the proof of \cref{lm:uniquenes} we make use of a classical result from
linear algebra, the \emph{Perron-Frobenius theorem}; for reference
see~\cite[chapter I.4]{Kato_1982}. We shortly summarize the contents of this
theorem and its prerequisites:

To any matrix $M \in \mathbb{\R}^{n \times n}$ we can associate a
\emph{directed} graph $G_M^\text{d}$. Its vertex set is given by $V_M :=
\{1,\ldots,n\}$; the edges of this graph are defined via $(i,j) \in E_M$ iff
$M_{i,j} \neq 0$. The matrix $M$ is called \emph{irreducible} if $G_M^\text{d}$
is \emph{strongly connected} (one can reach every vertex from every other
vertex by following the \emph{directed} edges of the graph). This directed
graph can be turned into an undirected graph $G_M^\text{u}$ by ``forgetting''
the directions of the edges. If $M$ is symmetric, then $G_M^\text{d}$ is
strongly connected iff $G_M^\text{u}$ is connected.

Next, a matrix $M \in \mathbb{\R}^{n \times n}$ is called \emph{essentially
nonnegative}, if it has nonnegative off-diagonal elements. For a matrix $M \in
\mathbb{\R}^{n \times n}$, its \emph{spectral radius} $\varrho(M)$ is defined
as $\varrho(M):= \max\{|\lambda|\, |\,\lambda \in \sigma(M)\} > 0$. We denote the
\emph{spectrum} of $M$ as $\sigma(M)$ and the \emph{eigenspace} of $M$
associated to the eigenvalue $\lambda \in \sigma(M)$ as
$\Eig(M,\lambda)$.

With these concepts, we can formulate the following theorem:
\begin{theorem}[Perron-Frobenius]
    \label{thm:PF}

    Let $M \in \R^{n\times n}$ be a irreducible and essentially nonnegative
    matrix. Then \ldots
    \begin{enumerate}

        \item The spectral radius $\varrho(M) > 0$ is an eigenvalue of $M$:
            $\varrho(M) \in \sigma(M)$.

        \item The eigenspace associated to $\varrho(M)$ is one-dimensional:
            $\dim_{\mathbb{C}} \Eig(M,\varrho(M)) = 1$.

        \item $M$ has an eigenvector $\vec{v} = (v_i)_{i = 1}^n$ with
            eigenvalue $\varrho(M)$ with only positive components: $v_i > 0$. 

    \end{enumerate}

\end{theorem}
We emphasize that the first claim is not trivial, even for finite-dimensional
systems: It is obvious that there exists an eigenvalue $\lambda$ with
$|\lambda| = \varrho(M)$, but \cref{thm:PF} also fixes its phase. In the
following, we only apply this theorem to real and symmetric matrices. In this
case, all eigenvalues of $M$ are real, thus it is well-defined to say that
$\varrho(M)$ is the largest eigenvalue of $M$.

With these tools at hand, we prove \cref{lm:uniquenes}:
\begin{proof}
    First, we consider the case $\Omega > 0$. Note that \cref{lm:uniquenes}
    makes statements about the \emph{ground state} of $H_\C$, but \cref{thm:PF}
    only makes claims about the eigenstate with the \emph{largest} eigenvalue.
    Thus we want to apply \cref{thm:PF} to the operator $-H_\C$ instead.
    However, for $\Omega >0$, the off-diagonal elements $\bra{\vec
    n}(-H_\C)\ket{\vec n'}$ are negative, so that the matrix is \emph{not}
    essentially nonnegative. To fix this, we must first unitarily transform the
    operator $-H_\C$:

    To this end, define the transformation
    \begin{align}
        Z := \prod_{i} \sigma_i^z,
    \end{align}
    where $\sigma_i^z = \ket{0}\bra{0}_i - \ket{1}\bra{1}_i$ denotes the
    Pauli-$z$ matrix. As $(\sigma_i^z)^2 = 1$ and $\sigma_i^z\sigma_j^z =
    \sigma_j^z\sigma_i^z$ for $i \neq j$, $Z$ is an involution, i.e., $Z^2 =
    \mathds{1}$. The operators $\sigma_i^z$ anticommute with $\sigma_i^x$ and
    commute with $\sigma_j^x$ for $i \neq j$. So we immediately obtain the
    identities
    \begin{align}
        Z \sigma_i^x Z = -\sigma_i^x
        \quad\text{and}\quad
        Z n_i Z = n_i\,.
    \end{align}
    Now we can define the auxiliary operator 
    \begin{align}
    \tilde{H}_\C := Z(-H_\C)Z=-ZH_\C Z
    \end{align}
    with the explicit form
    \begin{align}
        \label{eq:sim_trans}
        \tilde{H}_\C &= -Z \left(\sum_{k<l}\,U(d_{kl})\,n_{k}n_{l} 
        + \sum_k\,\left(\Omega\,\sigma_{k}^x - \Delta_{k}\,n_{k}\right)\right) Z 
        \nonumber\\
        &= -\sum_{k<l}\,U(d_{kl})\,n_{k}n_{l} 
        + \sum_k\,\left(\Omega\,\sigma_{k}^x + \Delta_{k}\,n_{k}\right).
    \end{align}
    Let $h$ denote the matrix representation of this Hamiltonian in the basis
    of excitation patterns, i.e., the matrix elements are defined by
    $h_{\vec{n},\vec{n}'} := \bra{\vec{n}} \tilde{H}_\C \ket{\vec{n}'}$. The
    off-diagonal elements ($\vec{n} \neq \vec{n}'$) of this matrix are given by 
    \begin{align}
        \label{eq:h_od}
        h_{\vec{n},\vec{n}'} = 
        \begin{cases} 
            \Omega & \text{if}\,\bra{\vec{n}}\sum_k \sigma_k^x \ket{\vec{n}'} \neq 0\\ 
            0 & \text{otherwise}
        \end{cases}\,.
    \end{align}
    Therefore $h$ is \emph{essentially nonnegative}.

    To apply \cref{thm:PF}, we have to show in addition that $h$ is
    \emph{irreducible}. To this end, we show that $\vec{n}$ is connected to
    $\vec{0}$ in $G_h^\text{u}$ for all $\vec{n} \in \Z_2^N$. Here, $N$ denotes
    the number of atoms in $\C$ and $\vec{0} = (0,\ldots, 0)$ is the
    configuration where no atom is excited. \cref{eq:h_od} shows that
    $h_{\vec{n},\vec{n}'} > 0$ iff $\vec{n}$ and $\vec{n}'$ differ in exactly
    one component. Hence a path from $\vec{n}$ to $\vec{0}$ in $G_h^\text{u}$
    can be constructed by successively flipping all components in $\vec{n}$ to
    $0$. This implies that $G_h^\text{u}$ is connected and establishes the
    irreducibility of $h$.

    Finally, we can apply \cref{thm:PF} to the matrix $h$. As spectral
    properties are independent of the chosen basis, the conclusions also hold
    for the auxiliary operator $\tilde{H}_\C$. Furthermore, from $\tilde{H}_\C
    = -Z H_\C Z$ and the invertibility of $Z$, if follows that
    $\sigma(\tilde{H}_\C) = -\sigma(H_\C)$ and thus $\max(\sigma(\tilde{H}_\C))
    = -\min(\sigma(H_\C))$. Via well-known facts about eigenspaces, we find
    that 
    \begin{subequations}
        \begin{align}
            &\Eig(\tilde{H}_\C,\max(\sigma(\tilde{H}_\C)))\nonumber\\
            =\,&\Eig(-ZH_\C Z,-\min(\sigma(H_\C)))\\
            =\,&\Eig(ZH_\C Z,\min(\sigma(H_\C)))\\
            =\,& Z\, \Eig(H_\C,\min(\sigma(H_\C)))\,.
            \label{eq:Eigenspace_transform}
        \end{align}
    \end{subequations}
    As $Z$ is invertible, this implies in particular that 
    \begin{align}
        &\dim_\mathbb{C} \Eig(H_\C,\min(\sigma(H_\C)))\nonumber\\
        =\,&\dim_\mathbb{C} \Eig(\tilde{H}_\C,\max(\sigma(\tilde{H}_\C))) 
        \stackrel{\substack{\text{Thm.~\ref{thm:PF}}\\\text{1.\,\&\,2.}}}{=} 1\,,
    \end{align}
    i.e., the ground state of $H_\C$ is unique. 

    The third statement of \cref{thm:PF} implies that any eigenvector
    $\ket{\tilde\Omega} \in \Eig(\tilde{H}_\C,\max(\sigma(\tilde{H}_\C)))$ has the
    form 
    \begin{align}
        \ket{\tilde\Omega} = e^{i\alpha} \sum_{\vec{n}}c_{\vec{n}} \ket{\vec{n}}\,,
    \end{align}
    where $e^{i\alpha}$ is global phase and $c_{\vec{n}} > 0$; in the
    following, we set $e^{i\alpha} = 1$ without loss of generality. 

    By \cref{eq:Eigenspace_transform}, the ground state $\ket{\Omega}$ of
    $H_\C$ has then the form
    \begin{subequations}
        \begin{align}
            \ket{\Omega} = Z \ket{\tilde\Omega}
            &= \sum_{\vec{n}} c_{\vec{n}} Z\ket{\vec{n}}\\
            &= \sum_{\vec{n}} (-1)^{\delta(\vec{n})}c_{\vec{n}}\ket{\vec{n}}\,,
            \label{eq:sign}
        \end{align}
    \end{subequations}
    where we defined
    \begin{align}
        \delta(\vec{n}) := \sum_{i} n_i
    \end{align}
    as the number of atoms in $\vec{n}$ that are excited. This finally
    proves the claim of \cref{lm:uniquenes} for the case $\Omega >0$.

    For completeness, we also sketch the proof for $\Omega < 0$. In this case,
    the off-diagonal elements of $\bra{\vec n}(-H_\C)\ket{\vec n'}$ are already
    positive; thus we can simply define $\tilde{H}_\C :=-H_\C$ as the auxiliary
    operator. The matrix representation of this operator is then irreducible
    and essentially nonnegative and we can apply \cref{thm:PF}. The remainder
    of the proof follows analogously with the substitution $Z\mapsto 1$.
    
    We stress that this proof remains valid in the strict ``PXP-limit'' with
    $U_0=\infty$, where the Hilbert space is restricted to the subspace of
    excitation patterns that do not violate blockades.
\end{proof}

Now we prove the second part of \cref{prop:1}, namely that all configurations
in $L_\C$ enter with equal weight into the ground state:
\begin{proof}
    From \cref{lm:uniquenes} we know that the ground state has the form 
    \begin{align}%
        \ket{\Omega} = \sum_{\vec n} C_{\vec{n}}(\Omega) \ket{\vec n},
    \end{align}%
    where $C_{\vec{n}}(\Omega) := (-1)^{\delta(\vec{n})}c_{\vec{n}}\neq 0$ (in
    the case $\Omega>0$, for $\Omega <0$ one can replace
    $\delta(\vec{n})\mapsto 0$). By \cref{prop:PXP_inv}, the operators $U_\phi$
    commute with $H_\C$ and therefore leave the eigenspaces invariant; in
    particular the ground state manifold is invariant. As the ground state
    manifold is one-dimensional (again by \cref{lm:uniquenes}), and the
    representation of $U_\phi$ is unitary, the action of $U_\phi$ on the ground
    state $\ket{\Omega}$ must have the form
    \begin{align}
        U_\phi \ket{\Omega} = e^{i\varphi(\phi)} \ket{\Omega}
    \end{align}
    and it follows that
    \begin{align}
        \label{eq:U1}
        U_\phi \ket{\Omega} 
        = \sum_{\vec n} e^{i\varphi(\phi)} C_{\vec{n}}(\Omega) \ket{\vec n}\,.
    \end{align}
    On the other hand, we obtain 
    \begin{subequations}
        \begin{align}
            U_\phi \ket{\Omega} 
            &= \sum_{\vec n} C_{\vec{n}}(\Omega) U_\phi \ket{\vec n}\\
            &= \sum_{\vec n} C_{\vec{n}}(\Omega) \ket{\phi\cdot\vec n}\\
            &= \sum_{\vec n} C_{\phi^{-1}\cdot\vec{n}}(\Omega) \ket{\vec n} \label{eq:U2}\,.
        \end{align}
    \end{subequations}
    By comparing \cref{eq:U1} and \cref{eq:U2} it follows that
    \begin{align}
        e^{i\varphi(\phi)} C_{\vec{n}}(\Omega) 
        = C_{\phi^{-1}\cdot\vec{n}}(\Omega)\,.
        \label{eq:phase}
    \end{align}
    Next we show that the phase must be equal to $1$. As our Hamiltonian is
    real, all eigenvectors can be chosen to have real components, and thus the only
    possibilities for the phases are $e^{i\varphi(\phi)} = \pm1$. By the
    definition of the coefficients $C_{\vec{n}}(\Omega)$, \cref{eq:phase}
    becomes
    \begin{align}
        e^{i\varphi(\phi)} (-1)^{\delta(\vec{n})}c_{\vec{n}} 
        = (-1)^{\delta(\phi^{-1}\cdot\vec{n})}c_{\phi^{-1}\cdot\vec{n}}.
    \end{align}
    As $\phi$ acts as a permutation of atoms, it preserves the number of exited
    atoms, hence $\delta(\vec{n}) = \delta(\phi^{-1}\cdot\vec{n})$. Thus we
    obtain
    \begin{align}
        e^{i\varphi(\phi)} c_{\vec{n}} = c_{\phi^{-1}\cdot\vec{n}}.
    \end{align}
    (The same relation follows for $\Omega <0$ where $\delta(\vec{n})\equiv
    0$.) As the coefficients $c_{\vec{n}}$ are real and positive, the only
    phase that satisfies this equation is $e^{i\varphi(\phi)} = 1$ for all
    $\phi$.
    This shows that if two configurations  $\vec{n},\vec{m}$ are part of the
    same orbit, i.e., $\vec{m} = \phi \cdot \vec{n}$ for some $\phi \in \A_\C$,
    then it must be $C_{\vec{n}}(\Omega) = C_{\vec{m}}(\Omega)$.

    So far this result holds for arbitrary blockade structures. Now we focus on
    the fully-symmetric case. For $\vec{n} \in L_{\C}$, we introduced the
    nomenclature $C_{\vec{n}}(\Omega) =: \lambda_{\vec{n}}(\Omega)$ in
    \cref{eq:ground_state}. If $L_\C$ is an orbit under the action of $\A_\C$,
    then our results above imply that $\lambda_{\vec{n}}(\Omega) =
    \lambda_{\vec{m}}(\Omega)$ for any $\vec{n},\vec{m} \in L_\C$; that is,
    $\lambda_{\vec{n}}(\Omega) \equiv \lambda(\Omega)$, as claimed.
\end{proof}

We want to note that these methods can also be applied to more general
Hamiltonians. Here we briefly sketch some results in that regard:

First, consider the generalized fluctuation term $\sum_{i} \Omega_i \sigma_i^x$
with site-dependent $\Omega_i$. If $\Omega_{\phi(i)} = \Omega_{i}$, then
$U_\phi$ remains a symmetry for this generalized Hamiltonian. Furthermore, the
proof for the uniqueness of the ground state still goes through by choosing the
modified transformation $Z = \prod_ {i: \Omega_{i} > 0} \sigma_i^z$. It is then
straightforward to check that \cref{prop:1} still holds.

Next, consider the generalized Hamiltonian $H'_\C$ with modified fluctuation
term $\sum_{i} \Omega_i \sigma_i^x$, where $\Omega_i = \pm \Omega$ (with
$\Omega > 0$). Define $W := \prod_{i: \Omega_{i} < 0} \sigma_i^z$, then $W
H'_\C W$ is in the form of \cref{eq:PXP_gen} with $\Omega_i\equiv\Omega >0$. If
$U_\phi$ is a symmetry of $W H'_\C W$ then $W U_\phi W$ is a symmetry of
$H'_\C$. As discussed in the last paragraph, the uniqueness of the ground state
still holds. Straightforward calculations show that the ground state of $H_\C'$
satisfies $W U_\phi W \ket{\Omega} = \ket{\Omega}$ which leads to the presence
of additional signs in the coefficients $\lambda_{\vec{n}}$. 

It is also possible to add additional terms to the
Hamiltonian~\eqref{eq:PXP_gen}, as long as they do not break the symmetries and
the uniqueness of the ground state. Since the latter is based on the
construction of an essentially nonnegative auxiliary matrix (using
transformations like $Z$), we are not aware of any general criterion that
characterized perturbations where this procedure works. An example for which
the uniqueness remains intact is a hopping term with $t>0$:
\begin{align}
    H' := -t \sum_{\{i,j\} \in E} (\sigma_i^+\sigma_j^- + \sigma_i^-\sigma_j^+)\,.
\end{align}
It is clearly invariant under $U_\phi$. It also does not break the uniqueness
proof, as irreducibility still holds and the auxiliary operator $-ZH'Z$ has
nonnegative matrix elements.

Finally, we want to mention that our proof for the uniqueness of the ground
state never used any properties of the interaction potential $U$. As a
consequence (with the modifications described in \cref{app:symmetries}),
\cref{prop:1} also applies to structures of Rydberg atoms where $U$ is given by
the van der Waals interaction.

\subsection{Example: Blockade graph automorphisms not realized by geometric symmetries}
\label{app:counterexample}

Let $G = (V,E)$ be a finite (simple) graph and $\rho: V \rightarrow \R^d$ be a
unit ball embedding in $d$ dimensions, i.e, it is injective and $\{i,j\} \in E$
iff $\d_2(\rho_i, \rho_j) \leq 1$; here $\d_2(\cdot,\cdot)$ denotes the
Euclidean metric. Any automorphism $\phi \in \aut{G}$ induces a map on
$\rho(V)$ by $f_\phi: \rho(V) \rightarrow \rho(V),\,\rho_i \mapsto
\rho_{\phi(i)}$. The map $f_\phi$ is a geometric symmetry of the (embedded)
graph, if it is an isometry.

Now consider the complete graph with $n$ vertices $G = K_n$, we set $V =
\{1,\ldots,n\}$. The automorphism group of this graph is $\aut{K_n} = S_n$
where $S_n$ denotes the symmetric group. For any two pairs of vertices
$\{i,j\},\{k,l\} \in E$, there exists an automorphism $\sigma \in S_n$ such
that $\{i,j\} = \{\sigma(k),\sigma(l)\}$. Suppose all graph automorphisms of
$K_n$ induce geometric symmetries, then it follows that
\begin{subequations}
    \begin{align}
        d_2(\rho_i,\rho_j) 
        &= d_2(\rho_{\sigma(k)},\rho_{\sigma(l)})\\
        &= d_2(f_\sigma(\rho_{k}),f_\sigma(\rho_l))
        = d_2(\rho_k,\rho_l)\,.
    \end{align}
\end{subequations}
Thus, all points in $\rho(V)$ have equal distance.

The maximum number of points that are pairwise equidistant is known as the
\emph{equilateral dimension}. It is a well-known fact~\cite{Guy_1983} that the
equilateral dimension of the Euclidean space $\R^d$ is $d+1$. Hence we must
have $n = |V| = |\rho(V)| \leq d+1$ for a unit ball embedding of $K_n$ that
implements all automorphisms as geometric symmetries.

For example, the complete graph $K_n$ can only be symmetrically unit disk
embedded in $d = 2$ if $n \leq 3$. In particular $K_4$ cannot be embedded in
$\R^2$ such that all its automorphisms are realized by geometric symmetries.
This could be achieved in $\R^3$ though, where all automorphisms can be
realized as symmetries of the tetrahedron.

\subsection{Lower bound on the automorphism group size}
\label{app:burnside}

Burnside's lemma yields a lower bound for the number of orbits in $L_\C$:
\begin{align}%
    \underbrace{\left|\bigslant{L_\C}{\A_\C}\right|}_{\text{Number of orbits}}
   \stackrel{\text{Burnside}}{=}
   \frac{1}{|\A_C|}\sum_{\phi\in\A_\C}|L_\C^\phi|
   \geq \frac{|L_\C|}{|\A_C|}
\end{align}%
where $L_\C^\phi$ denotes the set of elements in $L_\C$ that are invariant
under $\phi$ and we used that $\id\in\A_\C$ and $L_\C^\id=L_\C$. This implies
in particular for a fully-symmetric structure:
\begin{align}
   1\stackrel{!}{=}\left|\bigslant{L_\C}{\A_\C}\right|
   \quad\Rightarrow\quad
   |\A_C| \geq |L_\C|\,.
\end{align}

\section{No-Go theorem for fully-symmetric tessellated blockade structures}
\label{app:proof_2}

Here we provide details on the proof of \cref{prop:2}.

\subsection{Basic results from graph theory}
\label{app:graphs}

In this section we introduce some terminology and nomenclature from graph
theory, that we will need in the following sections, for reference, see,
e.g., Ref.~\cite{Diestel2025}.  In the following we only consider finite simple
graphs $G = (V, E)$, with vertex set $V$ and edge set $E$. We denote edges as
$\{u,v\} \in E$ for vertices $u,v \in V$. If the edge set is clear then we also
write $u \sim v$ iff $\{u,v\} \in E$, in this case we say $u$ and $v$ are
adjacent. The \emph{degree} of a vertex $v$, denoted $\deg(v)$, is defined as the number of
vertices adjacent to $v$. For any subset $W \subseteq
V$, its \emph{induced subgraph} is defined as the graph with vertex set $W$ and
edge set $E_W := \{\{u,v\} \in E\, |\, u,v \in W\}$.

A \emph{path} in $G$ is an ordered list of vertices $\gamma = (v_0,\ldots,v_n)$,
where $v_0,\ldots, v_n \in V$ and $\{v_i,v_{j}\} \in E$ if $j = i \pm
1$. We denote the length of this list as $|\gamma|$. The \emph{graph metric}
$\d_G(u,v)$ on $G$ is defined as the length of the shortest path between the
vertices $u$ and $v$ in $G$ minus one, i.e,
\begin{align}
    \d_G(u,v) := \min\{|\gamma|-1\, |\, \text{paths } \gamma \text{ from } u \text { to } v\}.
\end{align}
The graph metric allows one to define the $k$-\emph{neighborhood} of a vertex $v
\in V$ as
\begin{align}
    B_{k}(v) = \{v' \in V\, |\, \d_G(v,v') \leq k\}
\end{align}
and the $k$-\emph{sphere} as
\begin{align}
    \del B_{k}(v) = \{v' \in V\, |\, \d_G(v,v') = k\}.
\end{align}
A graph $G$ is called \emph{connected} if there is a path between any two
vertices in $V$.  A maximal subset of $V$ that is connected is called a
\emph{connected component}. A vertex $v \in V$ is called a \emph{cut-vertex} if
the induced subgraph of $V\setminus\{v\}$ has more connected components than
$G$.

The graph metric is invariant under automorphisms of $G$, i.e,
\begin{align}
    \dg_G(v,w) = \dg_G(\phi(v),\phi(w)),
\end{align}
for $u,v \in V$ and $\phi \in \aut{G}$. Thus, the $k$-neighborhood and the
$k$-sphere are also compatible with automorphisms in the sense that
\begin{align}
    \phi(B_{k}(v)) = B_{k}(\phi(v)).
\end{align}
and
\begin{align}
    \label{eq:spere_inv}
    \phi(\del B_{k}(v)) = \del B_{k}(\phi(v)).
\end{align}
As a special case of \cref{eq:spere_inv}, the degree of a vertex is invariant
under automorphisms, i.e., $\deg(v) = \deg(\phi(v))$ for $v \in V$.

Furthermore, let $\gamma$ be a path from $u$ to $v$. Then $\phi(\gamma)$ is a
path from $\phi(u)$ to $\phi(v)$. This implies that if $S$ is a connected
induced subgraph of $G$ then $\phi(S)$ is also a connected induced subgraph of
$G$. Thus, if $v$ is a cut-vertex then $\phi(v)$ must also be a cut-vertex.

\subsection{Proof of the No-Go theorem}

The setup of our No-Go theorem is the following. The building block of the
tessellated graph $G$ is an unknown \emph{site-structure} that is connected to
four \emph{link-vertices}.  These link-vertices can be connected to an
arbitrary number of vertices in the site-structures. For the sake of
comprehensibility in the illustrations, we indicate these multiple edges by a
single (thick) line; however, it is important to keep in mind that a link-vertex
can have multiple connections into the site-structure.
\begin{center}
    \includegraphics[width=0.6\linewidth]{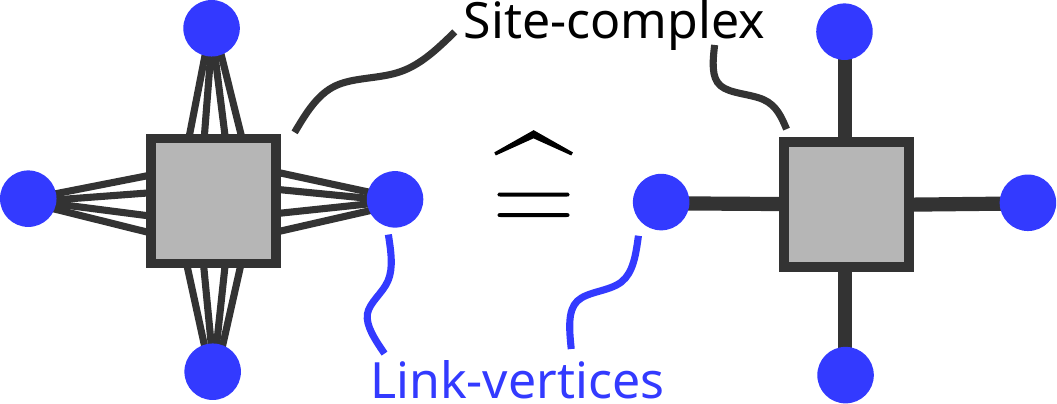}
\end{center}
We refer to this graph as $G_\text{sing}$; furthermore, we assume that it is
connected.  We refer to the graph metric on $G_\text{sing}$ as
$\d_\text{sing}(\cdot,\cdot)$. The illustration may suggest that the problem
is highly symmetric. However, we do \emph{not} assume any symmetries
(rotation or reflection) of the site-structure, and thus the illustrations have to be
treated carefully.

The tessellated graph is then built by copying and translating $G_\text{sing}$
and identifying the overlapping link-vertices. \\
\begin{center}
    \includegraphics[width=0.6\linewidth]{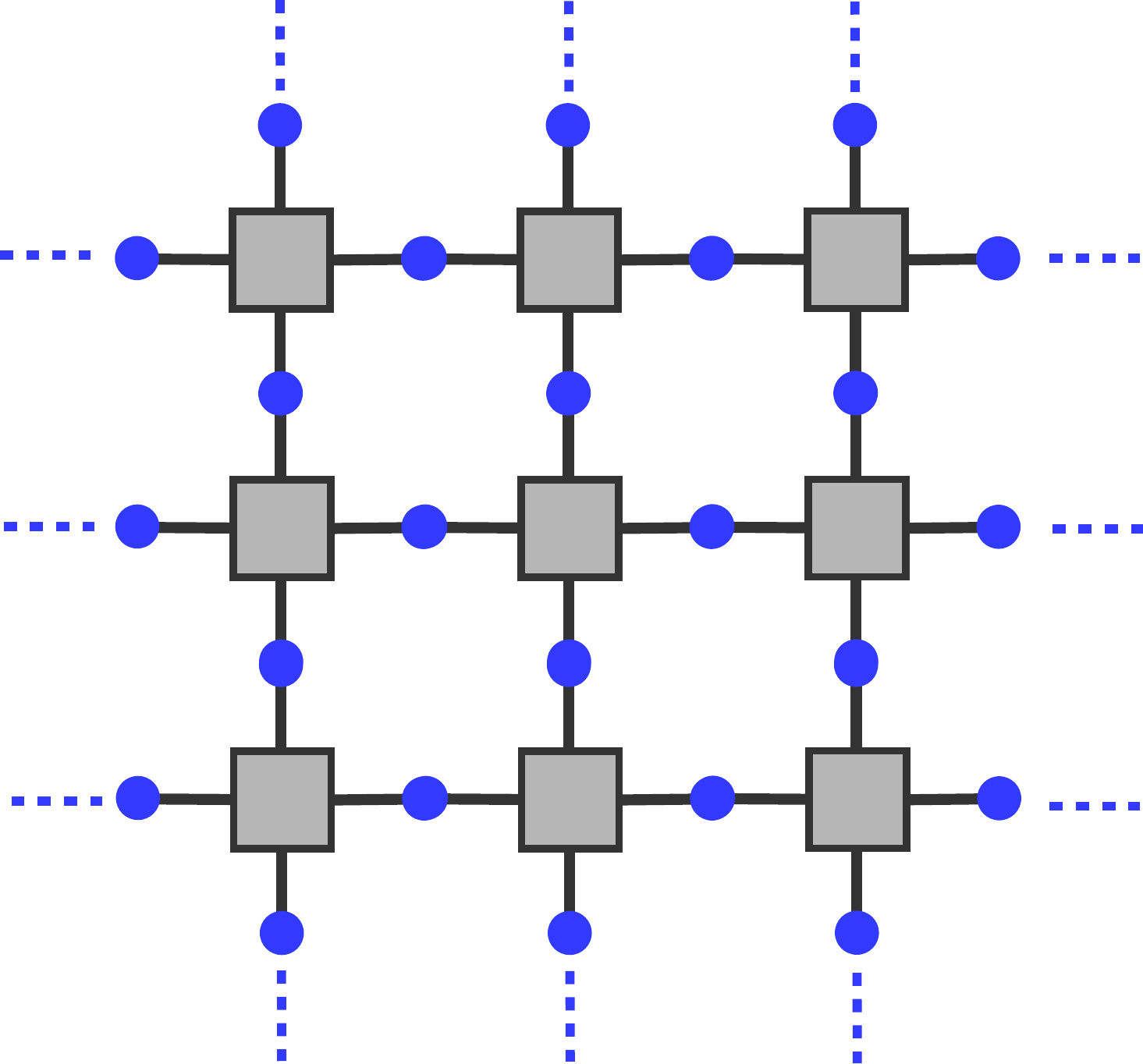}
\end{center}
We call a graph $G$ that can be build in this way a \emph{single-link
tessellated graph}.  We also consider periodic boundary conditions; here the
link-vertices on opposing boundaries of the graph are identified.  The main
result of this section is the proof of the following theorem, which essentially
states that a single-link tessellated graph does not admit nontrivial local
automorphisms.

\begin{theorem}
    \label{thm:main_res}
    
    Let $G$ be a single-link tessellated graph together with a weight function
    $\Delta$ such that the link-vertices are interpreted as ports in the sense
    of \cref{sec:review}.  Let $\phi \in \aut{G}$ such that any vertex outside
    of a rectangular region is fixed, e.g.,
    \begin{center}
        \includegraphics[width=0.5\linewidth]{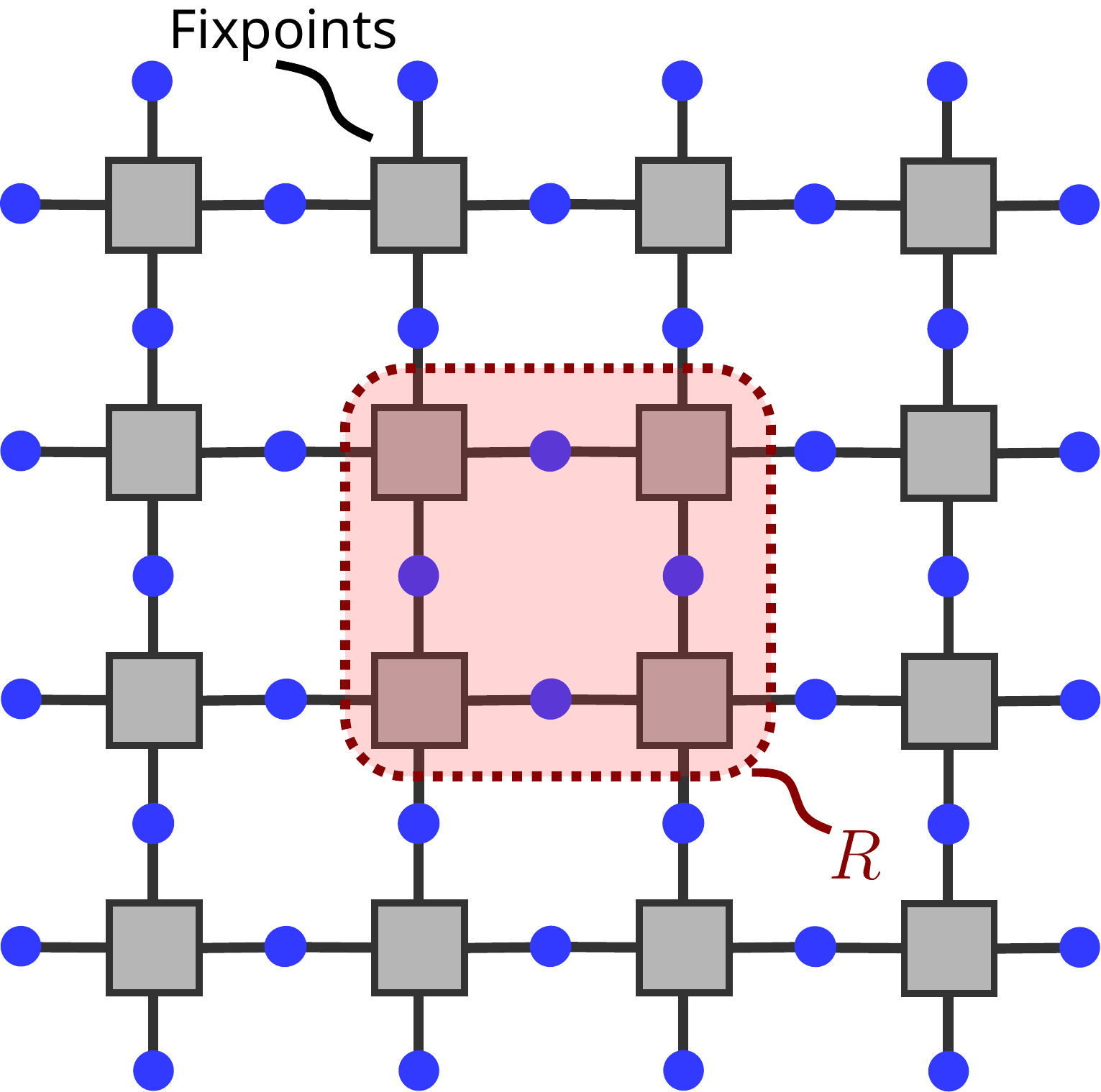}
    \end{center}
    Then $\phi$ acts trivially on the maximum-weight independent sets of $G$.

\end{theorem}

As a first step, we argue that we only have to consider the rectangular region
and the link-vertices that are directly connected to it.

\begin{lemma}
    \label{lm:reduction_to_R}

    Let $\widetilde{V} \subseteq V$ be the vertex set that consists of the
    vertices in the rectangular region and the link-vertices connected to them.
    Let $\widetilde{G}$ be the induced subgraph of $\widetilde{V}$. If all
    vertices outside of the rectangular region are fixed by an automorphism
    $\phi \in \aut{G}$ then we can restrict it to an automorphism on
    $\widetilde{G}$, for which we use the same symbol. This automorphism fixes
    all link-vertices on the border of $\widetilde{G}$.

\end{lemma}

With this in mind, the idea for the proof of \cref{thm:main_res} is
to show that the topology 
of the graph implies that the link-vertices in $\widetilde{G}$ that are 
closest to the boundary are also fixed. To this end, we use the invariance 
of the graph metric under graph automorphisms. However, the relation between 
the graph metrics $d_\text{sing}$ and $d_{\widetilde{G}}$ are a priori 
nontrivial. So the first step is to establish a connection.

The graph metric $d_\text{sing}$ is the length of the shortest path between two
link-vertices, sharing the same site-structure $U$ that is contained in this
site-structure. Thus, for a fixed sequence of link-vertices $(v_1,\ldots,v_n)$ for
$n \in \mathbb{N}$, the length of the shortest path in $\widetilde{G}$ that
crosses between two site-structures in exactly the vertices
$v_1,\ldots,v_n$ is given by
\begin{align}
    \label{eq:len_linkvert}
    L((v_1,\ldots,v_n)) := \sum_{i = 1}^{n-1} d_\text{sing}(v_{i},v_{i+1}).
\end{align}
This precise phrasing is actually crucial here, as the path may contain other
link-vertices than $v_1,\ldots,v_n$, if the path returns to the same
site-structure it originated from. This case cannot be excluded; for example in
the site-structure
\begin{center}
    \includegraphics[width=0.3\linewidth]{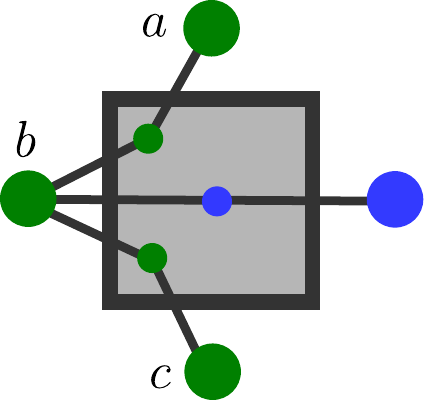}
\end{center}
the shortest path between the link-vertices $a$ and $c$ contains the
link-vertex $b$. However, this is fully incorporated into the graph metric
$d_\text{sing}$, so we do not have to keep track of link-vertices where the
path does not cross to an adjacent site-structures. This helps a lot with the notation.
This makes the notation much more convenient.

For the expression~\eqref{eq:len_linkvert} to be well-defined we can set
$d_\text{sing}(v_i,v_{i+1}) = \infty$, if $v_i$ and $v_{i+1}$ do not share a
site-structure. Consequently, for link-vertices $v,w$,
$d_{\widetilde{G}}(v,w)$ is given by the minimum of $L((v,v_1,\ldots,v_n,w))$
over all sequences of link-vertices $(v,v_1,\ldots,v_n,w)$.

We can exclude certain sequences that cannot be minimal. If $v_{i-1},v_i$ and
$v_{i+1}$ are all part of one graph $G_\text{sing}$ then, by the triangle
inequality, we find that
\begin{align}
    \d_\text{sing}(v_{i-1},v_i) + \d_\text{sing}(v_i,v_{i+1}) \geq \d_\text{sing}(v_{i-1},v_{i+1}).
\end{align}
Thus, to find the path of minimal length, it is sufficient to consider paths
with the property that three successive vertices cannot belong to the same
graph $G_\text{sing}$. Furthermore, only paths where any two successive
link-vertices are part of one graph $G_\text{sing}$ have finite length. We
refer to such paths as \emph{candidate-paths}.

To formalize arguments about the path of minimal length, we introduce a
labeling of the link-vertices by tuples $(n_x,n_y) \in \Z^2$, where $n_x$ and
$n_y$ have different parities.
\begin{center}
    \includegraphics[width=0.4\linewidth]{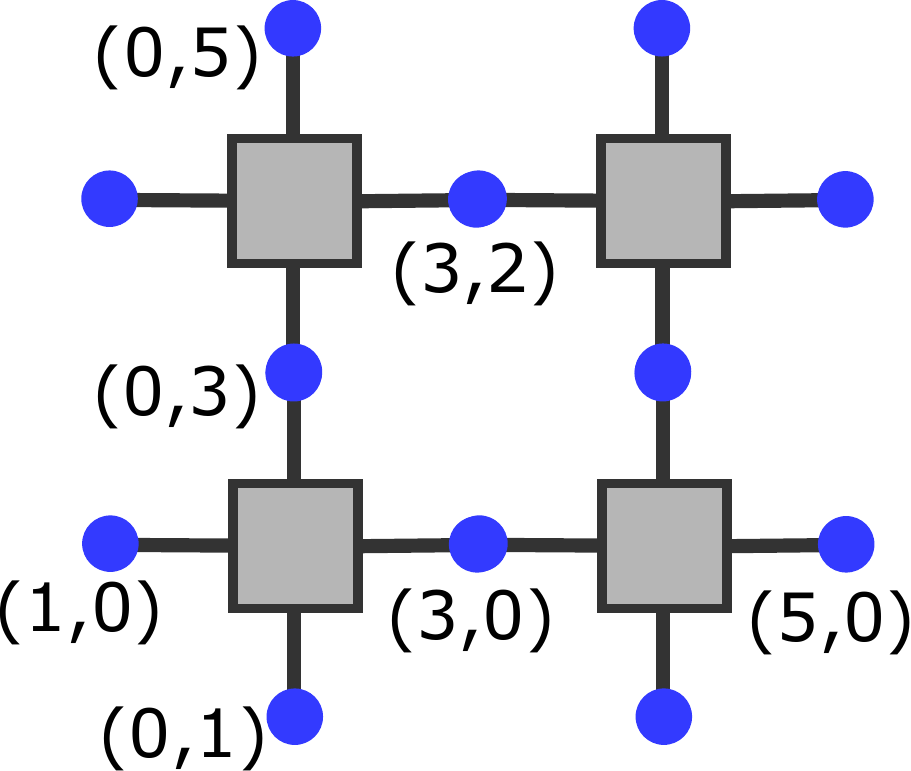}
\end{center}
A sequence of link-vertices is now represented as a sequence of tuples
$((n_x^{(i)},n_y^{(i)}))_{i = 1}^n$. We denote the difference of two successive
tuples as $(\delta_x^{(i)},\delta_y^{(i)}) := (n_x^{(i+1)},n_y^{(i+1)}) -
(n_x^{(i)},n_y^{(i)})$; we refer to such a tuple as a \emph{move}. The sequence
of labels can be reconstructed from the sequence of differences and the start
point via
\begin{align}
    \label{eq:reconstr}
    (n_x^{(n)},n_y^{(n)}) = (n_x^{(1)},n_y^{(1)}) + \sum_{i = 1}^{n-1} (\delta_x^{(i)},\delta_y^{(i)}).
\end{align}
Conversely, to construct a sequence from $(n_x^{(1)},n_y^{(1)})$ to
$(n_x^{(n)},n_y^{(n)})$, we have to find a sequence
$(\delta_x^{(i)},\delta_y^{(i)})$ such that \cref{eq:reconstr} holds.

However not all such sequence can be realized by an actual sequence of
link-vertices.  As discussed above, two successive link-vertices must share a
site-structure directly connected to them.  Thus, the set of moves is restricted
to $(\pm1,\pm1)$ and depending on the parity of the vertex, $(0,\pm2)$ or
$(\pm2,0)$. Also, note that the same move can represent different lengths,
depending on the parity of the start vertex of the move.
\begin{center}
    \includegraphics[width=0.7\linewidth]{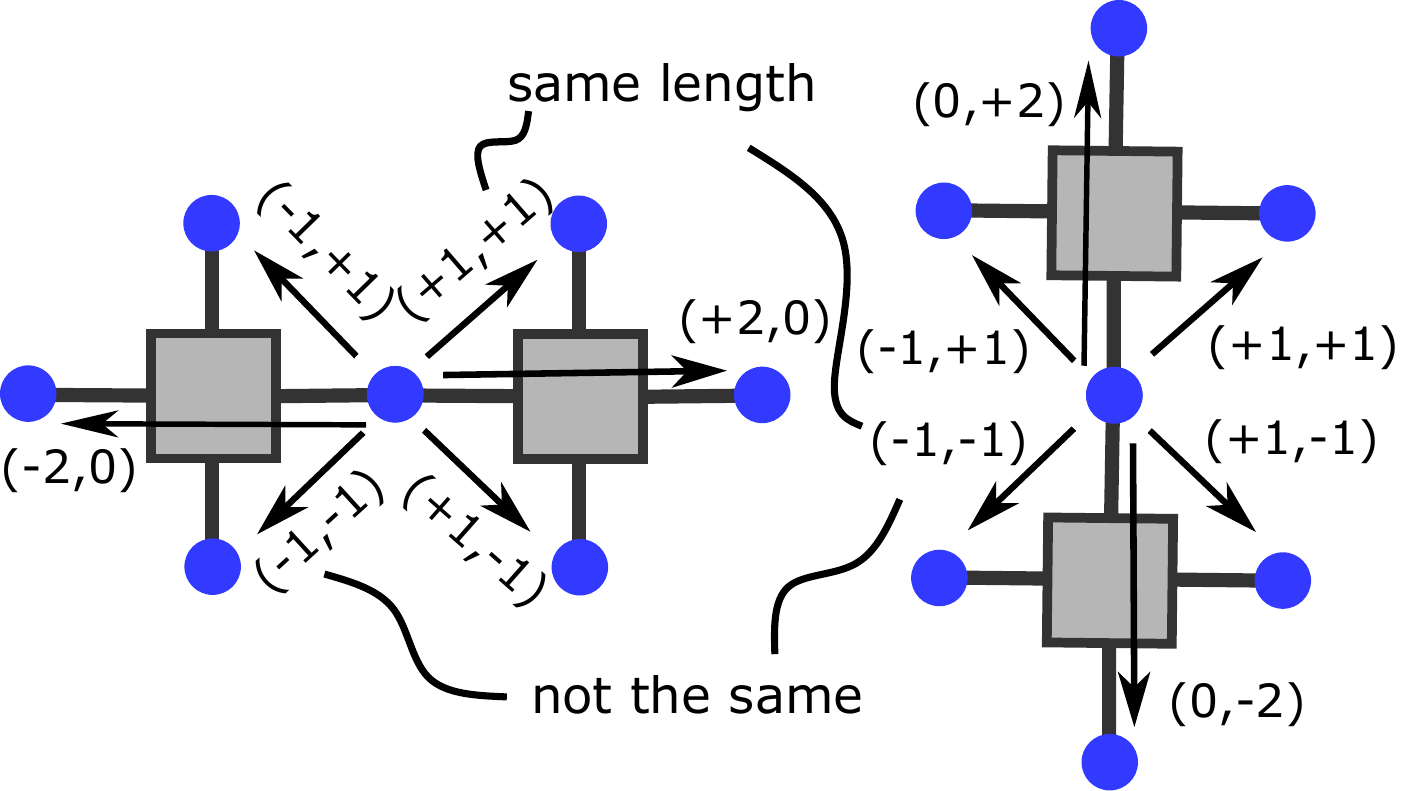}
\end{center}
For candidate-paths, the restrictions translate directly to rules which moves
can come in immediate succession.
\begin{center}
    \begin{tabular}{ c | c | c }
     previous move & parity & possible next move \\ 
     \hline
     $(+1,-1)$ & (even, odd) & $(+1,+1),(+1,-1),(+2,0)$ \\  
     $(+1,-1)$ & (odd, even) & $(+1,-1),(-1,-1),(0,-2)$ \\  
     \hline
     $(+1,+1)$ & (even, odd) & $(+1,+1),(+1,-1),(+2,0)$\\
     $(+1,+1)$ & (odd, even) & $(-1,+1),(+1,+1),(0,+2)$\\
     \hline
     $(-1,+1)$ & (even, odd) & $(-1,+1),(-1,-1),(-2,0)$\\
     $(-1,+1)$ & (odd, even) & $(-1,+1),(+1,+1),(0,+2)$\\
     \hline
     $(-1,-1)$ & (even, odd) & $(-1,+1),(-1,-1),(-2,0)$\\
     $(-1,-1)$ & (odd, even) & $(+1,-1),(-1,-1),(0,-2)$\\
     \hline
     $(+2,0)$ & - & $(+1,+1),(+1,-1),(+2,0)$\\
     $(-2,0)$ & - & $(-1,+1),(-1,-1),(-2,0)$\\
     $(0,+2)$ & - & $(-1,+1),(+1,+1),(0,+2)$\\
     $(0,-2)$ & - & $(+1,-1),(-1,-1),(0,-2)$\\
    \end{tabular}
\end{center}
Suppose we are given a sequence of moves that contains a move, that changes
only one coordinate, e.g., $(0,+2)$. Removing this move from the sequence
creates another valid path, as this type of move does not change the possible
next moves. By \cref{eq:reconstr} the endpoint of this reduced sequences is
$(n_x^{(n)},n_y^{(n)}) - (\delta_x^{(i)},\delta_y^{(i)})$, where
$(n_x^{(n)},n_y^{(n)})$ was the endpoint of the original sequence. As each move
is associated with a strictly positive distance, the length path to the new
endpoint is strictly less than the original path.

Equipped with these tools, we can prove the following lemma.

\begin{figure}[tb]
    \centering
    \includegraphics[width=0.8\linewidth]{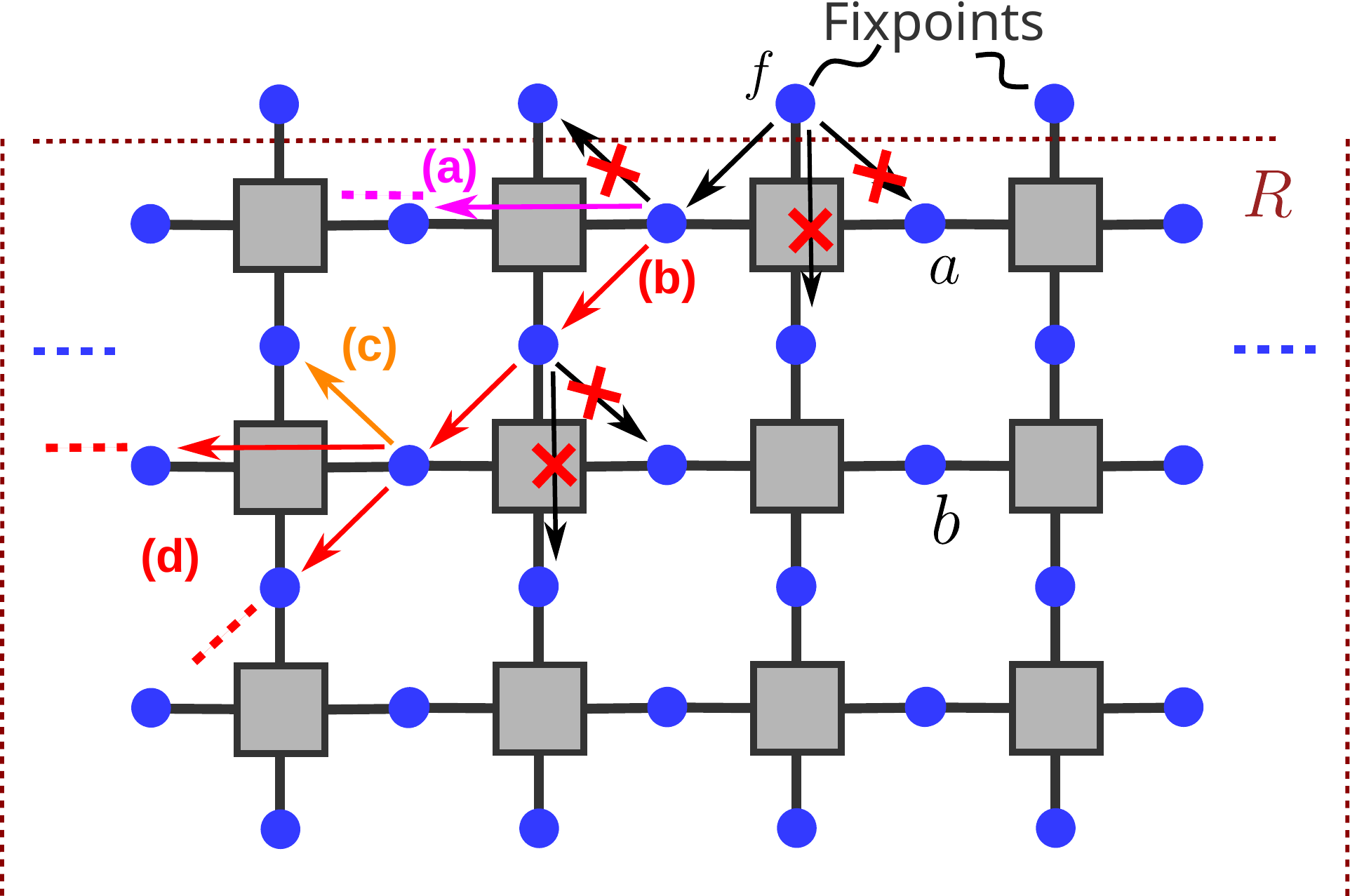}
    \caption{%
    \emph{Visualizations for \cref{lm:greater_distance}.} 
    Shown is the positioning of the vertices used in
    \cref{lm:greater_distance}. The arrows indicate possible candidate-paths
    from $f$ to $b$, we search for a path that is shorter than all
    possible paths from $f$ to $a$. Moves that can be immediately excluded are
    crossed out. 
    (a)~A path containing only moves $(-2,0)$ cannot reach the vertex $b$. 
    (b)~To reach $b$, the path must contain $(-1,-1)$, the only
    allowed follow-up is $(-1,-1)$.
    (c)~This path would contain three distinct diagonal move and can thus be
    excluded via the triangle inequality.
    (d)~The only remaining path is a descent to the bottom left boundary. Thus,
    this path can never reach $b$.
    }
    \label{fig:lm_greater_diatance_visual}
 \end{figure}

\begin{lemma}
    \label{lm:greater_distance}

    The distances of vertices $f,a,b \in \widetilde{G}$, positioned as shown in
    \cref{fig:lm_greater_diatance_visual}, satisfy $\d_{\widetilde{G}}(f,b) >
    \d_{\widetilde{G}}(f,a)$.

\end{lemma}
This may look trivial, but the possible asymmetry of the site-structures makes
the proof rather technical. However, the gist is quiet simple to summarize: To
go to the right one either has to go to the right or three times to the left.
Both paths make the path from $f$ to $b$ longer than the path from $f$ to $a$.
\begin{proof}

    We show that, for any path from $f$ to $b$, there is a path from $f$ to $a$
    with strictly smaller length.

    For a path from $f$ to $a$, the sum of the moves must equal $(+1,-1)$; for a
    path from $f$ to $b$, the sum of the moves must equal $(+1,-3)$ [cf.~\cref{eq:reconstr}].
    The difference between the vertices $a$ and $b$ is
    exactly $(0,-2)$. Consequently, if a path from $f$ to $b$ contains the move
    $(0,-2)$, it can be removed, which creates a path from $f$ to $a$ with
    strictly smaller length (as discussed above).

    In a similar fashion, there is a path from $f$ to $a$ consisting of this
    singular move $(+1,-1)$. Thus, if a path from $f$ to $b$ contains this move
    then there is a path from $f$ to $a$ with strictly smaller length [i.e.,
    the path consisting of only the move $(+1,-1)$].
    
    Thus if there is a path from $f$ to $b$ that is shorter than every path
    from $f$ to $a$, it cannot contain the moves $(0,-2)$ and $(+1,-1)$. So in
    the following we only have to consider paths that do not contain these
    moves.
    
    For any candidate-path from $f$ to $b$, there are three possible moves from
    the start point, i.e, $(-1,-1)$, $(+1,-1)$, and $(0,-2)$. As we already
    excluded the latter two moves, the only possibility is the move $(-1,-1)$
    (cf.~\cref{fig:lm_greater_diatance_visual}). The vertex $f$ is in a
    (even,odd) position; thus, as the first move is $(-1,-1)$, the next vertex
    in the candidate-path is in an (odd, even) position. Thus, the next possible
    moves are $(-1,+1),(-1,-1),(-2,0)$.
    
    The move $(-1,+1)$ leads to a link-vertex, that is only connected to one
    site-structure; hence there are no possible further moves (as we are interested in
    candidate-paths, cf.~\cref{fig:lm_greater_diatance_visual}).
     
    The move $(-2,0)$ does not change the next possible moves. Moreover,
    repeating the move $(-2,0)$ over and over, only decreases the
    $x$-component of the label and the path eventually hits the boundary on the
    left [cf.~\cref{fig:lm_greater_diatance_visual}~(a)]. Thus, to reach the
    vertex $b$, the path must eventually include the move $(-1,-1)$ 
    [cf.~\cref{fig:lm_greater_diatance_visual}~(b)]. 

    After the move $(-1,-1)$ the path is in an (even, odd) position, as was the
    vertex $f$. As discussed before, we only have to further examine paths
    where the next move is $(-1,-1)$. Then the new endpoint is in an (odd,
    even) position and the next possible moves are $(-1,+1),(-1,-1),(-2,0)$.

    Now the move $(-1,+1)$ is available; however, then the path contains three
    different diagonal moves, i.e., $(-1,-1)$ starting from an (even, odd) and
    an (odd, even) position and $(-1,+1)$. Thus, by the triangle inequality,
    this path has length strictly larger than the length associated with
    $(+1, -1)$, a path from $f$ to $a$ [cf.~\cref{fig:lm_greater_diatance_visual}~(c)].
    \begin{center}
        \includegraphics[width=0.8\linewidth]{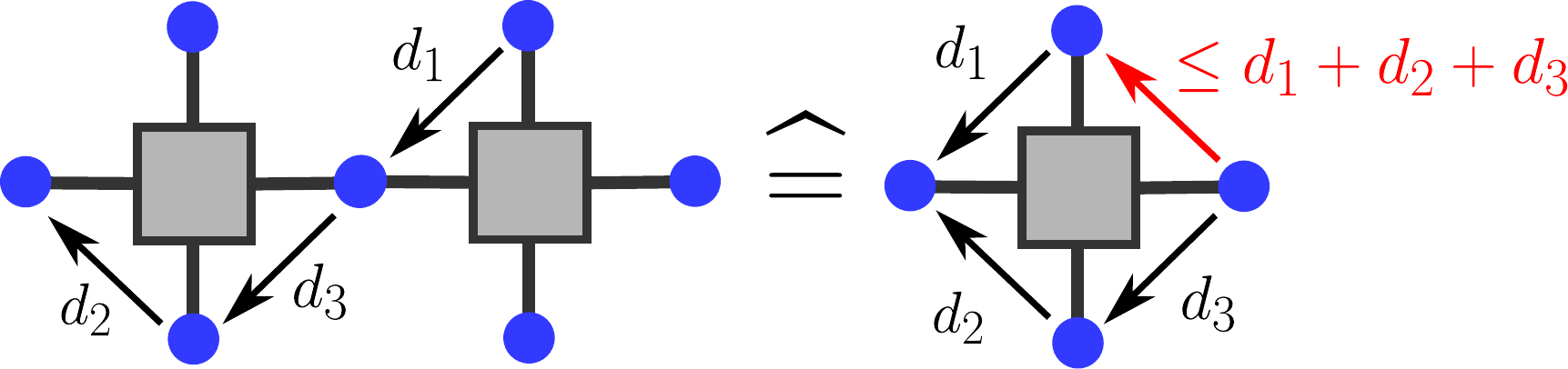}
    \end{center}
    However, excluding the move $(-1, +1)$, the path can only contain the moves
    $(-1,-1)$ and $(0,-2)$ (if it is available); hence, the $x$-component is
    only ever decreased along the path. Hence, the path cannot reach $b$ 
    [cf.~\cref{fig:lm_greater_diatance_visual}~(d)].

\end{proof}

We remark that this proof can be easily extended to all vertices in
$\widetilde{G}$, such that their labels have a difference of $(0,-2k)$, $k \in
\mathbb{N}$, from $a$. Furthermore, \cref{lm:greater_distance} also holds in the
mirrored or rotated situation, as the argument is symmetric under rotations
and reflections.

We proceed to prove the first main ingredient for \cref{thm:main_res}. 

\begin{figure}[tb]
    \centering
    \includegraphics[width=0.7\linewidth]{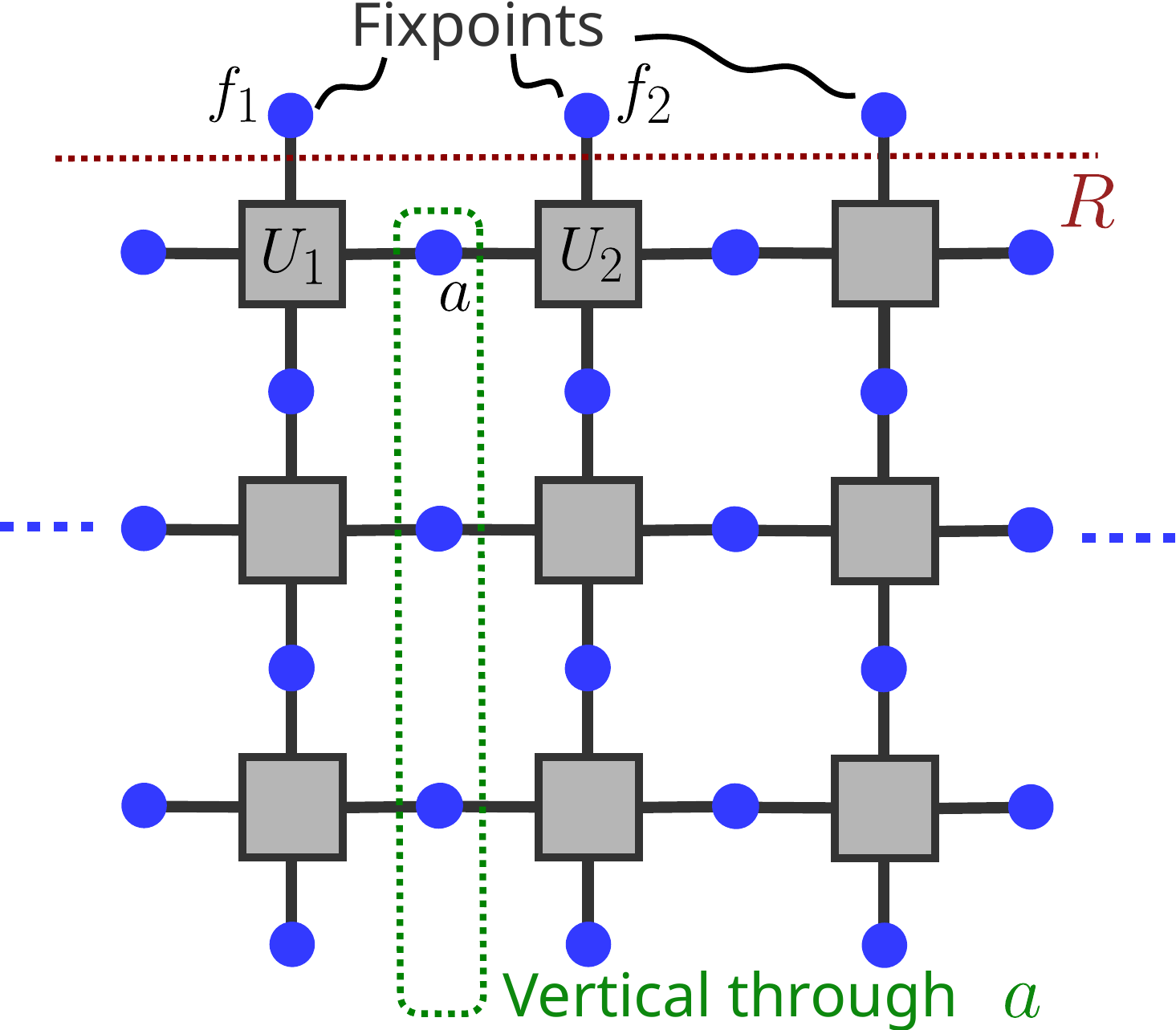}
    \caption{%
    \emph{Visualizations for \cref{lm:fixpoints_1}.} 
    Shown is the positioning of the vertices used in \cref{lm:fixpoints_1}.  In
    addition, the vertical of the link-vertices through $a$ is shown. For $d_1 :=
    d_{\widetilde{G}}(f_1,a)$, the set $B_{d_1}(f_1)$ cannot contain vertices
    right of this vertical and, analogously, for $d_2 := d_{\widetilde{G}}(f_2,a)$,
    the set $B_{d_2}(f_2)$ cannot contain vertices left of this vertical.
    }
    \label{fig:lm_greater_diatance_visual2}
 \end{figure}

\begin{lemma}
    \label{lm:fixpoints_1}
     
    Let $f_1,f_2\in \widetilde{G}$ be two link-vertices at the boundary,
    attached to neighboring site-structures $U_1,U_2$ and $a \in \widetilde{G}$
    the link-vertex between $U_1$ and $U_2$ (cf.~\cref{fig:lm_greater_diatance_visual2}),
    then $\phi(a) = a$.

\end{lemma}

\begin{proof}
    
    We focus on one spacial orientation, as the argument is the same in all
    cases.

    Let $d_1 := d_{\widetilde{G}}(f_1,a)$; then $B_{d_1}(f_1)$ cannot contain
    vertices right of the vertical through $a$, as shown in
    \cref{fig:lm_greater_diatance_visual2}.  Any path $\gamma$ from $f_1$ to
    the right of the vertical must contain one of the link- vertices on this
    vertical. Thus, $\gamma = (f_1,\ldots,l,\ldots, w)$, where $l$ is a link-
    vertex on the vertical and $w$ is the final vertex. However, by
    \cref{lm:greater_distance},
    \begin{align}
        |(f_1,\ldots,l)|-1 \geq d_{\widetilde{G}}(f_1,l) \geq d_{\widetilde{G}}(f_1,a).
    \end{align}
    As the remaining part of the path contains at least one vertex $|(\ldots,
    w)| \geq 1$, we find that $|\gamma|-1 > d_{\widetilde{G}}(f_1,a)$. As this
    is true for any path from $a$ to $w$, we have $w \notin B_{d_1}(f_1)$. 

    In addition, \cref{lm:greater_distance} shows that the only link-vertex on
    the vertical that can be part of $B_{d_1}(f_1)$ is $a$. 

    As the argument is symmetric, we also know that $B_{d_2}(f_2)$ cannot
    contain vertices left of the vertical through $a$, where $d_2 =
    d_{\widetilde{G}}(f_2,a)$. Furthermore, the only vertex on the vertical
    that can be part of $B_{d_2}(f_2)$ is $a$.

    This shows that $B_{d_1}(f_1) \cap B_{d_2}(f_2)  =\{a\}$. For the image
    under $\phi$, we find that
    \begin{subequations}
        \begin{align}
            \{\phi(a)\} &= \phi(B_{d_1}(f_1)) \cap \phi(B_{d_2}(f_2))\\
            &= B_{d_1}(\phi(f_1)) \cap B_{d_2}(\phi(f_2))\\
            &= B_{d_1}(f_1) \cap B_{d_2}(f_2) = \{a\}.
        \end{align}
    \end{subequations}
    Hence, $\phi(a) = a$.

\end{proof}

Note that, for the setup shown in \cref{fig:lm_greater_diatance_visual2},
despite the suggestive appearance, it is in principle impossible to find a
similar bound for $B_{d_1}(f_1)$ to the left side. The reason is that the
site-structure can be arbitrarily asymmetric.

Repeated application of this lemma leads to site-structures with three fixed
link-vertices.  To proceed, we prove another lemma.

\begin{lemma}
    \label{thm:three_ports}

    Let $\phi \in \aut{G}$ and $f_1,f_2,f_3,u$ be the link-vertices attached to
    some site-structure $U \subseteq G$.
    Suppose that three of the four link-vertices $f_1,f_2,f_3$ are fixed by
    $\phi$. Furthermore, assume that the graph does not end with $u$ (i.e., there
    is a site-structure different from $U$ adjacent to $u$) and that $|U| <
    |U^C|$.
    Then the fourth link-vertex $u$ is also fixed and $\phi(U) = U$.
    \begin{center}
        \includegraphics[width=0.5\linewidth]{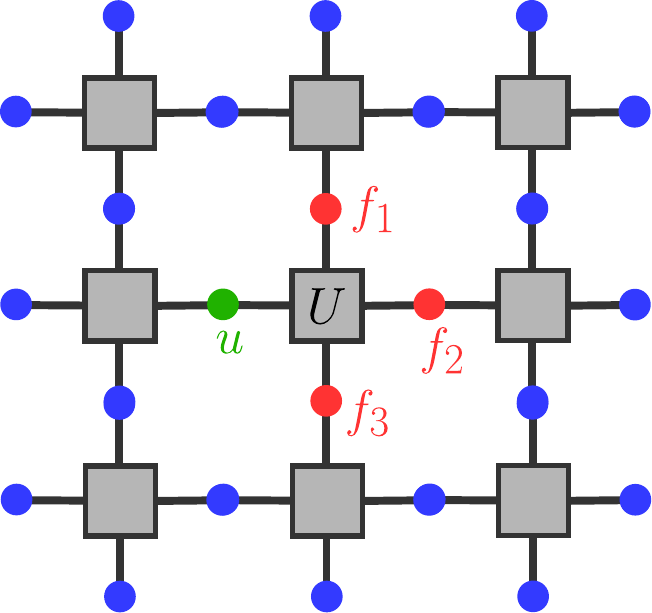}
    \end{center}

\end{lemma}

\begin{proof}

    We remove $f_1,f_2,f_3$ from the graph to create the vertex set $V'$ that
    induces the subgraph $G'$. Note that this graph $G'$ is still connected.
    Furthermore, by construction, the vertex $u$ is a cut-vertex in $G'$. Under
    removal of $u$, $G'$ decomposes in the connected components $U$ and $U^C :=
    V\setminus U$.

    The restriction of $\phi$ onto $V'$ is a well-defined automorphism that we
    denote with the same symbol. The vertex $u$ being a cut-vertex in $G'$
    implies that $\phi(u)$ is also a cut-vertex in $G'$
    (cf.~\cref{app:graphs}). Thus, the removal of $\phi(u)$ splits $G'$ into the
    components $\phi(U)$ and $\phi(U^C)$.

    Now we assume that $\phi(u) \neq u$ and show that this leads to a
    contradiction.

    \begin{itemize}

    \item\textbf{Case 1:} $\phi(u) \in U$.
    We know that $U^C \cup \{u\}$ is connected in $G'$. As $\phi(u)$ does not lie
    in this component, its connectivity cannot be influenced by the removal of
    $\phi(u)$. So $U^C \cup \{u\}$ must be part of \emph{one} connected
    component, i.e., either $U^C \cup \{u\} \subseteq \phi(U^C)$ or $U^C \cup
    \{u\} \subseteq \phi(U)$.  As the set on the left is in both cases larger
    than the set on the right, this is a contradiction.

    \item\textbf{Case 2:} $\phi(u) \in U^C$.
    We know that $U \cup \{u\}$ is connected in $G'$. As $\phi(u)$ does not lie in
    this component, its connectivity cannot be influenced by the removal of
    $\phi(u)$. So $U^C \cup \{u\}$ must be part of \emph{one} connected
    component, i.e., either $U\cup \{u\} \subseteq \phi(U^C)$ or $U \cup \{u\}
    \subseteq \phi(U)$.  As $|\phi(U)| = |U| < |U \cup \{u\}|$, the only
    possibility is $U\cup \{u\} \subseteq \phi(U^C)$. Furthermore, this implies
    that
    \begin{align}
        \phi(U) = V'\setminus\phi(U^C) \subseteq U^C \setminus \{u\}.
    \end{align}

    \begin{itemize}

    \item\textbf{Case 2.1:} $\phi(u)$ is a link-vertex.
    As $\phi(u) \neq u$, none of the adjacent site-structures can be $U$, we
    denote these adjacent site-structures as $U_1,U_2$. Both of these
    site-structures remain connected if $\phi(u)$ is removed, so both
    are a subset of $\phi(U)$ or $\phi(U^C)$, respectively. They cannot both lie
    in the same component, as $\phi(u)$ must be connected to both components
    and is only connected to $U_1$ and $U_2$. Let \Wlog $U_1 \subseteq
    \phi(U)$.

    The site-structure $U_1$ must have at least one link-vertex in $G'$ other
    then $\phi(u)$. We denote this link-vertex as $c$.  The reason for this is
    that in $G$ each site-structure has four link-vertices. We removed three
    link-vertices, which only share the site-structure $U$. Therefore, all other
    site-structures have at least three link-vertices left in $G'$.

    By construction, $U_1 \cup \{c\}$ is connected, and thus so is $U_1 \cup \{c\}
    \subseteq \phi(U)$. However, the left side is larger then the right side as
    $|U_1| = |U|$. So we have reached a contradiction.

    \item\textbf{Case 2.2:} $\phi(u)$ is not a link-vertex.
    Suppose that there is a link-vertex $x \in \phi(U)$; let $U_1, U_2$ be its
    adjacent unit cells. Let \Wlog $\phi(u) \in U_1$, then as $U_2 \cup \{x\}$
    is connected; $U_2  \cup \{x\}\subseteq \phi(U)$, which is a
    contradiction.

    \end{itemize}
    The only remaining possibility is that no link-vertex is in $\phi(U)$. Let
    $U'$ be the unit cell with $\phi(u) \in U'$. Then we have $\phi(U)
    \subseteq U' \setminus \{\phi(u)\}$, as connections to other unit cells
    would contain a link-vertex. However, this is a contradiction, as  the right
    side is smaller than the left side.

    \end{itemize}

    Therefore we have proven that $\phi(u) = u$. This creates the connected
    components $U$ and $U^C$. As we assume that $|U| < |U^C|$ and $\phi$ is
    bijective, the only possibility is that $\phi(U) = U$.
    
\end{proof}

By combining \cref{lm:fixpoints_1} and \cref{thm:three_ports} we can "grow" the
boundary by one site-structure to the inside. However, we cannot immediately
repeat this procedure, as we only know that the link-vertices are fixed and not
the site-structures. Moreover, we cannot prove this from our assumptions. If
every site-structure would contain a disconnected \texttt{NOT}-gate, then there
is a graph automorphism that exchanges the two atoms of the \texttt{NOT}-gate
and leaves all other vertices invariant. To obtain a connected counterexample,
connect the two vertices of the \texttt{NOT}-gate to one arbitrary vertex.

However, in the end, we care about the maximum-weight independent sets of the
graph.  If we interpret the link-vertices as ports then such an automorphism
inside the site-structures cannot change the ports and thus the
maximum-weight independent sets.  So these single site-structure automorphisms
``do not matter.''

To formalize this idea, we prove that we can factorize the automorphism into a
trivial part and a nontrivial part, where we can apply \cref{lm:fixpoints_1}
and \cref{thm:three_ports} again.

In the following we denote
the composition by juxtaposition $\phi_1\phi_2 := \phi_1 \circ \phi_2$.
Likewise we use ``$\prod$" to denote the composition of multiple automorphisms.

\begin{lemma}
    \label{lm:factor}

    If all link-vertices of a site-structure $U$ are fixed by an automorphism
    $\phi \in \aut{G}$, then $\phi$ factors as
    \begin{align}
        \phi = \phi_A \phi_B
    \end{align}
    with $\phi_\textup{A},\phi_\textup{B} \in \aut{G}$ and
    \begin{align}
        \phi_\textup{A} \big|_U &= \textup{id}_U \quad \phi_\textup{B} \big|_{U^C} = \textup{id}_{U^C}.
    \end{align}
    In particular, as $\phi_\textup{A}$ and $\phi_\textup{B}$ have disjoint
    support, they commute.

    In addition, if $\Delta$ is a weight function for $G$ that is invariant
    under $\phi$, then it is also invariant under $\phi_\textup{A}$ and
    $\phi_\textup{B}$.

\end{lemma}

\begin{proof}
    We define the functions
    \begin{subequations}
        \begin{align}
            \phi_\textup{A}(v) &:= \begin{cases}  \phi(v), & v\in U \\ v & v\notin U\end{cases}\\
            \phi_\textup{B}(v) &:= \begin{cases}  \phi(v), & v\notin U \\ v, & v\in U \end{cases}.
        \end{align}
    \end{subequations}
    By construction, these fulfill the desired requirements. Now we show that
    $\phi_\textup{A}$ and $\phi_\textup{B}$ are indeed automorphisms.  First,
    note that, by \cref{thm:three_ports}, we have $\phi(U) = U$.
    \begin{enumerate}

        \item For injectivity, the cases $u,v \in U$ and $u,v \notin U$ are
            clear. So let $u \in U$, $v\in U^C$. Therefore $\phi_\textup{A}(u)
            = \phi(u) \in U$ and $\phi_\textup{A}(v) = v \in U^C$. Therefore
            $\phi_\textup{A}(u) \neq \phi_\textup{A}(v)$.

        \item For surjectivity, let $u \in V$. If $u \in U^C$; then
            $\phi_\textup{A}(u) = u$. If $u \in U$ then the required preimage
            is $u = \phi^{-1}(w) \in U$.

        \item For edge transitivity, let $v,w \in V$, such that $\{v,w\} \in
            E$.  The cases $u,v \in U$ and $u,v \in U^C$ are clear. So let $u
            \in U$, $v \in U^C$.  This implies that $v$ is one of the
            link-vertices of $U$ and, according to requirement, $\phi(v) =
            v \phi_\textup{A}(v)$. Therefore, we have
            \begin{align}
                \{\phi_\textup{A}(u),\phi_\textup{A}(v)\} = \{\phi(u),\phi(v)\} \in E.
            \end{align}

    \end{enumerate}

    This shows that $\phi_\textup{A} \in \aut{G}$. The proof is analogous
    for $\phi_\textup{B}$.

    The additional claim follows directly from the fact that $\Delta$ is invariant
    under $\phi$ and the identity.
\end{proof}

Now we formalize the idea of growing the boundary to the inside.
\begin{lemma}
    \label{thm:local_autom}

    Let $G$ be a single-link tessellated graph and $\phi \in \aut{G}$ an
    automorphism that fixes every vertex outside of a rectangular region $R$.
    Then $\phi$ factors as
    \begin{align}
        \phi = \prod_{i} \phi_{U_i},
    \end{align}
    where the index $i$ labels the site-structures in $R$ such that, for the
    $i$-th site-structure $U_i$,
    \begin{align}
        \phi_i \big|_{U_i^C} = \text{id}_{U_i^C}.
    \end{align}

\end{lemma}

\begin{proof}
    By \cref{lm:reduction_to_R} we can restrict $\phi$ to an automorphism on
    the induced subgraph $\widetilde{G}$, which has all vertices in $R$ and the
    link-vertices directly connected to them as its vertex set. We denote this
    automorphism again by $\phi$. The link-vertices that form the boundary of
    $\widetilde{G}$ (i.e, link-vertices that are only connected to one
    site-structure) are, by construction fixpoints of $\phi$. Let $U_1, \ldots,
    U_n$ be the site-structures that are directly connected to the link-vertices
    on the boundary of $\widetilde{G}$.

    We proceed to prove \cref{thm:local_autom} by mathematical induction. If $R$ contains no
    site-structure then there is nothing to show. If $R$ contains exactly one
    site-structure $U_1$ then $\phi\big|_{U_1^C} = \text{id}_{U_1^C}$; thus, $\phi$ is
    of the desired form.

    For larger $R$, we proceed as follows.
    \begin{enumerate}

        \item We apply \cref{lm:fixpoints_1} repeatedly to the site-structures
            $U_i$; we find that the link-vertices between the site-structures
            $U_i$, are fixpoints of $\phi$.
            \begin{center}
                \includegraphics[width=0.9\linewidth]{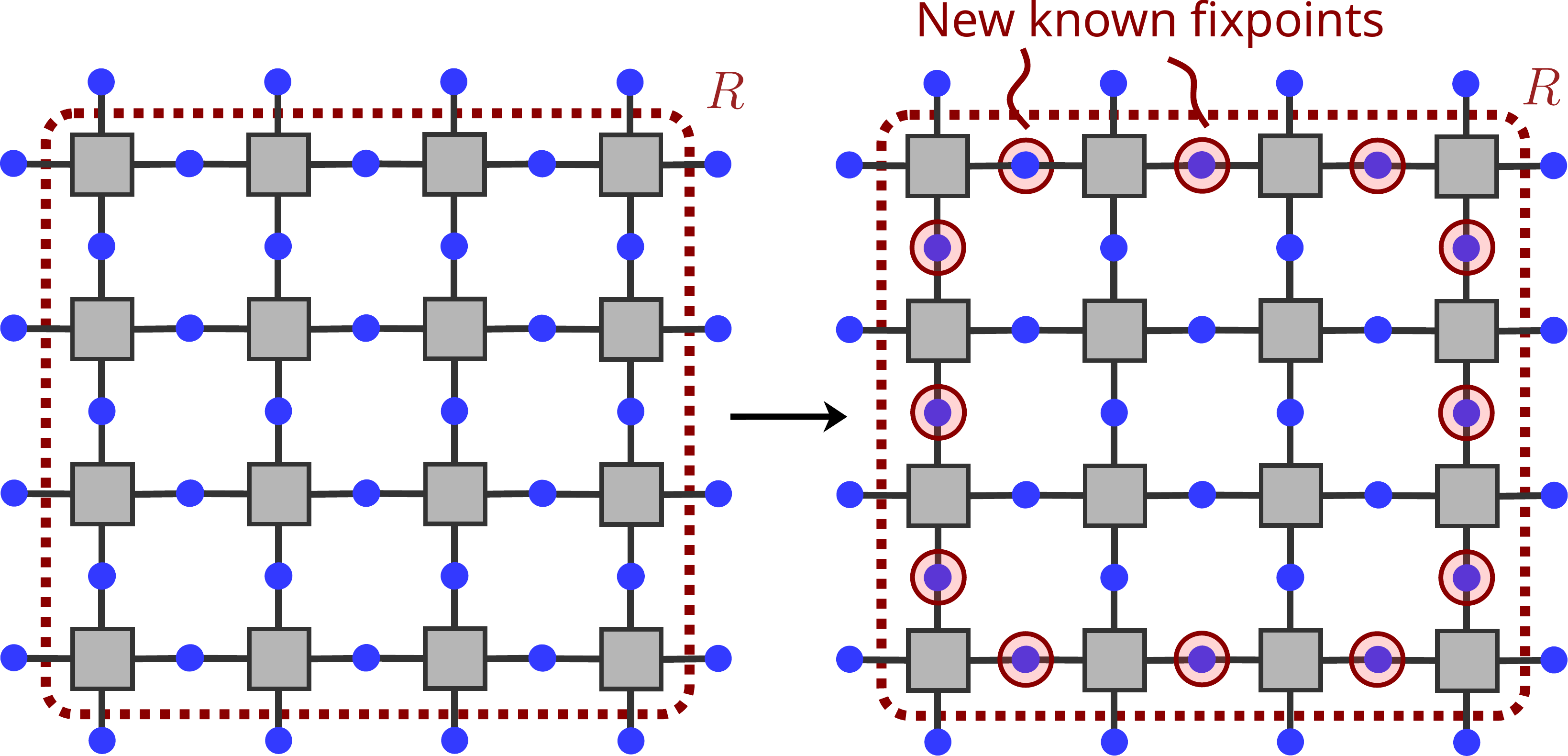}
            \end{center}

        \item Now, for the site-structures $U_1,\ldots,U_n$, three link-vertices
            are known to be fixpoints. Then \cref{thm:three_ports} implies that
            also the fourth link-vertex of these site-structures is also a fixpoint. 
            \begin{center}
                \includegraphics[width=0.9\linewidth]{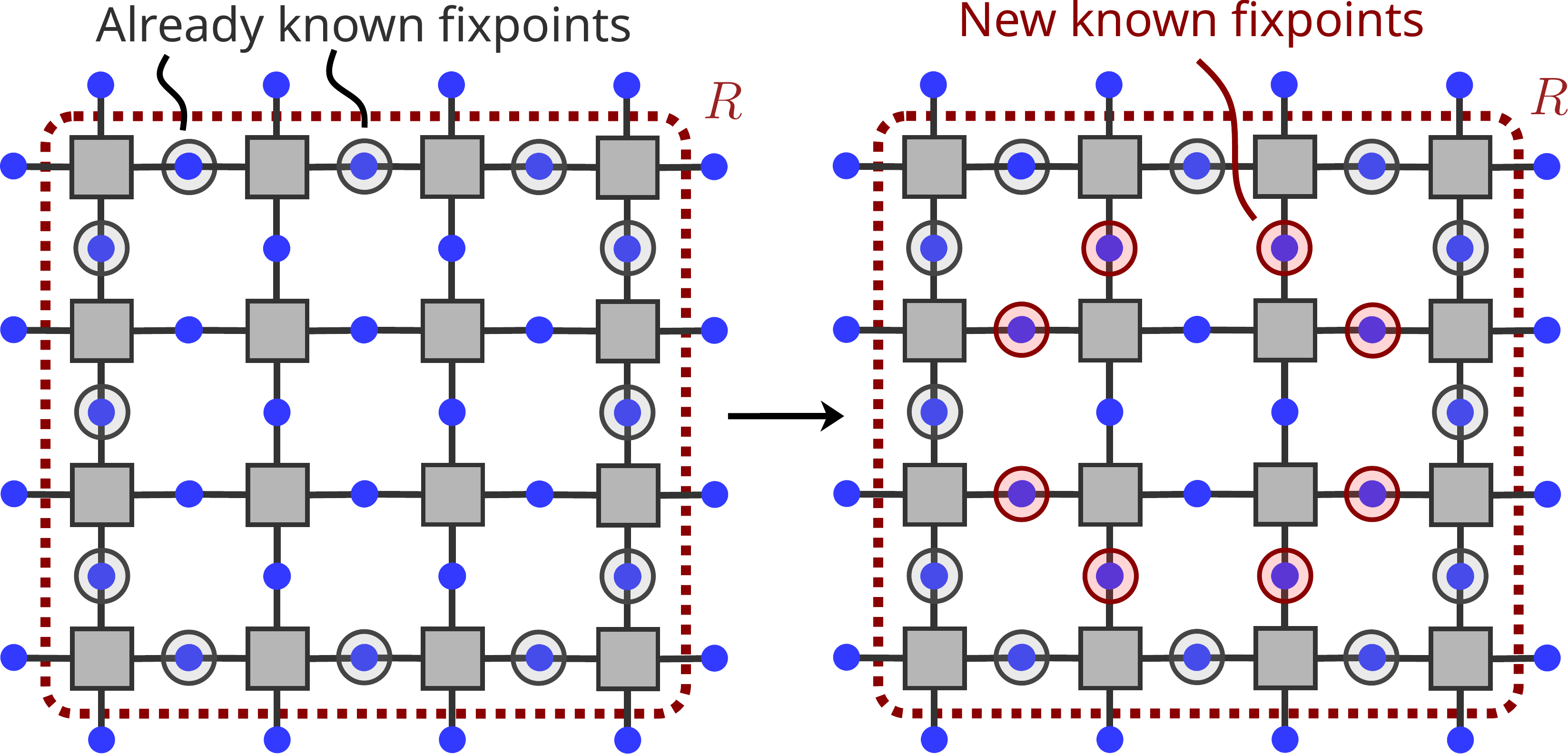}
            \end{center}

        \item We apply \cref{lm:factor} successively to the site-structures
            $U_1,\ldots,U_n$.  For the first site-structure $U_1$, we obtain the
            factorization
            \begin{align}
                \phi = \phi_{U_1}\phi_1,
            \end{align}
            where $\phi_{U_1} \big|_{U_1^C} = \text{id}_{U_1^C}$ and
            $\phi_1\big|_{U_1} = \text{id}_{U_1}$. As every vertex that is a
            fixpoint of $\phi$ is also a fixpoint of $\phi_1$, we can apply
            \cref{lm:factor} to $\phi_1$ and the next site-structure $U_2$ to
            obtain 
            \begin{align}
                \phi = \phi_{U_1} \phi_{U_2} \phi_2,
            \end{align}
            with the properties
            \begin{align}
                \phi_{U_i} \big|_{U_i^C} 
                &= \text{id}_{U_i^C},\ \phi_2\big|_{U_i} 
                = \text{id}_{U_i}
            \end{align}
            for $i \in \{1,2\}$. Repeating this argument for all site-structures
            $U_1,\ldots, U_n$, we obtain
            \begin{align}
                \phi = \phi_\text{rem} \prod_{i = 1}^n \phi_{U_i} 
            \end{align}
            with 
            \begin{align}
                \phi_i \big|_{U_i^C} 
                &= \text{id}_{U_i^C}, \ \phi_\text{rem}\big|_{U_i} 
                = \text{id}_{U_i}.
            \end{align}

        \item All site-structures $U_1,\ldots,U_n$ and their link-vertices are
            fixpoints of the automorphism $\phi_\text{rem}$. Thus, the boundary
            of the rectangular region $R$ can be advanced inwards by one
            site-structure. 
            \begin{center}
                \includegraphics[width=0.9\linewidth]{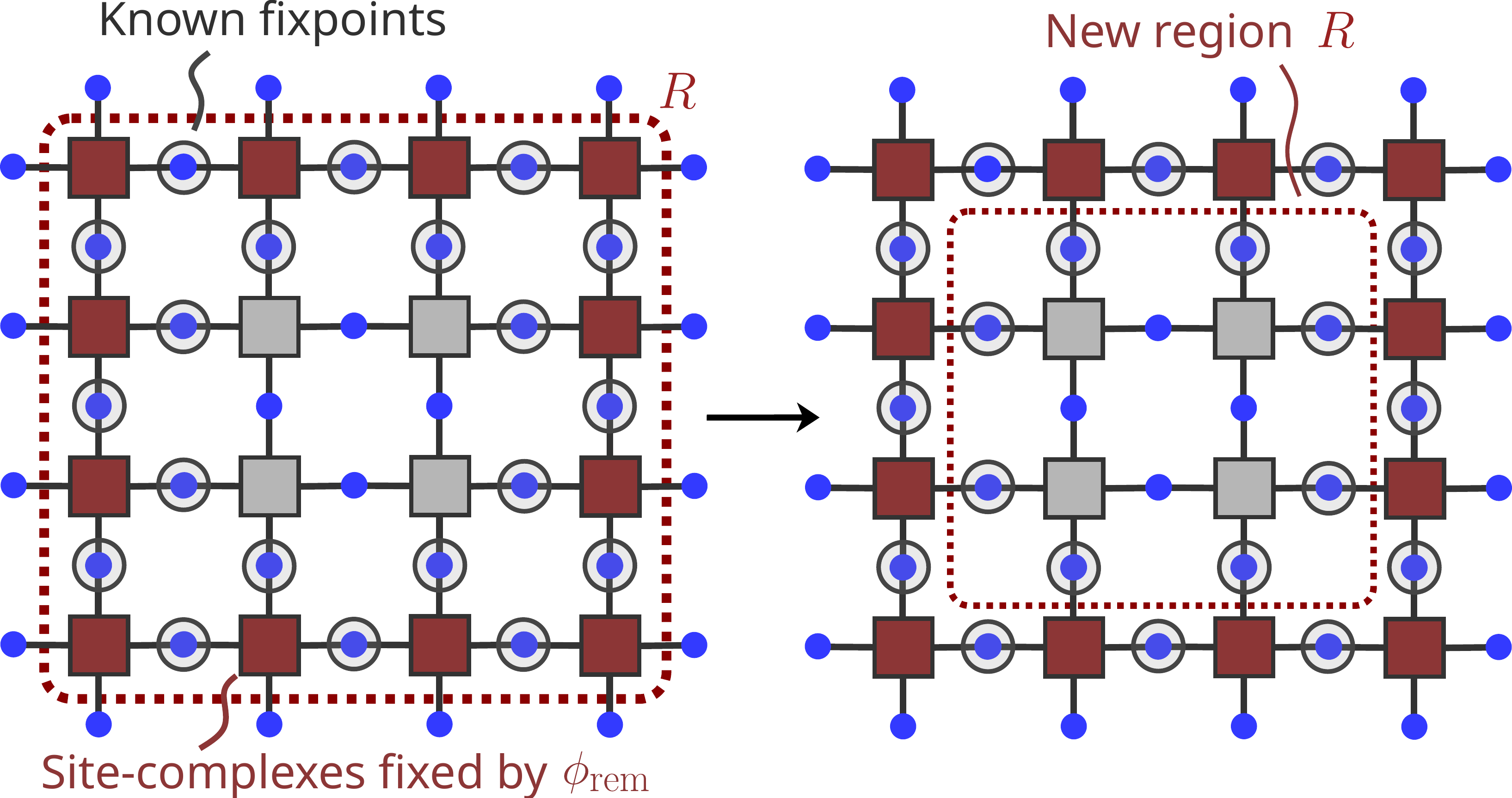}
            \end{center}

    \end{enumerate}
    Hence, by induction the automorphism $\phi$ factorizes as desired.  In
    particular, we have shown that $\phi$ can only act trivially on the
    link-vertices.
\end{proof}

We complete the proof of \cref{thm:main_res} by applying \cref{thm:local_autom}
to a vertex-weighted single-link tessellated graph.

\begin{corollary}
    \label{prop:aut_trivial}

    Let $\Delta$ be a weight function that turns $G$ into a vertex-weighted
    graph, such that the link-vertices are the ports of $G$ in the sense of
    \cref{sec:review}. In addition let $\phi \in \aut{G}$, i.e., $\Delta$ is
    invariant under $\phi$, such that $\phi\big|_R = \text{id}_R$ for some
    rectangular region $R$. Then $\phi$ acts trivially on all
    maximum-weight independent sets $I$, i.e, $\phi(I) = I$.

\end{corollary}

\begin{proof}
    The set $\phi(I)$ is an independent set, as $I$ itself is an independent set 
    (cf.~\cref{app:graphs}). As the weights $\vec{\Delta}$ are invariant under $\phi$,
    $\phi(I)$ is a maximum-weight independent set. However, by \cref{thm:local_autom},
    $\phi(I)$ contains the same ports as $I$; hence $\phi(I) = I$.
\end{proof}

\section{Automorphism groups of various tessellated structures}

\subsection{Verresen model}
\label{app:verresen}

\begin{figure*}[t]
    \centering
    \includegraphics[width=0.8\linewidth]{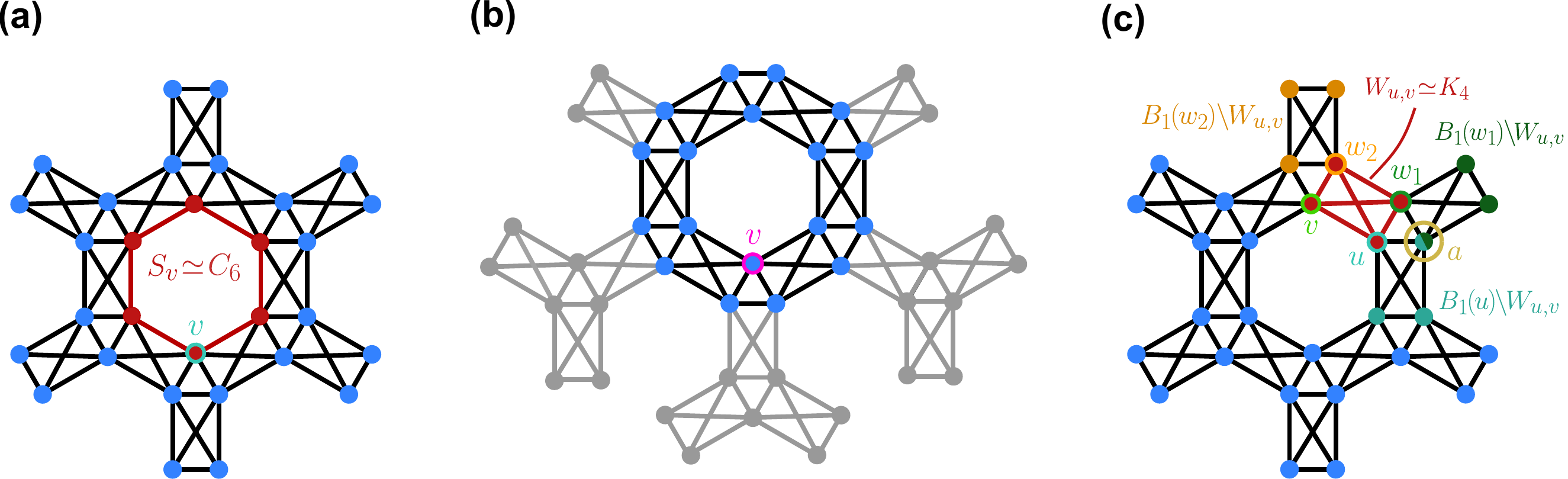}
    \caption{%
	  \emph{Automorphisms of the Verresen model.} 
    (a)~Shown is one plaquette of the blockade graph $G_\VER$ that encodes the
    arrangement of Rydberg atoms proposed by Verresen
    \etal~\cite{Verresen2021} to realize a $\ZZ$ topological spin liquid. The
    induced subgraph formed by the red colored vertices is isomorphic to $C_6$.
    For each vertex $v \in V_\VER$, there exists exactly one set $S_v$ such
    that $v\in S_v$ and its induced subgraph is isomorphic to $C_{6}$.
    (b)~Shown is the neighborhood $B_6(v)$ for some vertex $v \in V_\VER$. From
    this neighborhood, we remove recursively every vertex $w \in B_6(v)$, such that
    $\partial B_1(w)$ is fully connected. These removed vertices are colored
    gray. The remaining vertices show that $S_v$ is indeed unique.
    (c)~For any adjacent vertices $u, v \in S_u$, there exists exactly one set 
    $W_{u,v}$ (red colored vertices and edges) that contains $u$ and 
    $v$ and its induced subgraph is isomorphic to $K_4$. The name of the vertices 
    is indicated by the edge color of the red vertices. The remaining vertices 
    $w_1$, $w_2$ in $W_{u,v}$ are uniquely characterized by the properties $(B_1(u)
    \cap B_1(w_1))\setminus W_{u,v} = \{a\}$ and $(B_1(u) \cap B_1(w_2))\setminus 
    W_{u,v} = \{\}$. The sets $B_1(\cdot)\setminus W_{u,v}$ are shown in orange, 
    dark green, dark cyan for vertices $w_2$, $w_1$, and $u$. The overlap of the 
    two sets is indicated by the split colored vertex, which is named $a$. This 
    implies that the images of $w_1$ and $w_2$ under an automorphism are 
    uniquely determined by the images of $u$ and $v$.
    }
    \label{fig:Verr_Illustr}
\end{figure*}

In this section we discuss the model proposed by Verresen
\etal~\cite{Verresen2021} to realize a $\ZZ$ topological spin liquid by Rydberg
atoms placed on a Ruby lattice. In the PXP approximation, this arrangement can
be interpreted as a blockade structure $\C_\VER$ with blockade graph $G_\VER$,
shown in \cref{fig:Verr_Illustr}. We explicitly characterize the automorphism
group of this structure and show that it is not fully-symmetric.

For each vertex $v \in V_\VER$, there is exactly one induced subgraph $S_v$
such that $v \in S_v$ and $S_v \simeq C_6$, where $C_6$ denotes the cycle graph
with six vertices. The existence is clear by inspection of the graph; one
example is shown in \cref{fig:Verr_Illustr}~(a) as red vertices.

We now argue that this cycle graph is unique for each vertex $v$. It is
clear that we only have to consider vertices in $B_6(v)$; in
\cref{fig:Verr_Illustr}~(b) such a neighborhood is shown. Furthermore, we can
disregard all vertices $w \in B_6(v)$, where $B_1(w)$ is fully connected.
Suppose that such a vertex $w$ is part of an induced subgraph $C_6$. Then there
are two adjacent vertices $x_1,x_2 \in B_1(w)$ such that $x_1 \nsim x_2$.
However, $B_1(w)$ cannot then be fully connected. These deleted vertices are
colored gray in \cref{fig:Verr_Illustr}~(b). The uniqueness of this cycle 
graph in the remaining vertices
can now be easily seen, by systematically checking all remaining possibilities.

Now let $\phi \in \aut{G_\VER}$ and $u,v \in V_\VER$. Then, if $\phi(u) = v$,
we know that $\phi(S_u) = S_v$. This leaves two possibilities; let $u_1,u_2 \in
S_u$ with $u_1 \nsim u_2$ and $v_1,v_2 \in S_v$ with $v_1 \nsim v_2$. The
two possibilities are
\begin{subequations}
    \begin{align}
        \phi(u_1) = v_1&,\;\phi(u_2) = v_2\\
        \phi(u_1) = v_2&,\;\phi(u_2) = v_1.
    \end{align}
\end{subequations}
In each case, \cref{lm:cycle_graph_unique} implies that the image of each vertex 
in $S_u$ is uniquely determined.

From this information, the image of every vertex in $ V_\VER$ can be
determined, i.e., this information uniquely characterizes the automorphism.
For two adjacent vertices $u,v \in S_u$, there is exactly one subset
$W_{u,v} \subseteq V_\VER$ such that the induced subgraph corresponding to
this subset is isomorphic to the fully connected graph with four vertices,
denoted $K_4$ [cf. \cref{fig:Verr_Illustr}~(c)]. The image of $W_{u,v}$ 
under an automorphism $\phi \in \aut{G_\VER}$ must again be isomorphic to $K_4$.

If the images $\phi(u),\phi(v)$ are known then we know that
$W_{\phi(u),\phi(v)} = \phi(W_{u,v})$. Thus, the remaining vertices in $W_{u,v}$
are either invariant under $\phi$ or get swapped. The structure of the graph $G_\VER$
[cf.~\cref{fig:Verr_Illustr}~(c)] shows that there is a vertex $w_1 \in W_{u,v}$ such
that
\begin{align}
   (B_1(u)\cap B_1(w_1))\setminus W_{u,v} 
   = \{a\} \text{ with } a \in V_\VER.
\end{align}
For the remaining vertex $w_2 \in W_{u,v}$, we find that
\begin{align}
   (B_1(u)\cap B_1(w_2))\setminus W_{u,v} = \{\}.
\end{align}
Thus, under the action of the automorphism $\phi$, we obtain
\begin{subequations}
    \begin{align}
        (B_1(\phi(u))\cap B_1(\phi(w_1)))\setminus W_{\phi(u),\phi(v)} 
        &= \{\phi(a)\}\\
        (B_1(\phi(u))\cap B_1(\phi(w_2)))\setminus W_{\phi(u),\phi(v)} 
        &= \{\}.
    \end{align}
\end{subequations}
Therefore, $\phi(w_1)$ and $\phi(w_2)$ can be uniquely identified in
$W_{\phi(u),\phi(v)}$.

This argument can be repeated for every pair of adjacent vertices in the set
$S_x$ for some $x \in V_\VER$, i.e., the automorphism is determined by its
action on a single set $S_u$ for some $u \in V_\VER$. Hence the automorphism
group consists of lattice translations, rotations and reflections.

We now explicitly determine the size of this automorphism group. One unit
cell of this tessellated graph contains $6$ vertices; therefore, the whole graph
contains $6N_xN_y$ vertices, where $N_x,N_y$ are the number of unit cells in the $x$,
$y$ directions. We assume periodic boundary conditions. Choosing one reference
vertex, there are $6N_xN_y$ possible images. For each of these, two
automorphisms are possible. Therefore, we find that $|\aut{G_\VER}| = 12N_xN_Y$.

\subsection{Zeng model}
\label{app:pichler}

Zeng \etal \cite{Zeng2025} introduced a blockade structure $\C_\PQDM$
(they dubbed it a Rydberg gadget), which realizes a quantum dimer model on a
square lattice as its low-energy subspace $\H_{\C_\PQDM}$.  In
\cref{fig:Pichler_Graph_red}~(a), the blockade graph $G_\PQDM$ of $\C_\PQDM$ is
shown. The structure has two different detunings $\Delta_\text{green}$ and
$\Delta_\text{blue}$ that correspond to the blue and green colored atoms,
respectively. The blue colored atoms are the ports of this structure, whereas
the green colored ones are the ancillas (or gadget atoms).  The only
restriction on the detunings is that $\Delta_\text{green} >
\Delta_\text{blue}$. We will see later that automorphisms cannot map atoms of
different colors into each other, so any automorphism leaves the detunings
invariant; therefore, we can safely ignore the detuning in the following
argument.  In this section, we rigorously characterize the automorphism group
of this blockade graph and show that $\C_\PQDM$ is not fully-symmetric.

\begin{figure*}[tb]
   \centering
   \includegraphics[width=0.8\linewidth]{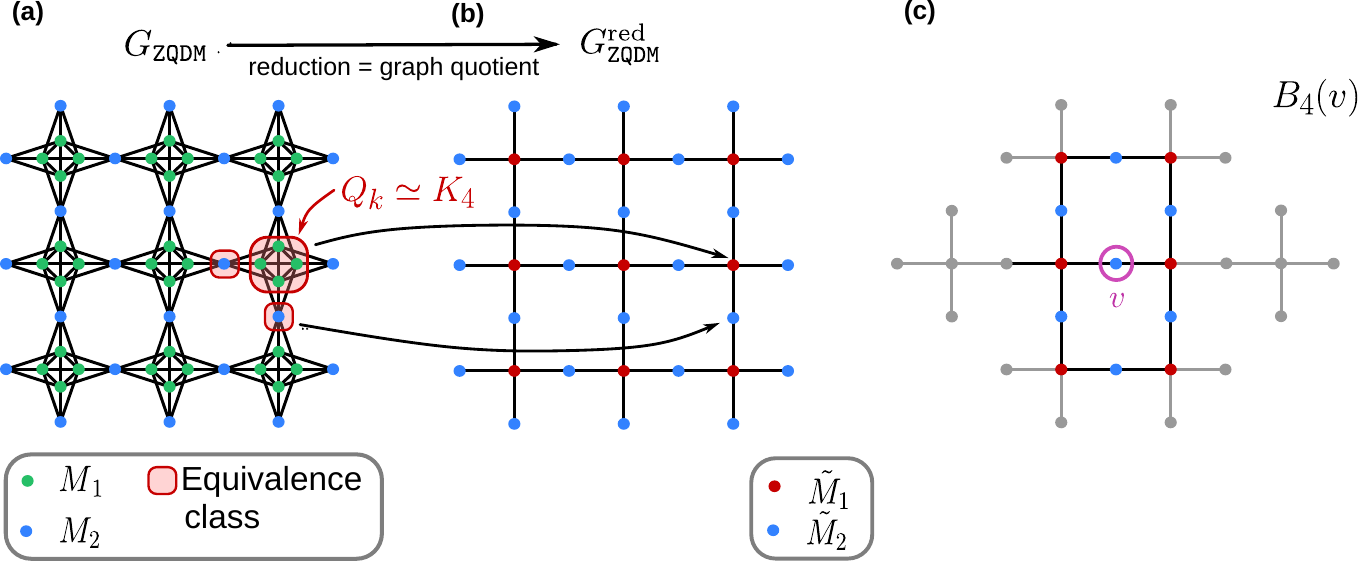}
   \caption{%
   \emph{Reduction of the Zeng model.} 
   \textbf{(a)} Shown is the blockade graph $G_\texttt{ZQDM}$, proposed by Zeng
   \etal to realize a quantum dimer model with a blockade structure. Vertices
   in the set $M_1$ are colored green, while those in the set $M_2$ are colored blue.
   The set $M_1$ decomposes into connected components $Q_k$; the induced
   subgraph of every such component is isomorphic to $K_4$. By an equivalence
   relation we identify the vertices that lie in the same component $Q_k$. Some
   of the equivalence classes are shown as red boxes.
   \textbf{(b)} Shown is the quotient graph of $G_\texttt{ZQDM}$ by the aforementioned
   equivalence relation; the resulting graph is denoted $G_\texttt{ZQDM}^\text{red}$.
   We prove in this section that the automorphism groups of both graphs have the
   same size.
   \textbf{(c)} Shown is the neighborhood $B_4(v)$ of some vertex $v\in M_2$.
   From this neighborhood, all pendant vertices are removed. From the remaining
   graph, again, all pendant vertices are removed. This process is repeated
   until no pendant vertices are left. The removed vertices are colored
   gray. This shows that there are exactly two sets $S_1,S_2$ with $v \in
   S_1,S_2$, such that their induced subgraph is isomorphic to $C_{8}$.
   }
   \label{fig:Pichler_Graph_red}
\end{figure*}

\subsubsection{Reduction of the blockade graph}

As a first step in determining the automorphism group, we construct a reduced
graph with the same automorphism group, i.e, we eliminate vertices that do not
contribute any degrees of freedom to the automorphism group.

First, observe that the vertices of the Zeng blockade graph $G_\PQDM =
(V_\PQDM,E_\PQDM)$ decompose into the two sets
\begin{subequations}
    \begin{align}
        M_1 &:= \{v \in V_\PQDM\, |\, \del B_1(v)\ \text{is connected}\}\\
        M_2 &:= \{v \in V_\PQDM\, |\, \del B_1(v)\ \text{is not connected}\}.
    \end{align}
\end{subequations}
The vertices in the set $M_2$ coincide with the green colored vertices, and the
vertices in the set $M_1$ coincide with the red colored vertices shown in
\cref{fig:Pichler_Graph_red}. By the construction of $G_\PQDM$, the vertex set $M_2$
decomposes into  a disjoint union of connected components $M_2 = \sqcup_{k}
Q_k$, the induced subgraph of each $Q_k$ is isomorphic to $K_4$, the complete
graph with $4$ vertices.  We now introduce an equivalence relation on
$G_\PQDM$:
\begin{align}
    v \equiv w :\Leftrightarrow v 
    = w \lor \exists k: v,w \in Q_k.
\end{align}
Symmetry and reflexivity of this relation are obvious. For transitivity
consider $v \equiv w$ and $w \equiv u$.  If $v = w$ or $w = u$, it is obvious
that $v \equiv u$. In the remaining case, note that $w$ is only part of exactly
one $Q_k$. It follows that $u, v \in Q_k$, as desired.

This equivalence relation partitions the vertex set $V_\PQDM$ into equivalence
classes $V_\PQDM^\text{red} = \{[v]| v \in V_\PQDM\}$, where $[v]$ denotes the
equivalence class of $v$. We define the reduced graph as $G_\PQDM^\text{red} =
(V_\PQDM^\text{red},E_\PQDM^\text{red})$ with the edge set $E_\PQDM^\text{red}
= \{\{[v],[u]\}\, |\, \{u,v\} \in E_\PQDM\}$. This construction is also known
as a \emph{quotient graph}. The resulting graph is shown in
\cref{fig:Pichler_Graph_red}.

\subsubsection{Establishing a bijection}

Now we show that there is a bijection between the automorphism group of the
reduced graph and the automorphism group of the original graph.

\begin{lemma}

   For $\phi \in \Aut{G_\PQDM}$, define
   \begin{align}
       \F_\phi : V_\PQDM^\text{red} \rightarrow V_\PQDM^\text{red},\ [u] \mapsto [\phi(u)].
   \end{align}
   Then the map $\F: \aut{G_\PQDM} \rightarrow \aut{G_\PQDM^\text{red}},\ \phi
   \mapsto \F_\phi$, is a bijection.

\end{lemma}

\begin{proof}
    First we have to show that $\F_\phi$ is well-defined. To this end, let $[u]
    = [v]$. Then either $u = v$ or $u, v \in Q_k$. Only the latter case is
    nontrivial. As $\phi$ preserves connectedness and $\phi(M_2) = M_2$, the
    images of $U$ and $v$ lie in the same connected component, i.e.,
    $\phi(u),\phi(v) \in Q_{k'}$ for some $k'$.  This implies that $[\phi(u)] =
    [\phi(v)]$.
 
    Second, we have to show that $\F_\phi$ defines a graph automorphism.  For
    bijectivity, note that the map $\F_{\phi^{-1}}$ inverts the map
    $\F_{\phi}$.  For edge transitivity, let $\{[u],[v]\} \in
    E_\PQDM^\text{red}$, so \Wlog we can assume that $\{u,v\} \in E_\PQDM$.
    This implies that $\{\phi(u),\phi(v)\} \in E_\PQDM$; hence,
    \begin{align}
	    \{\F_{\phi}([u]),\F_{\phi}([v])\} 
        = \{[\phi(u)],[\phi(v)]\} \in E_\PQDM^\text{red}.
    \end{align}
    Now we have to show the bijectivity of the map $\F$. To this end, we
    construct an inverse map $\G: \aut{G_\PQDM^\text{red}} \rightarrow
    \aut{G_\PQDM}$ such that $\G$ inverts $\F$.
 
    The key ingredient in this construction is that an automorphism $\phi \in
    \aut{G_\PQDM}$ is uniquely determined when its action on the set $M_2$ is
    known. Let $u \in M_1$; note from the structure of the graph that there
    exist exactly three vertices $v_1,v_2,v_3 \in M_2$ such that
    \begin{align}
        \label{eq:cap}
        \bigcap_{i \in \{1,2,3\}} \del B_1(v_i) 
        = \{u\}.
    \end{align}
    The vertex $u$ does not have any more adjacent vertices in $M_2$. The
    application of the automorphism $\phi$ on \ref{eq:cap} yields
    \begin{align}
        \label{eq:phi_u}
        \bigcap_{i \in \{1,2,3\}} \del B_1(\phi(v_i)) = \{\phi(u)\},
    \end{align}
    i.e., $\phi(u)$ is uniquely determined by the images $\phi(v_i)$.
 
    Now let $\Phi \in \Aut{G_\PQDM^\text{red}}$, we construct an
    automorphism $\G_\Phi \in \aut{G_\PQDM}$.  The vertex set
    $V_\PQDM^\text{red}$ decomposes into the components
    \begin{subequations}
        \begin{align}
            \tilde{M}_1 &:= \{[v] \in V_\PQDM^\text{red}\, |\, \textup{deg}([v]) = 2\}\\
            \tilde{M}_2 &:= \{[v] \in V_\PQDM^\text{red}\, |\, \textup{deg}([v]) = 4\}.
        \end{align}
    \end{subequations}
    Here $\textup{deg}(\cdot)$ denotes the \emph{degree} of a vertex, i.e.\ the
    number of vertices it is connected to. As the degree is invariant under
    automorphisms, so are the sets $\tilde{M}_1$ and $\tilde{M}_2$. From the
    construction of $G_\PQDM^\text{red}$ we obtain a different characterization
    of these sets,
    \begin{subequations}
        \begin{align}
            \tilde{M}_1 &= \{[v] \in V_\PQDM^\text{red}\, |\, |[v]| = 4\}\\
            \tilde{M}_2 &= \{[v] \in V_\PQDM^\text{red}\, |\, |[v]| = 1\}.
        \end{align}
    \end{subequations}
    Hence, for $[v] \in \tilde{M}_2$, $\G_\Phi([v])$ is an equivalence class
    with one element.  For $v \in M_2$, we also know that $[v] = 1$; therefore,
    $[v] \in \tilde{M}_2$ and hence $\Phi([v])$ only contains one element.  So we
    can define $\G_\Phi(v)$ such that $\Phi([v]) = \{\G_\Phi(v)\}$. For $u \in
    M_1$, we can now define $\G_\Phi(u)$ as the unique element in the
    intersection on the left side of \cref{eq:phi_u}.
 
    We have to show that the map $\G_\Phi$ defined this way is actually an
    automorphism.  For edge transitivity, let $\{v,v'\} \in E_\PQDM$. The first
    case is, \Wlog, $v' \in M_1$ and $v \in M_2$. In this case, $\G_\Phi(v')$ is
    defined such that $\{\G_\Phi(v),\G_\Phi(v')\} \in E_\PQDM$.  The other case
    is $v,v' \in M_1$; here we know that $v,v' \in Q_k$ for some $k$, so they
    share two neighbors $u_1,u_2 \in M_2$. From the previous case, it follows
    that $\G_\Phi(u_1)$ and $\G_\Phi(u_2)$ are neighbors of both $\G_\Phi(v)$
    and $\G_\Phi(v')$. This is only possible if $\G_\Phi(v)$ and $\G_\Phi(v')$
    lie in the same connected component $Q_{k'}$ for some $k'$.  This implies
    that $\{\G_\Phi(v),\G_\Phi(v')\} \in E_\PQDM$.
 
    To show bijectivity, we show that, for $\Phi \in \Aut{G_\PQDM^\text{red}}$,
    $\G_{\Phi^{-1}} = \G_{\Phi}^{-1}$. To this end, for $v \in M_2$, consider
    \begin{subequations}
        \begin{align}
            \{v\} 
            = [v] 
            &= \Phi(\Phi^{-1}([v]))\\
            &= \Phi(\{\G_{\Phi^{-1}}(v)\})\\
            &= \Phi([\G_{\Phi^{-1}}(v)])\\
            &= \{\G_{\Phi}(\G_{\Phi^{-1}}(v))\},
        \end{align}
    \end{subequations}
    where we have used the fact that the equivalence classes of $v$ and $\G_{\Phi^{-1}}(v)$
    only contain one element.  As an automorphism in $\aut{G_\PQDM}$ is
    uniquely determined by its action on the set $M_2$ and the action of
    $\G_{\Phi}\circ\G_{\Phi^{-1}}$ on $M_2$ is the identity, we infer that
    $\G_{\Phi}\circ\G_{\Phi^{-1}}$ acts as the identity on $V_\PQDM$. The proof
    for $\G_{\Phi^{-1}}\circ\G_{\Phi}$ is analogous.
    
    Now we prove that $\G$ is the inverse of $\F$. Let $\phi \in
    \text{Aut}(G_\PQDM)$. For $v \in M_2$, we obtain
    \begin{align}
        \{\G(\F(\phi))(v)\} 
        = \F(\phi)([v]) 
        = [\phi(v)] 
        = \{\phi(v)\}.
    \end{align}
    As this uniquely determines the automorphism, $\phi = \G(\F(\phi))$.
 
    On the other hand, starting with an automorphism $\Phi \in
    \text{Aut}(G_\PQDM^\text{red})$, for $v \in M_2$, we obtain
    \begin{align}
        \F(\G(\Phi))([v]) 
        = [\G(\Phi)(v)] 
        = \Phi([v]).
    \end{align}
    This uniquely determines the action of $\F(\G(\Phi))$, because for any
    vertex $[u] \in \tilde{M}_1$, we find two neighbors $[v_1],[v_2] \in
    \tilde{M}_2$ such that
    \begin{align}
        \partial B_1([v_1]) \cap \partial B_1([v_2]) 
        = \{[u]\}.
    \end{align}
    Therefore, for any automorphism $\Phi' \in \text{Aut}(G_\PQDM^\text{red})$,
    we find that
    \begin{align}
        \partial B_1(\Phi'([v_1])) \cap \partial B_1(\Phi'([v_2])) = \{\Phi'([v])\}.
    \end{align}
    Therefore, we conclude that $\F(\G(\Phi)) = \Phi$. 
\end{proof}

\subsubsection{Characterization of the automorphism group}

We now characterize the automorphism group
$\text{Aut}(G_\PQDM^\text{ret})$.  Suppose that we have a square lattice with
periodic boundary conditions and side lengths $N_x,N_y > 2$. As we only
deal with one graph in this section, we write $u \sim v$ whenever $\{u,v\} \in
E_\PQDM^\text{ret}$.

Pick any vertex $v \in V_\PQDM^\text{ret}$ with degree $2$.  Then, for this
vertex, there exist exactly two subsets $S_1,S_2 \subset V_\PQDM^\text{ret}$
such that $v \in S_1,S_2$ and the induced subgraph of $S_1$ and $S_2$ is
isomorphic to the cycle graph $C_8$. To prove this, we note that a set $S$ with
these properties must be contained in the neighborhood $B_4(v)$. This
neighborhood is shown in \cref{fig:Pichler_Graph_red} (here we used the
assumption that $N_x,N_y > 2$). Also, $S$ cannot contain any pendant vertices (i.e.\
vertices with only one neighbor) in $B_4(v)$, so we remove them. From the
resulting set, we again remove all pendant vertices and repeat until we reach a
set with no pendant vertices.  These removed vertices are colored gray in
\cref{fig:Pichler_Graph_red}. In the resulting graph it is clear, that there
exist exactly two subsets with the desired property.

In the following argument, we make use of the next lemma.

\begin{lemma}
   \label{lm:cycle_graph_unique}

   Let $G_1,G_2 \simeq C_k$ for some $k \geq 3$, and $x_1,x_2 \in G_1$ such
   that $x_1 \sim x_2$.  Then any graph isomorphism $\phi: G_1 \rightarrow G_2$
   is uniquely determined by $\phi(x_1)$ and $\phi(x_2)$.

\end{lemma}

\begin{proof}
    We write $G_1 = \{x_1,\ldots,x_k\}$ and $G_2 = \{y_1,\ldots,y_k\}$ such
    that, for all $l \leq k$, $x_{l} \sim x_{l+1}$ and $y_{l} \sim y_{l+1}$,
    with the boundary condition $k+1 = 0$.

    Suppose that we already determined the images $\phi(x_1),\ldots,\phi(x_l)$ for
    $l \geq 2$. We know that $x_{l-1} \sim x_l \sim x_{l+1}$; therefore,
    $\phi(x_{l-1}) \sim \phi(x_l) \sim \phi(x_{l+1})$.  As $G_2$ is a cycle
    graph, each vertex has exactly two neighbors. Therefore, $\phi(x_{l+1})$ is
    the unique neighbor of $\phi(x_l)$ which is not $\phi(x_{l-1})$. By
    induction, this shows that all images $\phi(x_l)$ for $l\geq 2$ can be
    determined.
\end{proof}

Let $\Phi \in \text{Aut}(G_\PQDM^\text{ret})$ and consider the image $\Phi(v)$.
For this vertex, exactly two vertex sets $S_1', S_2'$ exist such that
$\Phi(v) \in S_1',S_2'$ and the induced subgraph of both sets is isomorphic to
$C_8$. The sets $\Phi(S_1)$ and $\Phi(S_2)$ are also isomorphic to $C_8$, as
$\Phi$ is an automorphism. Furthermore, $\Phi(v) \in \Phi(S_1),\Phi(S_2)$.
Thus, there are two possibilities:
\begin{subequations}
    \begin{align}
        S_1' = \Phi(S_1) &\text{ and } S_2' = \Phi(S_2)\\
        S_1' = \Phi(S_2) &\text{ and } S_2' = \Phi(S_1).
    \end{align}
\end{subequations}
Let $v_1,v_2$ be the two adjacent vertices of $v$; they lie in $S_1$ and $S_2$.
Let $w_1,w_2$ be the two adjacent vertices of $\Phi(v)$; then there are again
two possibilities:
\begin{subequations}
    \begin{align}
	    w_1 = \Phi(v_1)& \text{ and } w_2 = \Phi(v_2)\\
	    w_1 = \Phi(v_2)& \text{ and } w_2 = \Phi(v_1).
    \end{align}
\end{subequations}
So, overall there are four possibilities. For each of these,
\cref{lm:cycle_graph_unique} implies that the images of all vertices in $S_1$
and $S_2$ are uniquely determined.

Now we show that this completely determines the automorphism $\Phi$.  Take some
vertex $u \in S_1$ with $u \neq v$. Then there exists another set of vertices
$S_3$ such that its induced subgraph is isomorphic to $C_8$ and $S_3 \neq
S_1$.  Likewise, there exists a set of vertices $S_3'$ such that $\Phi(u)
\in S_3'$, $S_3' \neq S_1'$ and the induced subgraph of $S_3'$ is isomorphic to
$C_8$.  So we know that $\Phi(S_3) \neq S_1'$, as $\Phi$ is bijective.
Therefore the only possibility is $\Phi(S_3) = S_3'$. As $u$ has only two
adjacent vertices, they also lie in $S_1$ and $S_3$.  Because they lie in $S_1$,
their image under $\Phi$ is already determined. Because they lie in $S_3$, it
follows from \cref{lm:cycle_graph_unique} that the image of every vertex in
$S_3$ is uniquely determined. This argument can be repeated for every degree
two vertex, whose image is already determined.

Now we can count the possible automorphisms. For $v\in V_\PQDM^\text{red}$,
there are $2N_xN_y$ possibilities for $\Phi(v)$. After that, there are $4$
additional possibilities that lead to $8N_xN_y$ possible automorphisms, so
$|\Aut{G_\PQDM}| = |\Aut{G_\PQDM^\text{red}}| = 8N_xN_y$. The size of this
automorphism group is actually as small as possible.  For a tessellated square
lattice of lengths $N_x,N_y$, there are $N_xN_y$ distinct translation
automorphisms.  Furthermore global reflections and rotations form a dihedral
group $D_8$, and combinations lead to at least $8N_xN_y$ distinct automorphisms.
Via the bound in \cref{app:burnside}, this result implies that, in this
structure, the number of orbits $|L_\C/\A_\C|$ grows exponentially with the
system size.

The requirement that $N_x,N_y > 2$ is actually important. If this is not the
case, additional automorphisms appear. The action of these can be pictured in
the following way. Take two adjacent unit cells in the direction with length
$2$. Then these unit cells actually share two vertices instead of one. The map
that exchanges these two vertices is a graph automorphism.

Even with these additional automorphisms, the automorphism group is in general
not large enough such that there is only one orbit. We checked this explicitly
for a system size of $N_x = 3$, $N_y = 2$ (this was one of the systems
diagonalized in Ref.~\cite{Zeng2025}), in which case the ground state space 
$L_{\C_\PQDM}$ decomposes into two distinct orbits.

\subsection{Stastny model}
\label{app:stastny}

\begin{figure*}[tb]
   \centering
   \includegraphics[width=0.8\linewidth]{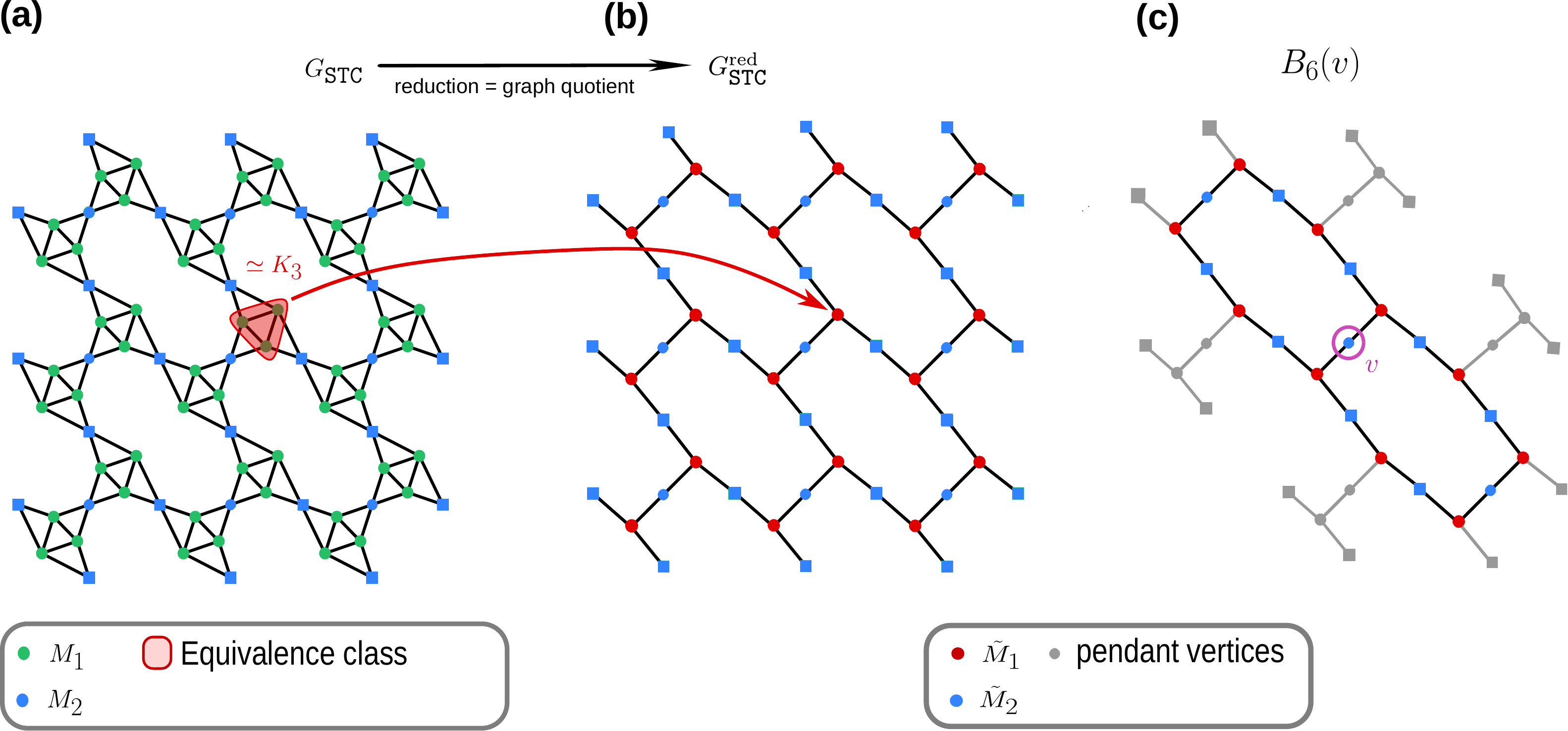}
   \caption{%
   \emph{Reduction of the Stastny model.} 
   \textbf{(a)} Shown is the blockade graph proposed by Stastny \etal
   $G_\texttt{STC}$ to realize the Toric Code with a blockade structure.
   Vertices in the set $M_1$ are colored green, while those in the set $M_2$ are
   colored blue. The set $M_1$ decomposes into connected components $Q_k$; the
   induced subgraph of every such component is isomorphic to $K_3$. By an
   equivalence relation we identify the vertices that lie in the same component
   $Q_k$. Some of the equivalence classes are shown as red boxes.
   \textbf{(b)} Shown is the quotient graph of $G_\texttt{STC}$ by the aforementioned
   equivalence relation; the resulting graph is denoted $G_\texttt{STC}^\text{red}$.
   We prove in this section, that the automorphism groups of both graphs have the
   same size.
   \textbf{(c)} Shown is the neighborhood $B_6(v)$ of some vertex $v\in \tilde{M}_1$.
   From this neighborhood, all pendant vertices are removed. From the remaining
   graph, again, all pendant vertices are removed. This process is repeated
   until no pendant vertices are left. The removed vertices are colored
   gray. This shows that there are exactly two sets $S_1,S_2$ with $v \in
   S_1,S_2$ such that their induced subgraph is isomorphic to $C_{12}$.
   }
   \label{fig:Stastny_Graph_red}
\end{figure*}

Stastny \etal~\cite{Stastny2023a} proposed a blockade structure
$\C_\texttt{STC}$ to realize the toric code. This structure is shown in
\cref{fig:Stastny_Graph_red}. In this section we characterize the automorphism
group of its blockade graph $G_\texttt{STC}$ and show that it is not
fully-symmetric. We use the same methods as developed in \cref{app:pichler}.
As all arguments are quite similar, we will not give them in full detail.

We decompose $V_\texttt{STC}$ into the subsets,
\begin{subequations}
    \begin{align}
        M_1 &:= \{v \in V_\texttt{STC}\, |\, \del B_1(v)\ \text{is connected}\}\\
        M_2 &:= \{v \in V_\texttt{STC}\, |\, \del B_1(v)\ \text{is not connected}\}.
    \end{align}
\end{subequations}
The vertices in $M_1$ are colored green in \cref{fig:Stastny_Graph_red}, and the
vertices in $M_2$ are colored blue. We observe, that the set $M_1$ decomposes
into connected components $M_1 = \sqcup_{k} Q_k$, where $Q_k \simeq K_3$. By
introducing the equivalence relation
\begin{align}
   \label{eq:eq_rel}
   v \equiv w :\Leftrightarrow v = w \lor \exists k: v,w \in Q_k,
\end{align}
we obtain the reduced graph $G_\texttt{STC}^\text{red} =
(V_\texttt{STC}^\text{red},E_\texttt{STC}^\text{red})$ as the quotient graph of
$G_\texttt{STC}$ under \cref{eq:eq_rel}.

The reduced graph is shown in \cref{fig:Stastny_Graph_red}. Via analogous
arguments as in \cref{app:pichler} the automorphism group of
$G_\texttt{STC}^\text{red}$ has the same size as the automorphisms group of
$G_\texttt{STC}$.

To explicitly characterize this automorphism group, we assume that we
assume, that $N_x, N_y \geq 3$.  Note that, for every vertex $v \in
\tilde{M}_2$, there are exactly two sets of vertices $S_1,S_2$ with $v \in
S_1,S_2$ such that their induced subgraph is isomorphic to $C_{12}$.  To see
this, we only have to consider the neighborhood $B_6(v)$. From this
neighborhood, all pendant vertices (vertices with only one neighbor) can be
removed.  The resulting graph is shown in \cref{fig:Stastny_Graph_red}; it
contains the two cycles graphs that contain $v$.

Let $\Phi \in \aut{G_\texttt{STC}^\text{red}}$; then, for $\Phi(v)$, there
are also exactly two sets of vertices $S_1',S_2'$ with $v \in S_1',S_2'$ such that
their induced subgraph is isomorphic to $C_{12}$. Also, $\Phi(S_1)$ and
$\Phi(S_2)$ have the same property as $\Phi$ is an automorphism. Therefore,
there are the possibilities
\begin{subequations}
    \begin{align}
        S_1' = \Phi(S_1)&,\;S_2' = \Phi(S_2)\\
        S_1' = \Phi(S_2)&,\;S_2' = \Phi(S_1).
    \end{align}
\end{subequations}
Let $v_1,v_2$ be the two adjacent vertices of $v$; they lie in $S_1$ and $S_2$.
Let $w_1,w_2$ be the two adjacent vertices of $\Phi(v)$; then there are again
two possibilities:
\begin{subequations}
    \begin{align}
        w_1 = \Phi(v_1)&,\;w_2 = \Phi(v_2)\\
        w_1 = \Phi(v_2)&,\;w_2 = \Phi(v_1).
    \end{align}
\end{subequations}
So, overall, there are four possibilities. For  each of these,
\cref{lm:cycle_graph_unique} implies that the images of all vertices in $S_1$
and $S_2$ are uniquely determined.

Now we show that this completely determines the automorphism $\Phi$.  Take some
vertex $u \in S_1$ with $u \neq v$. Then there exists another set of vertices
$S_3$ such that its induced subgraph is isomorphic to $C_{12}$ and $S_3 \neq
S_1$.  Likewise, there exists such a set of vertices $S_3'$ such that $\Phi(u)
\in S_3'$, $S_3' \neq S_1'$ and the induced subgraph of $S_3'$ is isomorphic to
$C_{12}$.  So we know that $\Phi(S_3) \neq S_1'$, as $\Phi$ is bijective.
Therefore, the only possibility is $\Phi(S_3) = S_3'$. As $u$ has only two
adjacent vertices, they also lie in $S_1$ and $S_3$.  Because they lie in $S_1$,
their image under $\Phi$ is already determined. Because they lie in $S_3$, it
follows from \cref{lm:cycle_graph_unique} that the image of every vertex in
$S_3$ is uniquely determined. This argument can be repeated for every degree
two vertex, whose image is already determined.

We can now count the automorphisms of this graph. Choosing an arbitrary
reference vertex $v \in \tilde{M}_2$, there are $3N_xN_y$
possibilities for $\Phi(v)$, each leading to four possible automorphisms (every
possibility described above is actually a valid automorphism but this is not
crucial). Therefore $|\aut{V_\texttt{STC}}| = |\aut{V_\texttt{STC}^\text{red}}|
= 12N_xN_y$.

\section{Ground state of $H_\sub{Loop}$ for $\Omega\neq 0$}
\label{app:spectrum}

In this section, we prove rigorously that by switching on $\Omega$, the gap in
the spectrum of $\tilde{H}_\text{Loop}(0,\omega)$ does not close. To this end,
we apply a result by Michalakis \etal~\cite{Michalakis2013} that provides
sufficient conditions for the stability of a spectral gap under perturbations.

\subsection{Conditions for gap stability}
\label{app:cond}

We briefly summarize the conditions for gap stability given by Michalakis.
The system considered is defined on the lattice $\Lambda_L := [0,L]^2
\subseteq \Z^2$, which is endowed with the usual graph metric (= the $\ell^1$
metric). We refer to $L$ as the system size. To each site $I \in \Lambda_L$, a
Hilbert space $\mathcal{H}_I$ is associated; the system Hilbert space is given
by $\mathcal{H}_{\Lambda_L} := \bigotimes_{I \in \Lambda_L} \mathcal{H}_I$. The
unperturbed system Hamiltonian $H_0$ is assumed to satisfy the following
properties.
\begin{enumerate}

    \item $H_0$ can be written as a sum of local operators
    \begin{align}
        H_0 := \sum_{I \in \Lambda} Q_I,
    \end{align}
    where each operator $Q_I$ has a constant range of support.

    \item The Hamiltonian $H_0$ satisfies periodic boundary conditions.

    \item The projection $P_0$ onto the ground state subspace of $H_0$
        satisfies $P_0 Q_I = q_{0,I} Q_I$, where $q_{0,I}$ denotes the minimal
		eigenvalue of $Q_I$ (Frustration freeness).

    \item For all $L \geq 2$, $H_0$ has a spectral gap $\gamma > 0$, which is
        independent of the system size.

\end{enumerate}
We denote the support of an operator by $\operatorname{supp}(Q_I)$.

The perturbation $V$ is assumed to have the form 
\begin{align}
    \label{eq:pert}
    V = \sum_{I \in \Lambda_L} \sum_{r = 0}^L V_I(r)
\end{align}
such that $\operatorname{supp}(V_I(r)) \subseteq B_r(I)$ for all $r$. Moreover, for
allowed perturbations, there exists a rapidly decaying function $f(r)$ such
that $\|V_I(r)\| \leq Jf(r)$ for some constant $J > 0$. Such a perturbation is
called a $(J,f)$-perturbation.  These conditions are trivially satisfied if
\cref{eq:pert} only contains terms up to a finite range. That is, there is a
constant $R > 0$ independent of the system size such that $V_I(r) = 0$ for all
$r \geq R$. The perturbation we consider in the following sections has this 
property.

There are two more conditions for gap stability: the \emph{local-gap}
condition and \emph{local TQO}.  For any set $A
\subseteq \Lambda_L$, we define
\begin{align}
   \label{eq:H_local}
   H_{A} = \sum_{\operatorname{supp}(Q_I) \subseteq A} Q_{I}.
\end{align}
Furthermore, for $\epsilon > 0$, let $P_B(\epsilon)$ be the projection onto the
eigenstates of $H_A$ with energy less than or equal to $\epsilon$.

A Hamiltonian $H$ satisfies the \emph{local-gap condition} iff there exists an
at most polynomial decaying function $\gamma(r) > 0$ such that, for all $I_0
\in \Lambda_L$, $P_{B_r(I_0)}(\gamma(r)) = P_{B_r(I_0)}(0)$. Here $B_r(I_0)$
denotes the ball with radius $r$.

Lastly we define the condition \emph{local topological quantum order (local TQO)}.
Let $I_0\in \Lambda$ and $A =
B_r(I_0)$, $A(l) = B_{r+l}(I_0)$ for some $r \leq L^*<L$ and $l \leq L-r$,
where $L$ denotes the system size. The cut-off parameter $L^*$ has to scale
with the system size $L$. Let $\ket{\psi_1},\ket{\psi_2}$ be two ground states
of $H_{A(l)}$, and let $\rho_i(A) :=
\Tr{\ket{\psi_i}\bra{\psi_i}}{A(l)\setminus A}$ for $i = 1,2$ be their
corresponding reduced density matrices. Then a system satisfies local TQO iff
\begin{align}
    \label{eq:LocalTQO}
    \|\rho_1(A) - \rho_2(A)\|_1 \leq 2F(l),
\end{align}
where $F$ is a decaying function and $\|\cdot\|_1$ is the Schatten-1 norm.

These conditions are sufficient to guarantee the stability of the spectral gap
of $H_0$.

\begin{theorem}[Michalakis~\cite{Michalakis2013}]

    Let $H_0$ be a Hamiltonian satisfying conditions 1.-4., local TQO, and
    the local-gap condition. Let $V$ be a $(J,f)$-perturbation. Then there exist constants
    $J_0 > 0$ and $L_0 \geq 2$ such that, for $J \leq J_0$ and $L \geq L_0$,
    the spectral gap of $H_0 + V$ is bounded from below by $\gamma/2$.

\end{theorem}

\subsection{Locality and frustration-freeness}
\label{app:loc_frust}

\begin{figure*}
    \centering
   \includegraphics[width=0.9\linewidth]{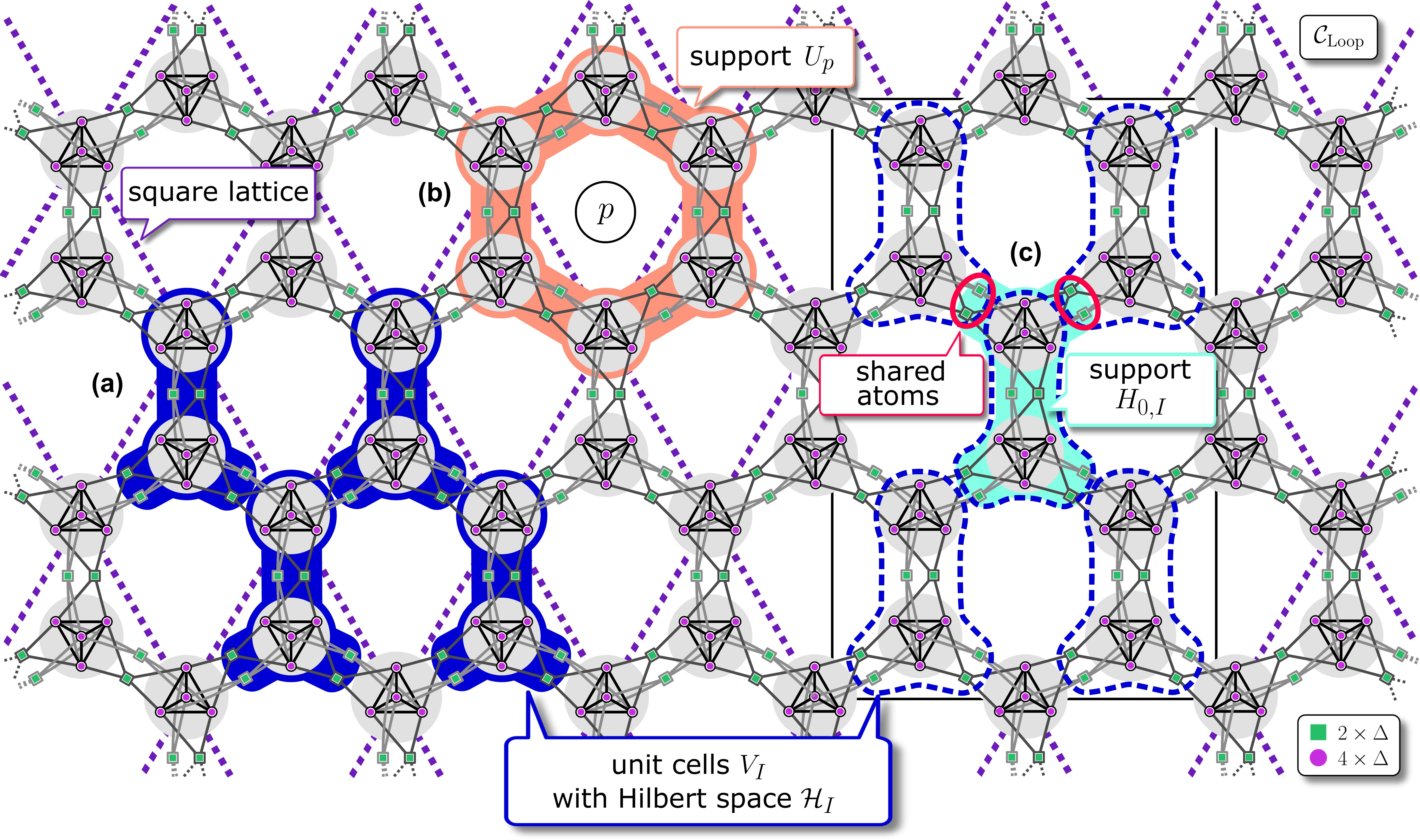}
   \caption{%
   \emph{Coarse graining of $\C_\text{Loop}$.} 
   This figure summarizes the constructions in \cref{app:loc_frust}, to show
   that the Hamiltonian $\tilde{H}_0$ can be viewed as frustration-free. The
   fill color of the atoms denotes their detuning in $\C_\sub{Loop}$ (green:
   $2\Delta$, pink: $4\Delta$).
   \textbf{(a)} The vertex set $V_\text{Loop}$ is partitioned into unit cells
   $V_I$; these are highlighted in dark blue. The structure $\C_\sub{Loop}$ can 
   be equivalently
   described by "contracting" these unit cells into larger Hilbert spaces
   $\H_I$, that are attached to sites in a square lattice. This square lattice
   is shown as dashed violet lines. 
   \textbf{(b)} The support of a plaquette operator $U_p$ is highlighted in
   orange. This support contains part of four unit cells.
   \textbf{(c)} The Hamiltonian $H_0^\sub{R}$ is partitioned into local terms
   $H_{0,I}$. The support of such an operator is highlighted in cyan. The
   dashed dark blue lines hint the surrounding unit cells. The red ellipses
   highlight the atoms where $H_{0,I}$ acts on its neighboring unit cells.
   }
   \label{fig:coarse_graining}
\end{figure*}
In the following sections, we show that the conditions stated in
\cref{app:cond} are satisfied for our setup. Our fist task is to define suitable
local Hilbert spaces and a suitable local decomposition of $\tilde{H}_0 :=
\tilde{H}_\sub{Loop}(0,\omega)$ and then verify frustration-freeness and the
local-gap condition.

We partition our Hamiltonian $\tilde{H}_0$ as
\begin{align}
    \begin{aligned}
        \label{eq:H_unp}
        \tilde{H}_0 &= \underbrace{-\sum_{i \in V_\sub{Loop}} \Delta_i n_i 
        + U \sum_{\{i,j\} \in E_\sub{Loop}} n_i n_j}_{:= H_{0}^\sub{R}} \\
        &\quad + \underbrace{\frac{\omega}{2}\sum_{\text{Faces }p} (\mathds{1} - U_p)}_{:= H_{0}^\text{p}}.
    \end{aligned}
\end{align}
We note, that \cref{eq:H_unp} already contains a local decomposition of
$\tilde{H}_0$.  However, this form does not satisfy frustration-freeness, as
the blockade competes with the detuning. We solve this problem, by
partitioning the individual atoms into larger blocks, such that the
frustration-freeness is recovered in this ``coarse-grained'' picture.

\subsubsection{Hilbert space}
\label{app:Hilbert}

As we are working with a finite blockade strength, the system Hilbert space is
given by the tensor product $\H = \bigotimes_{i \in V_\sub{Loop}} \H_i$, where
$\H_i \simeq \mathbb{C}^2$, is the Hilbert space of one atom. 

As shown in \cref{fig:coarse_graining}~(a), we partition $\C_\sub{Loop}$ into
unit cells, consisting of two tetrahedrons and the "wings" attached to one of
the tetrahedrons.  Furthermore, \cref{fig:coarse_graining}~(a) shows that we can
view these unit cells as attached to the vertices of a square lattice. We
denote this square lattice as $\Lambda$.  To avoid confusion with the vertices
of $\C_\sub{Loop}$, we refer to the vertices of $\Lambda$ as \emph{sites}. We
denote by $V_I$ the collection of atoms contained in the unit cell attached to the
site $I \in \Lambda$. We endow this square lattice with the usual
graph metric (cf.~\cref{app:graphs}) and write $I \sim J$ for neighboring
sites. The square lattice is natural to this construction, as for two vertices
$i \in V_I$ and $j \in V_J$, $i \sim j$ is only possible if $I = J$ or $I \sim
J$. This means that the coarse graining is compatible with the local structure of
the graph. This construction can be immediately generalized to the extended
tessellations $\C_\sub{Loop}^\sub{ext}$ described in \cref{sec:embedding}.

The Hilbert space describing one unit cell $V_I$ is given by the tensor product
\begin{align}
    \H_I := \bigotimes_{i \in V_I} \H_i \simeq \mathbb{C}^{28}.
\end{align}
The last equality follows as one unit cell contains $14$ atoms. For
$\C_\sub{Loop}^\sub{ext}$ (cf. \cref{sec:embedding}), this number is modified 
to $6K + 8$, where $K$ denotes the length of the links. By associativity of 
the tensor product, the Hilbert space $\H$ has the form
\begin{align}
    \H = \bigotimes_{I \in \Lambda} \H_I.
\end{align}

\subsubsection{Frustration-freeness of the blockade Hamiltonian}
\label{app:frust_block}

The next step is to construct a frustration-free decomposition of
$\tilde{H}_0$. We first focus on the part $H_0^\sub{R}$ that describes the
blockade. The structure $\C_\sub{Loop}$ is constructed as an amalgamation
of \texttt{FSU}-structures. The amalgamation has the nice property that the
resulting structure is frustration-free, and thus the idea underlying the following
constructions is to "reverse" the amalgamation and fit the emerging pieces into
a square lattice.

For a site $I$, we define the two sets $\text{Int}(V_I)$ and $\partial V_I$ 
as shown in the folloing figure. Note that $\partial V_I$ is not a subset of $V_I$.
\begin{center}
    \includegraphics[width=0.5\linewidth]{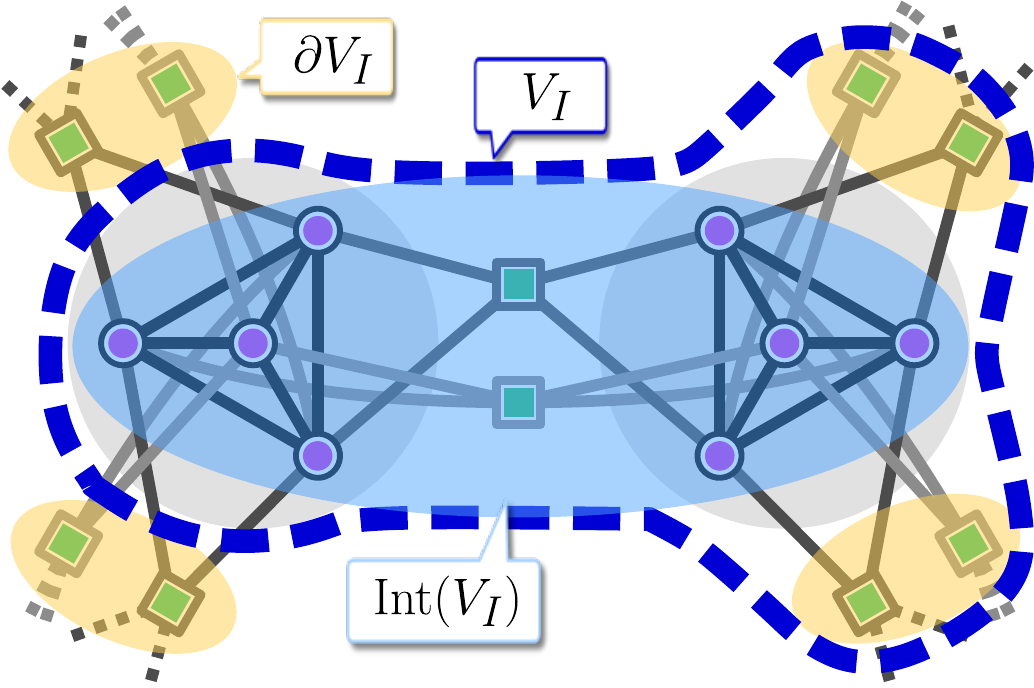}
\end{center}
We define local Hamiltonians associated to the site $I \in \Lambda$, whose
support is contained in $B_1(I)$, by
\begin{align}
    \begin{aligned}
        \label{eq:H0local}
    H_{0,I} &:= -\sum_{i \in \operatorname{Int}(V_I)} \Delta_i n_i  
        -\sum_{i \in \del V_I} \frac{\Delta_i}{2} n_i\\
        &\quad + U\sum_{\substack{i,j\in V_I\\i \sim j}}n_in_j.
    \end{aligned}
\end{align}
The support of these Hamiltonians is visualized in
\cref{fig:coarse_graining}~(b).

Any atom $i \in V_I$ is either contained in the interior of exactly one site or
the boundary of exactly two sites. Moreover, there are no blockades between
atoms in $\partial V_I$; thus, these local Hamiltonians form a decomposition of
$H_0$, as
\begin{align}
    \label{eq:HZZ_local_decomp}
    H_0^\sub{R} = \sum_{I \in \Lambda} H_{0,I}.
\end{align}
We now show the frustration-freeness of \cref{eq:HZZ_local_decomp}, i.e., a
ground state $\ket{E_0^\sub{R}}$ of $H_0^\sub{R}$ satisfies $H_{0,I}
\ket{E_0^\sub{R}} = E_{0,I} \ket{E_0^\sub{R}}$, where $E_{0,I}$ denotes the
minimal eigenvalue of $H_{0,I}$. 

All operators $H_{0,I}$ and $H_0^\sub{R}$ are diagonal in the configuration
basis.  Hence, the ground state energy of $H_0^\sub{R}$ is lower bounded by the
sum of the ground state energies of $H_{0,I}$. We show that this lower bound
is actually attained. This then implies that every ground state of
$H_0^\sub{R}$ is a ground state of $H_{0,I}$ for all $I \in \Lambda$.

\begin{figure}
   \centering
   \includegraphics[width=0.8\linewidth]{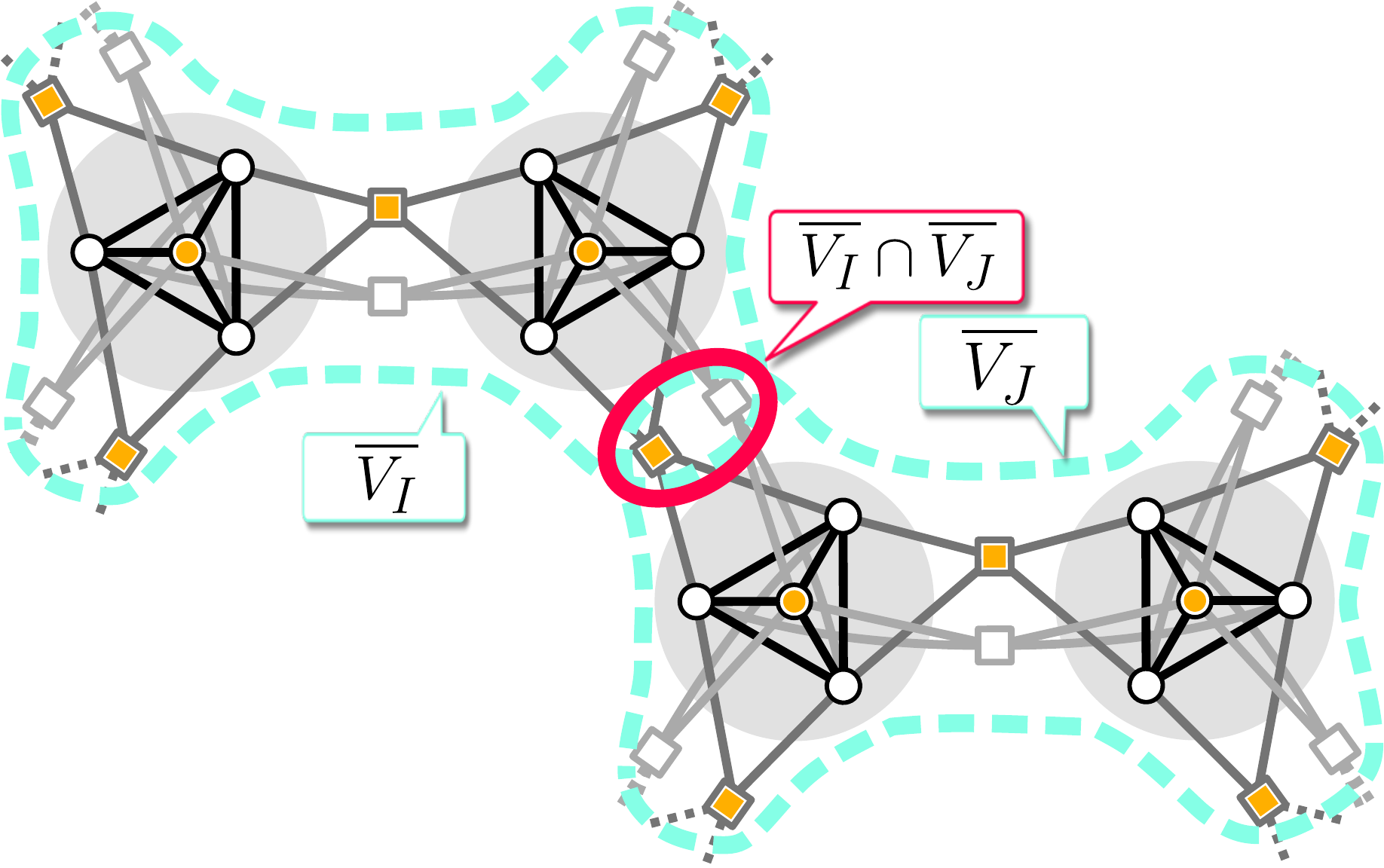}
   \caption{%
   \emph{Ground state configuration on $V_\sub{Loop}$.} Shown is a section of
   $\C_\sub{Loop}$ consisting of two sets $\overline{V_I}$ and
   $\overline{V}_J$ (marked with cyan dashed lines). The fill colors of the
   atoms denotes their state (white: unexcited, orange: excited). The shown
   configuration realizes a ground state of $H_{0,I}$ and $H_{0,J}$. Focusing
   on only $\overline{V}_I$, the configuration assigns the pair of edge atoms on the top
   left and the pair of edge atoms on the bottom right the same configuration
   (dark gray: excited, light gray: unexcited). The same is true for the edges
   on the top right and the bottom left. Thus, it is possible to assign the
   identical configuration to $\overline{V_J}$ with $J \sim I$, as both
   assignments assign the same state to the atoms in $\overline{V_I} \cap
   \overline{V_J}$ (red ellipse).
   }
   \label{fig:ground_state_conf}
\end{figure}

The operator $H_{0,I}$ is supported on the set $\overline{V_I} :=
\operatorname{Int}(V_I) \cup \del V_I$. Thus, the requirement that a configuration is a
ground state of $H_{0,I}$ only constrains the state of the atoms in
$\overline{V_I}$. 

Let $\C_{2\texttt{FSU}}$ be the structure arising as the amalgamation of an
\texttt{FSU}-structure and its mirror image, such that equivalent atoms on two
pairs of edges are identified (cf.~\cref{sec:const}). Thus we know the
configurations on $\overline{V_I}$ that are ground states of $H_{0,I}$. One
such configuration is shown in \cref{fig:ground_state_conf}. It is possible to
extend this to all of $V_\sub{Loop}$, by assigning this configuration to
$\overline{V_I}$ for \emph{all} $I \in \Lambda$. As $\overline{V_I}$ and
$\overline{V_J}$ have nonempty intersection if $I \sim J$, it is not
immediately clear that this is possible. Luckily, whenever $I \sim J$,
assigning this configuration to $\overline{V_I}$ or $\overline{V_J}$ results
in the same configuration of the atoms in $\overline{V_I} \cap \overline{V_J}$,
as shown in \cref{fig:ground_state_conf}.

This configuration on $V_\sub{Loop}$ realizes the lower bound for the ground
state energy described above and is thus a ground state of $H_0^\sub{R}$. 

\subsubsection{Frustration-freeness of the full Hamiltonian}
\label{app:frust_full}

Lastly, we focus on the term $H_{0}^\sub{p}$, which introduces new local
operators $\mathds{1} - U_p$. To be consistent with the conventions
in Ref.~\cite{Michalakis2013}, we attach these operators to Hilbert spaces
in our square lattice. To this end, we extend our square lattice $\Lambda$, by
introducing additional sites on the plaquettes of $\Lambda$ and the edges of
$\Lambda$. We refer to the former set of sites as $\Lambda_\text{e}$ and the
latter as $\Lambda_\text{p}$. We refer to this extended lattice as
$\Lambda^\text{ext}$.

To the sites $I \in \Lambda_\text{e}$ and $I \in \Lambda_\text{p}$, we attach
the trivial Hilbert spaces $\H_I \simeq \{0\}$. This is just a formality, to be
consistent with Ref.~\cite{Michalakis2013}; the Hilbert space is unchanged,
\begin{align}
    \H \simeq \bigotimes_{I \in \Lambda^\text{ext}} \H_I.
\end{align}
Note that the support of the operators $H_{0,I}$, defined in \cref{eq:H0local},
is now contained in $B_2(I)$ in the extended lattice $\Lambda^\text{ext}$. In
the extended lattice, it is natural to attach the operator $\mathds{1} - U_p$
to the sites in $\Lambda_\text{p}$. Then the support of these operators is
contained in $B_2(p)$, as shown in \cref{fig:coarse_graining}~(b). We
subsequently define the local operators $Q_I$ by
\begin{align}
    Q_{I} := 
    \begin{cases} 
        \frac{\omega}{2}(\mathds{1} - U_p), &  I = p \in \Lambda_\text{p}\\
        H_{0,I}, &  I \in \Lambda\\
        0, & I \in \Lambda_\text{e}.
    \end{cases}
\end{align}
Then the Hamiltonian~\eqref{eq:H_unp} can be written as
\begin{align}
    \label{eq:redef_H}
    \tilde{H}_0 = H_0^\sub{R} +  H_0^\sub{p} = \sum_{I \in \Lambda} Q_{I}.
\end{align}
We proceed to show the frustration-freeness of \cref{eq:redef_H}. By
construction (cf.~\cref{sec:localz2}), the loop automorphisms $\phi_p$, that
induce the operators $U_p$ are compositions of edge permutations and vertex
permutations (permutations of vertices in a tetrahedron). Thus, it is easily
verified that $i \in \del V_I$ implies that $\phi_p(i) \in \del V_I$, and
likewise for $i \in \operatorname{Int}(V_I)$. This implies that the operators
$H_{0,I}$ are invariant under $U_p$, i.e., $[H_{0,I}, U_p] = 0$. Moreover, we
know that $[H_{0,I},H_{0,J}] = 0$, as the operators $H_{0,I}$ are diagonal in
the same basis and $[U_p,U_q] = 0$, as discussed in \cref{sec:localz2}. Hence,
all operators $Q_I$ are simultaneously diagonalizable. It follows that the
ground state energy of $\tilde{H}_0$ is lower bounded by the sum of the ground
state energies $E_0[Q_I]$ of the operators $Q_I$.

As $U_p^2 = \mathds{1}$, the operator $\tfrac{1}{2}(\mathds{1}- U_p)$ is a 
projector. Thus, the smallest eigenvalue of this operator is $0$. An eigenstate $\ket{\psi}$
with this eigenvalue must satisfy $U_p \ket{\psi} = \ket{\psi}$. It is easily
verified that the state 
\begin{align}
   \ket{\omega} 
   := \frac{1}{\sqrt{|L_\sub{Loop}|}}\sum_{\vec n\in L_\sub{Loop}}\ket{\vec n}
\end{align}
satisfies $U_p \ket{\omega} = \ket{\omega}$ and, thus, $Q_I \ket{\omega} = 0$ for
$I \in \Lambda_\text{p}$. Here $|L_\sub{Loop}|$ denotes the number of elements
in $L_\sub{Loop}$. Note that the state $\ket{\omega}$ is denoted $\ket{\Omega_0}$ 
in the main text. All states $\ket{\vec{n}} \in \H_\sub{Loop}$ are, by
definition ground states of $H_0^\sub{R}$. It follows that $\ket{\omega}$ is
also a ground state of $H_0^\sub{R}$. Then the previous discussion implies
that $\ket{\omega}$ is a ground state of  all operators $Q_I = H_{0,I}$ for $I
\in \Lambda$. Hence we have shown the frustration-freeness of
\cref{eq:redef_H}.

In addition to frustration-freeness, the Hamiltonian~\eqref{eq:redef_H} has a
spectral gap independent of the system size. An excited state must either fail
to be a ground state of an operator $H_{0,I}$, leading to a spectral gap of
$\Delta$, or fail to be a ground state of $\omega/2(\mathds{1} - U_p)$, leading
to a spectral gap of $\omega$. Hence, $\tilde{H}_0$ has a spectral gap of
$\min(\Delta, \omega)$, independent of the system size.

\subsection{Local-gap condition}
\label{app:local_gap}

The argument for the local-gap condition is analogous to that for the spectral
gap in \cref{app:frust_full}. Let $A = B_r(I)$ for some site $I \in 
\Lambda^\text{ext}$ and $ r > 0$. Analogously as in \cref{eq:H_unp}, we split the
localized Hamiltonian $\tilde{H}_{0,A}$, defined by \cref{eq:H_local}, into
$\tilde{H}_{0,A} = H_{0,A}^\text{R} + H_{0,A}^\text{p}$. The
operators $H_{0,A}^\text{R}$ ($H_{0,A}^\text{p}$) contain all
terms associated to sites in $\Lambda$ ($\Lambda_\text{p}$). 

Let $L_{\text{Loop},A}$ denote the ground state configurations of
$H_{0,A}^\sub{R}$, as defined by \cref{eq:H_local}. Further define the
multiplicative group $\mathcal{U}$ generated by the operators $U_p$ with
support inside $A$, i.e, $\mathcal{U} := \langle U_p | p \in A\rangle$. Pick an
arbitrary $\vec{n} \in L_{\text{Loop},A}$ and define the set
$\mathcal{U}\ket{\vec{n}} := \{U\ket{\vec{n}}| U \in \mathcal{U}\}$. Then the
state
\begin{align}
    \label{eq:local_gs}
   \ket{\omega,\vec{n}}_A
   := \frac{1}{\sqrt{|\mathcal{U}\ket{\vec{n}}|}} \sum_{\ket{\vec{m}} 
   \in \mathcal{U}\ket{\vec{n}}} \ket{\vec{m}}
\end{align}
satisfies $U_p \ket{\omega,\vec{n}}_A = \ket{\omega,\vec{n}}_A$, because, by
construction, $U_p$ acts as a bijection on the set $\mathcal{U}\ket{\vec{n}}$.
Thus, $Q_I\ket{\omega,\vec{n}}_A = 0$ for $I \in \Lambda_\text{p}$. Furthermore
$\ket{\omega,\vec{n}}_A$ is a ground state of $H_{0,A}^\text{R}$. By
the discussion in \cref{app:frust_block}, $\ket{\omega,\vec{n}}_A$ is a ground
state of every operator $Q_I$ for $I \in \Lambda$.

Analogously as in \cref{app:frust_full}, it follows that $H_{0,A}$ has a spectral
gap of at least $\gamma = \min(\Delta,\omega)$, independent of the parameter $r$.
Hence, the local-gap condition is satisfied for a constant function
$\gamma(r)$.

\subsection{Local TQO}

\begin{figure}
    \centering
    \includegraphics[width=0.9\linewidth]{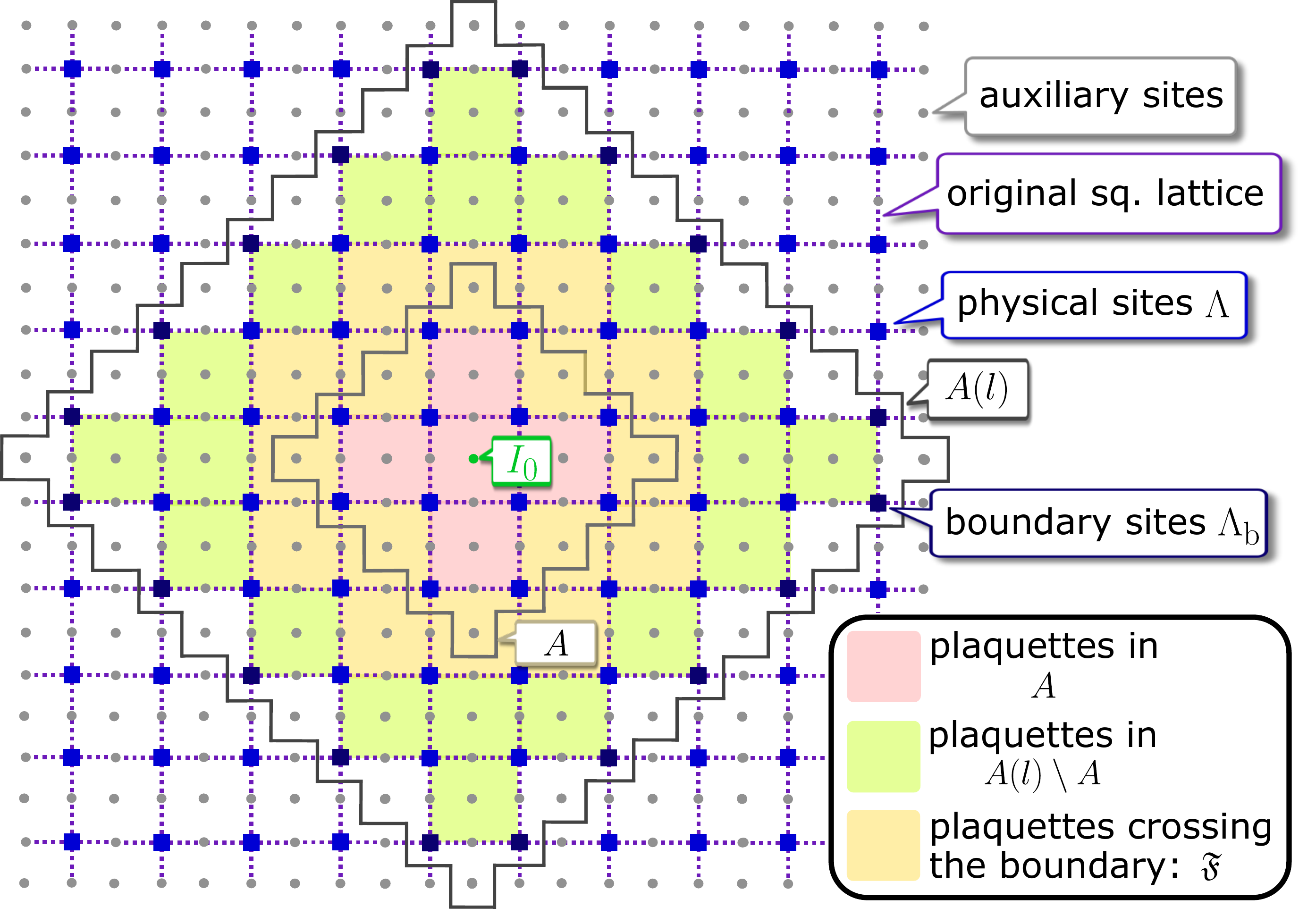}
    \caption{%
    \emph{Setup for local TQO.} Shown is the setup for the proof of the local
    TQO condition.  The blue squares and gray circles represent sites of the
    square lattice $\Lambda^\sub{ext}$, constructed to describe the structure
    $\C_\sub{Loop}$. The blue squares represent the physical sites, where the
    Hilbert spaces $\H_I$ are attached. The gray circles represent the
    auxiliary sites only used to attach the plaquette operators. The square
    lattice $\Lambda$, which only contains the physical sites, is depicted as
    dashed violet lines. The physical sites on the boundary of $A(l)$, denoted
    $\Lambda_\text{b}$, are highlighted in a darker blue color.  The green
    vertex $I_0$ is the midpoint of the region $A = B_r(I_0)$, in the depicted
    configuration $r = 4$. The boundary of the region $A$ is shown as a light
    gray line. The boundary of $A(l)$ is shown as a dark gray line, in the
    depicted configuration $l = 6$.  The plaquette operators $U_p$ with support
    lying in $A(l)$ are partitioned into three classes; in the figure this is
    shown by the color of the corresponding plaquette of the lattice $\Lambda$
    (Plaquettes fully in A: red, Plaquettes fully in $A(l)\setminus A$: green,
    Plaquettes crossing the boundary: orange). The distance $l$ is chosen
    sufficiently large such that all plaquettes that intersect $A$ are fully
    contained in $A(l)$.
    }
    \label{fig:loc_tqo}
 \end{figure}

In this section we verify the local TQO condition; in fact, we show that,
for sufficiently large $l$, the reduced density matrices as defined in
\cref{app:cond} are equal for all ground states. To this end, we need to
explicitly characterize the ground states of the local Hamiltonians
$H_{0,A(l)}$. The calculation is in some parts very similar to the
calculation that leads to the entanglement entropy for the toric code model; for
reference, see Ref.~\cite{Hamma_2005}. However, there are some differences, due to our
model having "internal degrees of freedom."

The local Hamiltonians $H_{0,A(l)}$ have a peculiar structure. For
two sites $I \sim J$ in $\Lambda$, with $I \in A$ and $J \notin A$, the
Hamiltonian $H_{0,I}$ is not part of $H_{0,A(l)}$. In this case, in
$V_I$ there are neither detunings nor blockades.  However, it can
happen that $H_{0,A(l)}$ contains an operator $U_p$ that acts
nontrivially on $I$. We denote these sites by $\Lambda_\text{b}$; they are
shown in \cref{fig:loc_tqo}.  We define $\Lambda_\text{i} := \Lambda \setminus
\Lambda_\text{b}$.

By construction, the ground state configurations of the sites in
$\Lambda_\text{i}$ can naturally be identified with loop configurations on the
(unextended) square lattice. In this picture, the symmetry operators $U_p$
behave in the same way as the operators $B_p$ in the toric code model. 

Let $\H_{\sub{Loop},A}$ denote the ground state subspace of $H_{0,A}^\sub{R}$ [cf.~
\cref{eq:H_unp}]. We note that the group $\mathcal{U}$, introduced in
\cref{app:local_gap}, acts freely. That is, for any nontrivial product
$U_{p_1}\cdots U_{p_n}$, it holds that $U_{p_1}\cdots U_{p_n}\ket{\vec{n}} \neq
\ket{\vec{n}}$ for all $\ket{\vec{n}} \in \H_{\sub{Loop},A}$. This follows as
any collection of plaquettes in $A(l)$ must have a boundary (as there are no
"half" plaquette operators). For two neighboring vertices $I,J$, this product 
swaps the atoms on the edge between $I$ and $J$, and
thus the action of the product acts nontrivially.

It immediately follows that, for every state $\ket{\vec{n}} \in
\H_{\sub{Loop},A}$, the states $\prod_{p\subseteq A(l)} U_p^{s_p} \ket{\vec{n}}$
are mutually orthonormal for all combinations of $s_p \in \{0,1\}$. Here the
notation $p \subseteq A(l)$ indicates that all physical sites of the support of
$U_p$ are contained in $A(l)$. Thus, the local ground state \ref{eq:local_gs}
has the form
\begin{subequations}
    \begin{align}
        \ket{\omega, \vec{n}}_{A(l)} &:= \mathcal{N}\sum_{\ket{\vec{m}} \in \mathcal{U}\ket{\vec{n}}}\ket{\vec{m}} \\
	    &= \mathcal{N}\sum_{s_p \in \{0,1\}}\prod_{p\subseteq A(l)} U_p^{s_p} \ket{\vec{n}}\\
        &= \mathcal{N}\prod_{p\subseteq A(l)} (\mathds{1} + U_p)\ket{\vec{n}}. \label{eq:gs_rewrite}
    \end{align}
\end{subequations}
The constant $\mathcal{N}$ denotes the normalization of the state. It is easy
to see that $\mathcal{N} = (\sqrt{2})^{-c_{A(l)}}$, where $c_{A(l)}$ denotes
the number of plaquettes that are part of $A(l)$. In the following, we
suppress the index $A(l)$ from the state~\eqref{eq:gs_rewrite}.

For any ground state of $H_{0,A(l)}^\sub{R}$, a ground state of $H_{0,A(l)}$
can be generated by \cref{eq:gs_rewrite}. However, if $\ket{\vec{m}} \in
\mathcal{U}\ket{\vec{n}}$ then both states generate the same ground state,
i.e., $\ket{\omega, \vec{n}} = \ket{\omega, \vec{m}}$. 

For convenience, we define $B := A(l)\setminus A$. As $\ket{\vec{n}}$ is a
product state, we can write it as a tensor product $\ket{\vec{n}} =
\ket{\vec{n}}_A \ket{\vec{n}}_{B}$, where $\ket{\vec{n}}_A$ (
$\ket{\vec{n}}_B$) denotes the configuration of atoms in $A$($B$). It is possible
to find a product of operators $U_p$ such that 
\begin{align}
    U_{p_1}\cdots U_{p_n}\ket{\vec{n}} = \ket{\vec{0}}_A \ket{\vec{n}'}_{B}.
\end{align}
This is due to the trivial homology group of $A$ with open boundary conditions.
We define the operator $X$ as a product of $\sigma_i^x$ on every atom in $B$,
where $\ket{\vec{n}'}_{B}$ differs from the state $\ket{\vec{0}}_{B}$. The
support of this operator is contained in $B$.  This implies that the states
$\ket{\vec{n}}$ and $X\ket{\vec{0}}$ generate the same ground state; hence,
\begin{align}
    \label{eq:gsX}
    \ket{\omega, \vec{n}} = \mathcal{N}\prod_{p\subseteq A(l)} (\mathds{1} + U_p)X\ket{\vec{0}}.
\end{align}
We partition the plaquettes contained in $A(l)$ into three sets: Plaquettes
contained in $A$, plaquettes contained in $B$, and plaquettes that cross the
boundary of $A$ and $B$. We denote the latter set by $\mathfrak{F}$. These
three sets of plaquettes are shown in \cref{fig:loc_tqo}. In the following we
assume that $l$ and thus $L_0$  are sufficiently large, such that all
plaquettes that intersect $A$ are contained in $A(l)$. This is the case in
\cref{fig:loc_tqo}. In this spirit, we define the states
\begin{subequations}
    \begin{align}
        \ket{\omega,\vec{0}}_A &:= \mathcal{N}_A\prod_{p\subseteq A} (\mathds{1} + U_p)\ket{\vec{0}}_A\\
        \ket{\omega,\vec{n}}_{B} &:= \mathcal{N}_{B} \prod_{p\subseteq B} (\mathds{1} + U_p)X\ket{\vec{0}}_{B}.
        \label{eq:gsB}
    \end{align}
\end{subequations}
The normalizations are given by $\mathcal{N}_A =
\displaystyle(\sqrt{2})^{-c_{A}}$ and $\mathcal{N}_{B} =
\displaystyle(\sqrt{2})^{-c_{B}}$, where $c_{A}$ and $c_{B}$ denote the numbers of
plaquettes contained in $A$ and $B$, respectively. Using these definitions,
\cref{eq:gsX} becomes
\begin{subequations}
    \begin{align}
        \ket{\omega, \vec{n}} 
        &= \widetilde{\mathcal{N}}\prod_{p \in \mathfrak{F}} (\mathds{1} + U_p) \ket{\omega,\vec{0}}_A \ket{\omega,\vec{n}}_{B}\\
	    &= \widetilde{\mathcal{N}}\sum_{s_p \in \{0,1\}}\prod_{p \in \mathfrak{F}} U_p^{s_p} \ket{\omega,\vec{0}}_A \ket{\omega,\vec{n}}_{B}.\label{eq:all_gs}
    \end{align}
\end{subequations}
Here the normalization is given by $\widetilde{\mathcal{N}} = (\sqrt{2})^{c_{A} +
c_{B} - c_{A(l)}} = (\sqrt{2})^{-c_\sub{AB}}$, where $c_\sub{AB} =
|\mathfrak{F}|$ denotes the number of plaquettes in $\mathfrak{F}$.

We define the vector $\vec{s} := (s_p)_{p \in \mathfrak{F}} \in
\Z_2^{c_\sub{AB}}$ and use the shorthand notation $U_{\vec{s}} := \prod\limits_{p
\in \mathfrak{F}} U_p^{s_p}$.  By construction, all products $U_{\vec{s}}$ are
induced by an automorphism $\phi \in \A_\sub{Loop}$. By \cref{sec:localz2},
these automorphisms are compositions of edge permutations and vertex
permutations [cf.~\cref{eq:loopaut}]. As the support of these permutations is
either contained in $A$ or $B$ and all these permutations commute, we can
partition them as $\phi = \phi_A \circ \phi_B$, where $\phi_A$ ($\phi_B$) only
acts on atoms in $A$ ($B$). The permutation $\phi_A$ defines an operator
$U_{\vec{s}}^A := U_{\phi_A}$ (cf.~\cref{sec:sym}) with support in $A$; the
permutation $\phi_B$ defines an operator $U_{\vec{s}}^B := U_{\phi_B}$ with
support in $B$. We want to emphasize that the permutations $\phi_A$ and
$\phi_B$ themselves are not automorphisms of $\C_\sub{Loop}$, and thus
$U_{\vec{s}}^A$ and $U_{\vec{s}}^B$ are not symmetries of the Hamiltonian.
However, because all permutations defining $\phi_A$ and $\phi_B$ square to the
identity, the same is true for the operators $U_{\vec{s}}^A$ and
$U_{\vec{s}}^B$, i.e., $(U_{\vec{s}}^A)^2 = \mathds{1}$ and $(U_{\vec{s}}^B)^2 = \mathds{1}$.
Furthermore, $U_{\vec{s}}^A$ and $U_{\vec{s}}^B$ are unitary, as they act as
permutations on the Hilbert space $\H$. We obtain the decomposition
\begin{align}
    \label{eq:U_decomp}
    U_{\vec{s}}^A \otimes U_{\vec{s}}^{B} = U_{\vec{s}} = \prod_{p \in \mathfrak{F}} U_p^{s_p}.
\end{align}
Applying this to \cref{eq:all_gs} yields
\begin{align}
    \ket{\omega, \vec{n}} = \widetilde{\mathcal{N}} \sum_{\vec{s} \in \Z_2^{c_\sub{AB}}} U_{\vec{s}}^A\ket{\omega,\vec{0}}_A U_{\vec{s}}^{B}\ket{\omega,\vec{n}}_{B}.
    \label{eq:all_gs2}
\end{align}
To reach the desired result, we have to compute the partial trace of
\cref{eq:all_gs2}. To this end, we show that the states
$U_{\vec{s}}^{B}\ket{\omega,\vec{n}}_{B}$ are orthonormal.  As a preliminary
step we consider the product $U_{\vec{s}'}^{B} U_{\vec{s}}^{B}$ for
$\vec{s},\vec{s}' \in \Z_2^{c_\sub{AB}}$. The fact that $U_p^2 = \mathds{1}$
for all plaquettes $p$, implies that $U_{\vec{s}} U_{\vec{s}'} =
U_{\vec{s}\oplus\vec{s}'}$, where the operation $\vec{s}\oplus\vec{s}'$ denotes
addition (modulo $2$) in $\Z_2^{c_\sub{AB}}$. Decomposing these operators as in
\cref{eq:U_decomp} yields
\begin{subequations}
    \begin{align}
        U_{\vec{s}}^A U_{\vec{s}'}^A \otimes U_{\vec{s}}^{B} U_{\vec{s}'}^{B} 
        &= U_{\vec{s}} U_{\vec{s}'}\\
        &= U_{\vec{s}\oplus\vec{s}'}\\
        &= U_{\vec{s}\oplus\vec{s}'}^A \otimes U_{\vec{s}\oplus\vec{s}'}^B.
    \end{align} 
\end{subequations}
From this equality, it follows that $U_{\vec{s}}^A U_{\vec{s}'}^A = \lambda
U_{\vec{s}\oplus\vec{s}'}^A$ and $U_{\vec{s}}^{B} U_{\vec{s}'}^{B} =
\lambda^{-1} U_{\vec{s}\oplus\vec{s}'}^B$, for $\lambda \in \mathbb{C}$. The 
fact that the operators act as permutations on the basis of excitation patterns
further implies that $\lambda = 1$, which shows that 
\begin{subequations}
    \begin{align}
        U_{\vec{s}}^A U_{\vec{s}'}^A &=  U_{\vec{s}\oplus\vec{s}'}^A\\
        U_{\vec{s}}^{B} U_{\vec{s}'}^{B} &=  U_{\vec{s}\oplus\vec{s}'}^B.
    \end{align}
\end{subequations}
Now we consider the scalar product
\begin{subequations}
    \begin{align}
       S_{\vec{s},\vec{s}'} &:= \bra{\omega,\vec{n}}_{B} U_{\vec{s}'}^{B} U_{\vec{s}}^{B}\ket{\omega,\vec{n}}_{B}\\
       &= \bra{\omega,\vec{n}}_{B}  U_{\vec{s}'\oplus\vec{s}}^{B}\ket{\omega,\vec{n}}_{B}.
    \end{align}
\end{subequations}
We show that $S_{\vec{s},\vec{s}'} = 0$ for $\vec{s} \neq \vec{s}'$. For
notational convenience, define the multiplicative group of plaquette operators
in $B$, $\mathcal{U}^B := \langle U_p\, |\, p \subseteq B \rangle$. Analogous
to $\mathcal{U}$, this group acts freely on the states $\ket{\vec{n}}$; thus,
\cref{eq:gsB} can also be written as 
\begin{align}
    \ket{\omega,\vec{n}}_{B} = \mathcal{N}_B \sum_{\vec{m} \in \mathcal{U}^B \ket{\vec{n}}} \ket{\vec{m}}.
\end{align}
Hence, if $U_{\vec{s}\oplus\vec{s}}^B \ket{\vec{n}} \notin \mathcal{U}^B
\ket{\vec{n'}}$ for all $\vec{s} \neq \vec{s}'$, it follows that
$S_{\vec{s},\vec{s'}} = 0$. The remainder of this section is dedicated to
proving this fact.

Consider a site $I \in B$ with an adjacent site $J \in A$; then there are two
plaquette operators $U_{p_1},U_{p_2}$ whose support contains $I$ and $J$. In
particular, this implies that $p_1,p_2 \in \mathfrak{F}$. The vertices in $V_I$
that are adjacent to vertices in $V_J$ are either a pair of edge vertices or a
tetrahedron. We denote the set of these vertices by $\del_J V_I$.

Suppose that $\del_J V_I$ is the pair of edge vertices on edge $e$. Then
$U_{p_1}$ and $U_{p_2}$ both contain the edge permutation $\varphi_e$. If the
product $U_{\vec{s}}$ contains one of these plaquette operators then
$U_{\vec{s}}^B$ also contains this edge permutation once.  On the other hand,
$U_{p_1}$ and $U_{p_2}$ are the only plaquette operators containing the edge
permutation $\varphi_e$. Thus, no product of plaquette operators $U_p$ with $p
\subseteq B$ contains this edge permutation, which implies that $U_{\vec{s}}^B
\ket{\vec{n}} \notin \mathcal{U}^B \ket{\vec{n}}$.

Suppose that $\del_J V_I$ consists of the vertices in one tetrahedron $s$. Then
the operators $U_{p_1}$, $U_{p_2}$ contain two different vertex permutations
$\varphi_s^{\alpha_1}$ and $\varphi_s^{\alpha_2}$.  No other plaquette
operators contain these two vertex permutations. If a product $U_{\vec{s}}$
were to contain only one of the plaquette operators $U_{p_1}$ or $U_{p_2}$,
then $U_{\vec{s}}^B$ would contain only one of the aforementioned vertex
permutations; \Wlog, let this be $\varphi_s^{\alpha_1}$. A product of plaquette
operators $U_p$ with $p \subseteq B$ can either contain $\varphi_s^{\alpha_3}$
($\alpha_3 \neq \alpha_1,\alpha_2)$ or the identity on the tetrahedron $s$.
Thus $U_{\vec{s}}^B \ket{\vec{n}} \notin \mathcal{U}^B \ket{\vec{n}}$.

This argument shows that $U_{\vec{s}}^B \ket{\vec{n}} \in \mathcal{U}^B
\ket{\vec{n}}$ implies that, for any site $I \in B$ with an adjacent site $J \in
A$, both or none of the plaquette operators, whose support contains $I$ and
$J$, must be part of the product $U_{\vec{s}}$. In other words, $\vec{s} =
(0,\ldots,0)$ or $\vec{s} = (1,\ldots, 1)$. 

We examine the case $\vec{s} = (1,\ldots, 1)$ separately. The product
$U_{\vec{s}}$ acts as edge permutations in a closed loop that encompasses the
region $A$. Thus this loop is not the boundary of a collection of plaquettes in $B$
(this loop is not homologically trivial) and thus $U_{(1,\ldots, 1)}^B
\ket{\vec{n}} \notin \mathcal{U}^B \ket{\vec{n}}$. 

In summary, we have shown that 
\begin{align}
    \label{eq:orth}
    \bra{\omega,\vec{n}}_{B}  U_{\vec{s}'\oplus\vec{s}}^{B}\ket{\omega,\vec{n}}_{B} = 0,
\end{align}
iff $\vec{s}'\oplus\vec{s} \neq \vec{0}$, which is equivalent to $\vec{s}' \neq
\vec{s}$. In the other case the scalar product is $1$, due to
$U_{\vec{s}'\oplus\vec{s}}^{B}$ being unitary.

Thus, we can complete the family $U_{\vec{s}}^{B}\ket{\omega,\vec{0}}_{B}$ to an
orthonormal basis of $\H_{B} = \bigotimes_{I \in B} \H_I$. We use this basis to
compute the partial trace of the density matrix corresponding to
\cref{eq:all_gs2} with respect to $B$. This results in 
\begin{align}
    \label{eq:red_density}
    \rho(\omega, \vec{n})_{A} = \widetilde{\mathcal{N}}^2\sum_{\vec{s},\vec{s}' \in \Z_2^{c_\sub{AB}}}
        U_{\vec{s}}^A \ket{\omega,\vec{0}}_A \bra{\omega,\vec{0}}_A U_{\vec{s}'}^A
\end{align}
Note that after taking the partial trace, \cref{eq:red_density} does not
depend on the state $\ket{\vec{n}}$ any more. As all ground states can be
generated from some state $\ket{\vec{n}}$ by \cref{eq:gs_rewrite}, all these
ground states have the same reduced density matrix~\eqref{eq:red_density}, if 
$l$ is sufficiently large. Thus, \cref{eq:LocalTQO} is satisfied, where
$F$ can be taken as a step function.

\section{Derivation of the effective Hamiltonian}
\label{app:Inv_Schriffer_Wolff}

In this section, we derive the properties of the low-energy effective
Hamiltonian given in \cref{eq:eff}. We use the method described by Datta
\etal~\cite{Datta1996}, which yields a well-defined effective Hamiltonian even
for extensive system sizes. However, we note that the results of this
section also hold for other approaches to perturbatively derive an effective
Hamiltonian (see, e.g.,~\cite[Section\ 3]{Bravyi_2011}).

Let $V_\text{Loop}$ and $E_\text{Loop}$ be the vertex set and edge set of the
blockade structure $\C_\text{Loop}$. Then the system Hamiltonian
$H_\text{Loop}$ is given by [cf.~\cref{eq:H}]
\begin{equation}
    \begin{aligned}
        \label{eq:HZ2}
        H_\text{Loop}(\Omega) &:= \underbrace{-\sum_{i \in V_\text{Loop}} \Delta_i n_i 
                    + U \sum_{\{i,j\} \in E_\text{Loop}} n_i n_j}_{=: H_0}\\
                    &\quad+ \underbrace{\Omega \sum_{i \in V_\text{Loop}} \sigma_i^x}_{=: V}.
    \end{aligned}
\end{equation}
This Hamiltonian describes a blockade interaction with a finite (but arbitrarily large)
blockade strength $U$. This is necessary to keep the tensor product
structure of the Hilbert space. 

For the perturbation theory, $H_0$ assumes the role of the classical
unperturbed Hamiltonian, whereas $V$ is the perturbation. 

\subsection{Unitary block-diagonalization by Datta \etal \cite{Datta1996}}
\label{app:ubd}

The method of Datta \etal \cite{Datta1996} provides a method to perturbatively
\emph{block-diagonalize} a Hamiltonian of the form $H = H_0 + \lambda V$. Here
$H_0$ is the unperturbed Hamiltonian  with exactly known spectrum, $V$ is the
perturbation and $\lambda$ is the perturbation strength. 

The blocks are defined with respect to the unperturbed Hamiltonian $H_0$. Let
$P_0$ be the projection onto the ground state manifold of $H_0$, and let $Q_0 :=
\mathds{1} - P_0$. Then the goal is to find a unitary transformation $e^{S}$,
with anti-Hermitian generator $S$, such that $H' := e^{S}He^{-S}$ is
block-diagonal in the sense that $P_0H'Q_0 = Q_0H'P_0 = 0$. The restriction of
this transformed Hamiltonian onto the ground state manifold, i.e., $P_0H'P_0$,
is the \emph{low-energy effective Hamiltonian} for $H$. The method of Datta
yields a series expansion of $S$ and $H'$, in powers of $\lambda$, such that
$H'$ is block-diagonal up to some order $n$. 

In this section we briefly describe the setup and conditions for the method of
Datta \etal \cite{Datta1996} to be applicable. The system is defined on  a finite subset of a square
lattice $\Lambda \subseteq \Z^2$.  To each vertex  $I \in \Lambda$, a copy of
some finite-dimensional Hilbert space $\mathcal{H}_I \simeq \mathcal{H}$ is
attached. The Hilbert space describing the whole system is the tensor
product $\mathcal{H}_\Lambda = \bigotimes_{I\in\Lambda} \mathcal{H}_I$. 

The unperturbed Hamiltonian is required to have the form
\begin{align}
   \label{eq:local_sum}
   H_0 = \sum_{X \subset \Lambda} H_{0,X},
\end{align}
where $H_{0,X}$ acts trivially on the Hilbert spaces attached outside of $X$.
In addition, it is required that there is a tensor product basis of
$\H_\Lambda$ such that every term $H_{0,X}$ is diagonal in this basis and the
ground state of $H_0$ is a ground state for every term $H_{0,X}$. We refer to
the latter condition as \emph{frustration-freeness}. That our Hamiltonian~\eqref{eq:HZ2}
satisfies these conditions (after the Hilbert space is properly partitioned) 
is already covered by the proof of the frustration-freeness in \cref{app:loc_frust}.
The desired tensor product basis is the basis of excitation patterns.

The perturbation is required to have the form 
\begin{align}
    V = \sum_{X \subseteq \Lambda} V_{X}, 
\end{align}
where the local terms $V_{X}$ act trivially on Hilbert spaces attached outside
of $X$. Additionally $\|V_{X}\|$ has to decay sufficiently fast
with the size of $X$. 

For our perturbation [cf. \cref{eq:HZ2}], we set 
\begin{align}
   V_{I} := \Omega \sum_{i \in I}\sigma_i^x
\end{align}
and $V_{X} = 0$ otherwise. As the perturbation only acts on individual sites,
it trivially satisfies the decay requirements. Hence, we can apply the
block-diagonalization methods of Datta \etal \cite{Datta1996} to our system.

In the remainder of this section, we briefly sketch the construction of the
effective Hamiltonian.  We start with some preliminary definitions.

For any set $X \subseteq \Lambda$, we define $B_X := \cup_{I \in X} B_1(I)$,
here $B_1(I)$ denotes the ball with radius $1$ in the graph metric on $\Lambda$
(cf.~\cref{app:graphs}).  Furthermore, define the projector $P_{B_X}^0$ as the
projection onto the subspace of states that are ground states of $H_{0,Y}$
for all $Y \subseteq X$. Moreover, define
\begin{align}
   P_{B_X}^1 
   := P_{B_X\setminus X}^0 - P_{B_X}^0,
\end{align}
the projector onto the states that are ground states in $B_X\setminus X$ but
fail to be ground states in $X$. For any local operator $A_X$ whose support is
contained in $X$, define its off-diagonal part as
\begin{align}
   A_{B_X}^{01} 
   := P_{B_X}^1 A_X P_{B_X}^0 + P_{B_X}^0 A_X P_{B_X}^1.
\end{align}
For the generator, Datta \etal \cite{Datta1996} utilized the ansatz 
\begin{align}
    S := \sum_{j = 1}^n \lambda^j S_j ,
\end{align}
where $n$ is arbitrary but finite. They further required the terms $S_j$ to be a
sum of local operators $S_j = \sum_{X \subset \Lambda} S_{j,X}$. The
transformed Hamiltonian $H'$ then has the form 
\begin{align}
    \label{eq:transformed}
    H' = H_0 + \underbrace{\sum_{j = 1}^n \lambda^j ([S_j,H_0] + V_j)}_{\overset{!}{=}\text{block-diagonal}} 
    + \underbrace{\sum_{j \geq n+1} \lambda^{j} V_j.}_{\text{higher-order remainder}}
\end{align}
The operators $V_j$ are a sum of local operators $V_j = \sum_{Y \subset
\Lambda} V_{j,Y}$, and the local terms are given by $V_{0,B_X} := V_{B_X}$ for $j = 0$ and by
\begin{widetext}
    \begin{align}
       \label{eq:def_V}
        V_{j,Y}&:=
        \begin{aligned}[t]
           &-\sum_{\substack{p\geq 2\\1\leq k_1,\ldots,k_p\leq n\\k_1+\ldots+k_p = j}} \frac{1}{p!}
           \sum_{\substack{\{X_1,\ldots,X_p\}_c\\Y = B_{X_1}\cup\cdots\cup B_{X_p}}} 
           [S_{k_p,B_{X_p}},\ldots,[S_{k_2,B_{X_2}},V_{k_1,B_{X_1}}^{01}]]\\
           &+\sum_{\substack{p\geq 1\\1\leq k_1,\ldots,k_p\leq n\\k_1+\ldots+k_p = j-1}} \frac{1}{p!} 
           \sum_{\substack{\{X,X_1,\ldots,X_p\}_c\\Y = B_{X}\cup B_{X_1}\cup\cdots\cup B_{X_p}}} 
           [S_{k_p,B_{X_p}},\ldots,[S_{k_1,B_{X_1}},V_{B_X}]].
       \end{aligned}
    \end{align}
\end{widetext}
for $j \geq 1$.
The notation $\{X_1,\ldots,X_p\}_c$ indicates that the summation is performed over
subsets of $\Lambda$ satisfying the conditions
\begin{align}
    \begin{aligned}
        X_2 \cap B_{X_1} &\neq \emptyset,\\
        X_3 \cap B_{X_1 \cup X_2} &\neq \emptyset,\\
        &\vdots\\
        X_p \cap B_{x_1\cup\cdots\cup X_p} &\neq \emptyset.
    \end{aligned}
\end{align}
This condition result from the commutators in \cref{eq:def_V} vanishing
otherwise. 

Requiring that the off-diagonal part in the first $n$ summands of
\cref{eq:transformed} vanishes yields a formula for $S$:
\begin{align}
   \label{eq:def_S}
   S_{j,B_X} 
   := \sum_{\substack{E,k, E', k'\\E \neq E'}} \frac{\bra{E, k} 
   V_{j,B_X}^{01}\ket{E', k'}}{E - E'}\ket{E,k}\bra{E',k'}.
\end{align}
The states $\ket{E,k}$ form an eigenbasis of the Hamiltonian
\begin{align}
   \overline{H}_{0,X} = \sum_{Y: Y\cap X \neq \emptyset} H_{0,Y}.
\end{align}
By construction, keeping only the first $n$ orders in \cref{eq:transformed}
yields a block-diagonal operator. The projection of this operator onto the
ground state manifold of $H_0$ is our low-energy effective Hamiltonian. It has
the form 
\begin{align}
    \label{eq:Heff}
   H_\text{eff} 
   := P_0 \sum_{j\geq 0}^n \lambda^j \sum_{X\subseteq \Lambda} V_{j,B_X} P_0.
\end{align} 
We note, that it is a priori unclear if neglecting the higher-order
off-diagonal terms and the diagonal terms in the high-energy manifold is a
good approximation. In particular, it is unclear if the spectral properties of
$H_\text{eff}$ carry over to $H$.

\subsection{Invariance of the effective Hamiltonian}
\label{app:Inv_Heff}

A local symmetry of $H_0$ ($V$) is a unitary operator $U$ that satisfies
$U^\dagger H_{0,X} U = H_{0,X}$ ($U^\dagger V_X U = V_X$) for all $X \subset
\Lambda$. In this section we show that any (local) symmetry of $H_0$ and $V$
is also a local symmetry of the constructed effective Hamiltonian. 

The local symmetry immediately implies that $[P^0_{B_X},U] = 0$ and
$[P^1_{B_X},U] = 0$; thus, if $U^\dagger A_X U = A_X$, then $U^\dagger A^{01}_{B_X} U
= A^{01}_{B_X}$ also. Furthermore, it implies that $\overline{H}_{0,X}$ is invariant under
$U$. Thus, $U\ket{E,k}$ is an eigenstate of $\overline{H}_{0,X}$ with the same energy $E$.

Because of the recursive definition of the operators $V_j$, it is natural to
prove their invariance under $U$ by induction. The base case is given by $V_0 =
V$, which is invariant under $U$ by definition.  Thus, assume that the
invariance of $S_{k,B_X}$ and $V_{k,B_X}$  under $U$ is already proven for $k <
j$. Note that the expression~\eqref{eq:def_V} only involves operators $V_{k}$
and $S_k$ with $k < j$. To calculate $U^\dagger V_{j,B_X} U$, we repeatedly use
the identity $U^\dagger[A,B]U = [U^\dagger A U, U^\dagger B U]$ that holds for
arbitrary operators $A$, $B$. As all operators inside the commutators are
invariant under $U$, it follows that $V_{j,B_X}$ is invariant under $U$.

Using the fact that $U$ acts block-diagonal with respect to the index $E$
implies the invariance of $S_{j,B_X}$:
\begin{widetext}
    \begin{subequations}
        \begin{align}
            U^\dagger S_{j,B_X} U 
            &= \sum_{\substack{E,k, E', k'\\E \neq E'}} 
            \frac{\bra{E, k} V_{j,B_X}^{01}\ket{E', k'}}{E - E'}U^\dagger\ket{E,k}\bra{E',k'}U\\
            &= \sum_{\substack{E,k, E', k'\\E \neq E'}} 
            \frac{\bra{E, k}U^\dagger V_{j,B_X}^{01}U\ket{E', k'}}{E - E'}\ket{E,k}\bra{E',k'}\\
            &= \sum_{\substack{E,k, E', k'\\E \neq E'}} 
            \frac{\bra{E, k} V_{j,B_X}^{01}\ket{E', k'}}{E - E'}\ket{E,k}\bra{E',k'} = S_{j,B_X}.
        \end{align}
    \end{subequations}
\end{widetext}
This completes the induction. It follows immediately that \cref{eq:Heff} is
invariant under $U$, as desired.

\subsection{Diagonal elements}

For a fully-symmetric blockade structure, the symmetries heavily constrain the
form of the effective Hamiltonian.  We first consider the diagonal elements.
Let $\vec{n},\vec{m} \in L_\sub{Loop}$. As $\C_\sub{Loop}$ is fully-symmetric,
there exists an automorphism $\phi \in \A_\sub{Loop}$ such that $\vec{n} =
\phi\cdot\vec{m}$.  Using the invariance of the effective Hamiltonian under the
symmetry operators $U_\phi$ yields
\begin{subequations}
    \begin{align}
        \bra{\vec{n}}H_\text{eff}\ket{\vec{n}} 
        &= \bra{\phi\cdot\vec{m}}H_\text{eff}\ket{\phi\cdot\vec{m}}\\
        &= \bra{\vec{m}}U_\phi^\dagger H_\text{eff}U_\phi\ket{\vec{m}}\\
        &= \bra{\vec{m}} H_\text{eff}\ket{\vec{m}},
    \end{align}
\end{subequations}
i.e.\ all diagonal elements are equal. Hence, by addition of a constant to the
effective Hamiltonian (a shift in the energy scale), the diagonal elements can
be set to $0$. This is consistent with our result in \cref{app:proof_1}, where
we showed that all states from $\H_\sub{Loop}$ enter with equal weight into
the ground state.

\subsection{Off-diagonal elements}

In this section, we derive the off-diagonal part of the effective
Hamiltonian to leading order, i.e., we consider the smallest $n$ such that
$H_\text{eff}$ has nonzero off-diagonal terms.

As our perturbation $V$ is a sum of operators $\sigma^x$, it can flip the state
of \emph{exactly one} atom. More precisely, this means
that the matrix element satisfies $\bra{\vec{n}'} V \ket{\vec{n}} \neq 0$, only if
$\vec{n}$ and $\vec{n}'$ differ in at most one component. 

We claim that $\bra{\vec{n}'} V_j \ket{\vec{n}} \neq 0$ only if $\vec{n}$ and
$\vec{n}'$ differ in at most $j$ components. As it is straightforward, we only
sketch the proof by induction. The previous paragraph contains the base case;
assume that this is true for $V_k$ and $S_k$ for $k < j$. In \cref{eq:def_V}, expand
all the commutators and insert identities of the form $\sum_{\vec{n}}
\ket{\vec{n}}\bra{\vec{n}}$ between all products. Then the matrix element
$\bra{\vec{n}'} V_j \ket{\vec{n}}$ can only be nonzero, if the sum contains at 
least one product of matrix elements that is nonzero. Each
of the operators $S_{k_l}$ can change at most $k_l$ components of the vector
$\vec{n}$. As $k_1+\cdots+k_p = j$ in the first sum and $k_1+\cdots+k_p+1 = j$
in the second sum, it follows that $V_j$ can maximally change $j$ components in
the vector $\vec{n}$. As the matrix elements of $S_j$ are only rescaled matrix 
elements of $V_j$, the same is true for $S_j$. 

The number of components that have to be changed to transform the configuration
$\vec{n}$ into the configuration $\vec{n'}$ is known as the \emph{Hamming distance}.
The previous section shows that $\bra{\vec{n}'} H_\text{eff} \ket{\vec{n}}
\neq 0$ requires that the order of perturbation theory $n$ is at least the
Hamming distance between $\vec{n}$ and $\vec{n}'$.

By construction two configurations $\vec{n},\vec{n}' \in L_\sub{Loop}$, differ
in edge sites along a closed loop. The smallest closed loop consists of a single
plaquette in $\C_\sub{Loop}$ (assuming that the lattice is sufficiently large).
Thus, if  $\vec{n}$ and $\vec{n}'$ have the smallest possible Hamming distance,
they satisfy $\ket{\vec{n}'} = U_p \ket{\vec{n}}$. We conclude that, to leading
order, the only nonzero matrix elements of $H_\text{eff}$ are $\bra{\vec{n}}
H_\text{eff} U_p \ket{\vec{n}}$ for $\vec{n} \in L_\sub{Loop}$.

The structure $\C_\sub{Loop}$ being fully-symmetric implies that, for
$\vec{n},\vec{m} \in L_\sub{Loop}$, there exists an automorphism $\phi \in
\A_\sub{Loop}$ such that $\vec{n} = \phi \cdot \vec{m}$. As $\A_\sub{Loop}$ is
abelian, we obtain
\begin{subequations}
    \begin{align}
        \bra{\vec{n}}H_\text{eff}U_p\ket{\vec{n}} 
        &= \bra{\phi\cdot\vec{m}}H_\text{eff}U_p \ket{\phi\cdot\vec{m}}\\
        &= \bra{\vec{m}}U_\phi^\dagger H_\text{eff} U_p U_\phi\ket{\vec{m}}\\
        &= \bra{\vec{m}} H_\text{eff} U_p \ket{\vec{m}}.
    \end{align}
\end{subequations}
Thus, the leading order off-diagonal matrix elements of the effective
Hamiltonian are equal to some constant $C_p := \bra{\vec{m}} H_\text{eff}
U_p \ket{\vec{m}}$ that can only depend on the plaquette $p$. However, for
periodic boundary conditions, the Hamiltonian $H_\sub{Loop}$ is invariant under
translations. This implies that $C :\equiv C_p$ for all plaquettes $p$.
Therefore, the leading order part of the effective Hamiltonian has the form
\begin{subequations}
    \begin{align}
        H_\text{eff,lo} 
	    &= \sum_{\vec{n},\vec{m} \in L_\text{Loop}}
        \bra{\vec{n}}H_\text{eff}\ket{\vec{m}} \ket{\vec{n}} \bra{\vec{m}} \\
	    &= \sum_{\vec{n} \in L_\text{Loop}}\sum_{\text{Faces }p} 
        \bra{\vec{n}}H_\text{eff}U_p\ket{\vec{n}} \ket{\vec{n}} \bra{\vec{n}} U_p\\
	    &= C \sum_{\vec{n} \in L_\text{Loop}}\sum_{\text{Faces }p} \ket{\vec{n}} \bra{\vec{n}} U_p\\
        &= C \sum_{\text{Faces }p} U_p.
    \end{align}
\end{subequations}
Lastly, we comment on the properties of $C$. Let $K$ denote the Hamming
distance between the states $U_p\ket{\vec{n}}$ and $U_p\ket{\vec{n}'}$, i.e., the
leading order. By the definition~\eqref{eq:Heff} we obtain $C =
\mathcal{O}(\Omega^K)$. The dependence of $C$ on $\Delta E$, $C =
\mathcal{O}(1/\Delta E^{K-1})$, is less trivial. It can be proven by induction,
as by the definition~\eqref{eq:def_S}, the matrix elements of $S$ are the matrix
elements of $V$, rescaled by constants of order $\mathcal{O}(1/\Delta E)$.
Counting orders of $S$ yields the exponent $K-1$. 

We were unable to derive the sign of $C$, as this would involve explicitly
calculating the effective Hamiltonian up to order $24$. However, we have two
reasons to believe that $C < 0$.  First, this is consistent with the fact that
the finite volume ground states are in the zero flux sector. Second, using
another method to derive the effective Hamiltonian (that is unfortunately not
valid in the thermodynamic limit), we were able to show that $C < 0$~\cite{Maier2023}.

From the leading order onwards, nonzero matrix elements are possible in every
even order.  Thus the next term in $H_\text{eff}$ has a coefficient of order
$\mathcal{O}(\Omega^{K+2}/\Delta E^{K+1})$.

\section{Unit ball embeddings}
\label{app:embedding}

Here we provide exact coordinates for the unit ball embeddings of the
\texttt{FSU}-structure and its extension by links depicted in
\cref{fig:embedding} (b) and (d) of the main text. The units are given in
blockade radii so that $\Rb=1$.

Positions of the vertices in the unit ball embedding of the
\texttt{FSU}-structure shown in \cref{fig:embedding}~(b):
\begin{center}
\begin{tabular}{c|lr@{.}l r@{.}l r@{.}lr}
    Vertex label & \multicolumn{8}{c}{Vertex position}\\
    \hline
    Tetrahedron 1 & \ \texttt{(}&\texttt{-0} & \texttt{45455 }& \texttt{-0} & \texttt{26243} & \texttt{0} & \texttt{0} & \texttt{)} \\
    Tetrahedron 2 & \ \texttt{(}&\texttt{0} & \texttt{45455 }& \texttt{-0} & \texttt{26243} & \texttt{0} & \texttt{0} & \texttt{)} \\
    Tetrahedron 3 & \ \texttt{(}&\texttt{0} & \texttt{0} & \texttt{0} & \texttt{52486} & \texttt{0} & \texttt{0} & \texttt{)} \\
    Tetrahedron 4 & \ \texttt{(}&\texttt{0} & \texttt{0} & \texttt{0} & \texttt{0} & \texttt{-0} & \texttt{74227} & \texttt{)} \\
    $A$ & \ \texttt{(}&\texttt{0} & \texttt{0} & \texttt{-0} & \texttt{7873} & \texttt{0} & \texttt{37113} & \texttt{)} \\
    $B$ & \ \texttt{(}&\texttt{0} & \texttt{68182 }& \texttt{0} & \texttt{39365} & \texttt{0} & \texttt{37113} & \texttt{)} \\
    $C$ & \ \texttt{(}&\texttt{-0} & \texttt{68182 }& \texttt{0} & \texttt{39365} & \texttt{0} & \texttt{37113} & \texttt{)} \\
    $\overline{B}$ & \ \texttt{(}&\texttt{-0} & \texttt{68182 }& \texttt{-0} & \texttt{39365} & \texttt{-0} & \texttt{74227} & \texttt{)} \\
    $\overline{C}$ & \ \texttt{(}&\texttt{0} & \texttt{68182 }& \texttt{-0} & \texttt{39365} & \texttt{-0} & \texttt{74227} & \texttt{)} \\
    $\overline{A}$ & \ \texttt{(}&\texttt{0} & \texttt{0} & \texttt{0} & \texttt{7873} & \texttt{-0} & \texttt{74227} & \texttt{)} \\
\end{tabular}
\end{center}

Positions of the vertices in the unit ball embedding of the extended
\texttt{FSU}-structure shown in \cref{fig:embedding}~(d):
\begin{center}
\begin{tabular}{c|lr@{.}l r@{.}l r@{.}lr}
    Vertex label & \multicolumn{8}{c}{Vertex position}\\
    \hline
    Tetrahedron 1 & \ \texttt{(}&\texttt{-0} & \texttt{45455 }& \texttt{-0} & \texttt{26243} & \texttt{0} & \texttt{0} & \texttt{)} \\
    Tetrahedron 2 & \ \texttt{(}&\texttt{0} & \texttt{45455 }& \texttt{-0} & \texttt{26243} & \texttt{0} & \texttt{0} & \texttt{)} \\
    Tetrahedron 3 & \ \texttt{(}&\texttt{0} & \texttt{0} & \texttt{0} & \texttt{52486} & \texttt{0} & \texttt{0} & \texttt{)} \\
    Tetrahedron 4 & \ \texttt{(}&\texttt{0} & \texttt{0} & \texttt{0} & \texttt{0} & \texttt{-0} & \texttt{74227} & \texttt{)} \\
    Wing $\overline{A}$ & \ \texttt{(}&\texttt{0} & \texttt{0} & \texttt{-0} & \texttt{7873} & \texttt{0} & \texttt{37113} & \texttt{)} \\
    Wing $\overline{B}$ & \ \texttt{(}&\texttt{0} & \texttt{68182 }& \texttt{0} & \texttt{39365} & \texttt{0} & \texttt{37113} & \texttt{)} \\
    Wing $\overline{C}$ & \ \texttt{(}&\texttt{-0} & \texttt{68182 }& \texttt{0} & \texttt{39365} & \texttt{0} & \texttt{37113} & \texttt{)} \\
    Wing $B$ & \ \texttt{(}&\texttt{-0} & \texttt{68182 }& \texttt{-0} & \texttt{39365} & \texttt{-0} & \texttt{74227} & \texttt{)} \\
    Wing $C$ & \ \texttt{(}&\texttt{0} & \texttt{68182 }& \texttt{-0} & \texttt{39365} & \texttt{-0} & \texttt{74227} & \texttt{)} \\
    Wing $A$ & \ \texttt{(}&\texttt{0} & \texttt{0} & \texttt{0} & \texttt{7873} & \texttt{-0} & \texttt{74227} & \texttt{)} \\
    Bridge $\overline{A}$ & \ \texttt{(}&\texttt{0} & \texttt{72727 }& \texttt{-1} & \texttt{1547} & \texttt{0} & \texttt{63093} & \texttt{)} \\
    Bridge $A$ & \ \texttt{(}&\texttt{0} & \texttt{72727 }& \texttt{1} & \texttt{1547} & \texttt{-1} & \texttt{00206} & \texttt{)} \\
    $A$ & \ \texttt{(}&\texttt{1} & \texttt{45455 }& \texttt{-1} & \texttt{1547} & \texttt{0} & \texttt{63093} & \texttt{)} \\
    $\overline{A}$ & \ \texttt{(}&\texttt{1} & \texttt{45455 }& \texttt{1} & \texttt{1547} & \texttt{-1} & \texttt{00206} & \texttt{)} \\
    Bridge $\overline{B}$ & \ \texttt{(}&\texttt{0} & \texttt{63636 }& \texttt{1} & \texttt{20719} & \texttt{0} & \texttt{63093} & \texttt{)} \\
    Bridge $B$ & \ \texttt{(}&\texttt{-1} & \texttt{36364 }& \texttt{0} & \texttt{05249} & \texttt{-1} & \texttt{00206} & \texttt{)} \\
    $\overline{B}$ & \ \texttt{(}&\texttt{0} & \texttt{27273 }& \texttt{1} & \texttt{83702} & \texttt{0} & \texttt{63093} & \texttt{)} \\
    $B$ & \ \texttt{(}&\texttt{-1} & \texttt{72727 }& \texttt{0} & \texttt{68232} & \texttt{-1} & \texttt{00206} & \texttt{)} \\
    Bridge $\overline{C}$ & \ \texttt{(}&\texttt{-1} & \texttt{36364 }& \texttt{-0} & \texttt{05249} & \texttt{0} & \texttt{63093} & \texttt{)} \\
    Bridge $C$ & \ \texttt{(}&\texttt{0} & \texttt{63636 }& \texttt{-1} & \texttt{20719} & \texttt{-1} & \texttt{00206} & \texttt{)} \\
    $\overline{C}$ & \ \texttt{(}&\texttt{-1} & \texttt{72727 }& \texttt{-0} & \texttt{68232} & \texttt{0} & \texttt{63093} & \texttt{)} \\
    $C$ & \ \texttt{(}&\texttt{0} & \texttt{27273 }& \texttt{-1} & \texttt{83702} & \texttt{-1} & \texttt{00206} & \texttt{)} \\
\end{tabular} 
\end{center}


\end{document}